\documentclass[11pt]{article}
\usepackage{enumerate,amsmath,amsfonts,amssymb,graphicx,dsfont}
\usepackage[english]{babel}
\usepackage[latin1]{inputenc}
\usepackage{latexsym}
\usepackage{graphicx}
\usepackage{listings}
\usepackage{subfigure}
\usepackage{pdfpages}
\usepackage{bm}
\usepackage{amsthm}
\usepackage{hyperref}
\usepackage[margin=1in]{geometry}    
\usepackage{abstract}            		
\usepackage{titlesec}
\usepackage{float}

\usepackage[section]{placeins} 

\newtheorem{theorem}{Theorem} 
\newtheorem{lemma}{Lemma} 

\usepackage[round]{natbib} 

\allowdisplaybreaks 

\title{\textbf{
				Estimation and inference of domain means subject to shape constraints}}

\author{
	\large
	\textsc{Cristian Oliva\thanks{coliva@colostate.edu} , Mary C. Meyer\thanks{meyer@stat.colostate.edu} and Jean D. Opsomer\thanks{jopsomer@stat.colostate.edu}} \\ [2mm] 
	\normalsize Department of Statistics, Colorado State University, Fort Collins, Colorado, USA, 80523 \\ 
}

\date{ \vspace{3mm} Updated version by \today}

\begin{document}

\maketitle

\begin{abstract}

	\noindent Population domain means are frequently expected to respect shape or order constraints that arise naturally with survey data. For example, given a job category, mean salaries in big cities might be expected to be higher than those in small cities, but no order might be available to be imposed within big or small cities. A design-based estimator of domain means that imposes constraints on the most common survey estimators is proposed. Inequality restrictions that can be expressed with irreducible matrices are considered, as these cover a broad class of shapes and partial orderings. The constrained estimator is shown to be consistent and asymptotically normally distributed under mild conditions, given that the  shape is a reasonable assumption for the population. Further, simulation experiments demonstrate that both estimation and variability of domain means are improved by the constrained estimator in comparison with usual unconstrained estimators, especially for small domains. An application of the proposed estimator to the 2015 U.S. National Survey of College Graduates is shown. 
	
\end{abstract}

\section{Introduction}

Fine-scale domain estimates are frequently of interest for large-scale surveys, as these are highly useful for many data users in data-producing agencies. Although the overall sample size of such surveys might be very large, samples sizes for numerous domains are often too small for reliable estimates. For instance, the National Compensation Survey (\url{www.bls.gov/ncs}), produced by the U.S. Bureau of Labor Statistics, is designed to provide wage and salary estimates by occupation for many metropolitan areas and for the nation. However, for certain cities or regions, the sample sizes might not be large enough to produce estimates with acceptable precision.

Domain estimators that are based only on the domain-specific sample data (\textsl{direct estimators}) tend to lack adequate precision for small domains \citep{rao03}. One possible approach to avoid such a problem could be to aggregate small domains into bigger scales so that more reliable direct estimators can be produced for those scales, leading to the generation of more aggregated information than the actual desired scale. An alternative to producing small domain estimates could be changing from a design-based to a model-based estimation methodology such as small area models. In this paper, we present an approach where domains are allowed to borrow information from their neighbors by imposing shape or order assumptions that are reasonable for the population.

Information regarding the \textsl{shape} of population domain means arises naturally in surveys. For instance, certain jobs might be expected to receive better salaries than others, or younger people are expected to have, on average, lower glucose level than older people. However, given that small domains tend to produce direct estimates with high variability, such shape constraints are often violated at the sample level. Recently, \cite{wu16} proposed a domain mean estimation methodology that relies on the assumption of monotone population domain means. By combining the monotonicity information of domain means and design-based estimators in the estimation stage, they proposed a \textsl{constrained} estimator that respects the monotone assumption. Such an estimator was shown to improve precision and variability of domain mean estimates in comparison with direct estimators, given that the assumption of monotonicity is reasonable.

Many other types of shape constraints beyond monotonicity may also be expected to hold in estimates of population domain means. In general, any set of constraints can be represented through a \textsl{constraint matrix}, where each of its rows defines a constraint. \cite{meyer99} introduced the concept of \textsl{irreducible} matrices to cover the possible case of having more constraints than dimensions. Intuitively, a constraint matrix is called irreducible when it does not contain redundant restrictions. For illustration of a constraint matrix, suppose the variable of interest is the annual average salary of faculty in certain university. Further, consider the 6 domains generated from the cross-classification of the variables job position ($x_1$; 1=Assistant and 2=Associate) and department ($x_2$; 1=Anthropology, 2=English and 3=Engineering). Under the assumptions that, within a discipline, professors with an associate rank have higher salaries than those with an assistant rank; and that, within a rank, Engineering faculty members are expected to have higher salaries than those in either the Anthropology or English departments, then we can express the corresponding restrictions as,
\begin{equation} \label{eq:Aexample}
\boldsymbol{A\mu} \boldsymbol{\geq 0}, \; \; \text{ where } \; \;  \boldsymbol{A}=\left( \begin{array}{rrrrrr}
-1 & 1 & 0 & 0 & 0 & 0 \\
0 & 0 & -1 & 1 & 0 & 0 \\
0 & 0 & 0 & 0 & -1 & 1 \\
-1 & 0 & 0 & 0 & 1 & 0 \\
0 & 0 & -1 & 0 & 1 & 0 \\
0 & -1 & 0 & 0 & 0 & 1 \\
0 & 0 & 0 & -1 & 0 & 1 
\end{array} \right),
\end{equation} 
$\boldsymbol{\mu}=(\mu_{11},\mu_{21}, \mu_{12},\mu_{22},\mu_{13},\mu_{23} )^{\top}$, with $\mu_{ij}$ representing the mean of the domain that corresponds to $x_1=i$ and $x_2=j$; $\boldsymbol{0}$ being the zero vector, and the inequality being element-wise. In this example, the constraint matrix $\boldsymbol{A}$ is irreducible.

This paper contains theoretical properties and applications of a new constrained estimator for population domain means that respect shape constraints that are expressed with irreducible matrices. Through combining design-based domain mean estimators with these shape constraints, we propose a broadly applicable estimator that improves precision and variability of the most common direct estimators. Moreover, we provide a design-based variance estimation method that depends on the sample-determined linear space where the constrained estimator lands. If the constraints correspond to partial orderings, as in Equation \ref{eq:Aexample}, then the proposed estimator is simply a design-based estimator computed after adaptively pooling domains to respect the imposed restrictions, and the variance estimator depends on the pooling chosen by the constrained estimator. As monotone constraints can be written as one particular case from the broad class of shapes covered by irreducible matrices, our proposed estimator is an extension of the monotone estimator developed by \cite{wu16}. Constrained estimators that respect constraints driven by irreducible matrices have been already proposed for non-survey data. For instance, \cite{meyer13} made use of them to perform convex regression or isotonic regression on partial orderings. However, this general class of shape constraints have not been considered yet for survey data.

This paper is organized as follows: in Section \ref{sec:constrained} we introduce the constrained estimator and propose a linearization-based method for variance estimation. This section also contains some scenarios of interest where shape constraints can naturally arise for survey data. Section \ref{sec:properties} states the main theoretical properties of the constrained estimator that guarantee its use for estimation and inference of population domain means. The necessary assumptions used in these theoretical derivations are also stated in this section. Proofs of main theorems and auxiliary lemmas are fully contained in the Appendix. Section \ref{sec:simulations} shows through simulations that the constrained estimator improves domain mean estimation and variability in comparison with the unconstrained estimator, even though the assumed shape holds only approximately at the population level. Section \ref{sec:nscg} demonstrates the advantages of the proposed methodology on real survey data through an application to the 2015 National Survey of College Graduates. Lastly, Section \ref{sec:conclusions} contains some potential research directions related to the proposed constrained methodology for survey data. The proofs of the theoretical results shown in this paper are included in Appendix \ref{sec:appendix2}.

\section{Constrained estimator for domain means} \label{sec:constrained}

\subsection{Notation and preliminaries}

Let $U_N$ be the set of elements in a population of size $N$. Consider a sample $s_N$ of size $n_N$ that is drawn from $U_N$ using a probability sampling design $p_N(\cdot)$. Denote $\pi_{k,N}=\text{Pr}(k \in s_N)$ and $\pi_{kl,N}=\text{Pr}(k \in s_N, l \in s_N)$ as the first and second order inclusion probabilities, respectively. Assume that $\pi_{k,N}>0, \pi_{kl,N}>0$ for $k,l\in U_N$. Denote $\{U_{d,N}\}_{d=1}^D$ as a domain partition of $U_N$, where $D$ is the fixed number of domains and each $U_{d,N}$ is of size $N_d$. Also, let $s_{d,N}$ be the subset of size $n_{d,N}$ of $s_N$ that belongs to $U_{d,N}$.

For any study variable $y$, denote $\boldsymbol{\overline{y}}_{U_N}=(\overline{y}_{U_{1,N}}, \dots, \overline{y}_{U_{D,N}})^{\top}$ to be the vector of population domain means, where
\begin{equation} \label{eq:popdomainmean}
\overline{y}_{U_{d,N}}=\frac{\sum_{k \in U_{d,N}}y_k}{N_d}.
\end{equation}
In addition, consider the Horvitz-Thompson (HT) and H\'ajek estimators of $\overline{y}_{U_{d,N}}$, respectively given by
\begin{equation} \label{eq:sampledomainmean}
\widehat{y}_{s_{d,N}}=\frac{\sum_{k \in s_{d,N}}y_k/\pi_k}{N_d}, \; \; \; \tilde{y}_{s_{d,N}}=\frac{\sum_{k \in s_{d,N}}y_k/\pi_k}{\widehat{N}_d}; 
\end{equation}
where $\widehat{N}_d=\sum_{k \in s_{d,N}} 1/\pi_k$. Denote $\boldsymbol{\widehat{y}}_{s_N}$ and $\boldsymbol{\tilde{y}}_{s_N}$ to be the vectors of HT and H\'ajek estimators, respectively. Taking into consideration that the H\'ajek estimator is more useful in practice since it does not require information about the population domain sizes $N_d$, then we exclusively focus this paper on properties based on it. However, all developed results can be adapted to the HT estimator by replacing $\widehat{N}_d$ with $N_d$. For simplicity in our notation, we will avoid using the subscript $N$ for the rest of this paper unless needed for clarification.

\subsection{Proposed estimator}

Assume there is information available regarding the shape of the population domain means that can be expressed with $m$ constraints through a $m \times D$ irreducible constraint matrix $\boldsymbol{A}$. A matrix $\boldsymbol{A}$ is irreducible if none of its rows is a positive linear combination of other rows, and if the origin is also not a positive linear combination of its rows \citep{meyer99}. To take advantage of $\boldsymbol{\tilde{y}}_s$ to obtain an estimator that respects these shape constraints, we propose the constrained estimator $\boldsymbol{\tilde{\theta}}_s=(\tilde{\theta}_{s_1}, \dots, \tilde{\theta}_{s_D})^{\top}$ to be the unique vector that solves the following constrained weighted least squares problem,
\begin{equation} \label{eq:opt}
\underset{\boldsymbol{\theta}}{\min} (\boldsymbol{\tilde{y}}_s-\boldsymbol{\theta})^{\top}\boldsymbol{W}_s (\boldsymbol{\tilde{y}}_s-\boldsymbol{\theta}) \; \; \text{ subject to } \; \; \boldsymbol{A\theta \geq 0};
\end{equation}
where $\boldsymbol{W}_s$ is the diagonal matrix with elements $\widehat{N}_1/\widehat{N}, \widehat{N}_2/\widehat{N}, \dots, \widehat{N}_D/\widehat{N}$, and $\widehat{N}=\sum_{d=1}^D\widehat{N}_d$. The constrained problem in Equation \ref{eq:opt} can be alternatively written as finding the unique vector $\boldsymbol{\tilde{\phi}}_s$ that solves
\begin{equation} \label{eq:opt2}
\underset{\boldsymbol{\phi}}{ \min}  ||\boldsymbol{\tilde{z}}_s-\boldsymbol{\phi}||^2 \; \; \text{ subject to } \; \; \boldsymbol{A}_s \boldsymbol{\phi \geq 0};
\end{equation}
where $\boldsymbol{\tilde{z}}_s=\boldsymbol{W}_s^{1/2}\boldsymbol{\tilde{y}}_s$, $\boldsymbol{\phi}=\boldsymbol{W}_s^{1/2}\boldsymbol{\theta}$, and $\boldsymbol{A}_s=\boldsymbol{A}\boldsymbol{W}_s^{-1/2}$. Note that solving the optimization problem in Equation \ref{eq:opt2} allows straightforward computation of the constrained estimator $\boldsymbol{\tilde{\theta}}_s$. Moreover, observe that the transformed constrained matrix $\boldsymbol{A}_s$ is also irreducible if $\boldsymbol{A}$ is, and that it depends on the sample although $\boldsymbol{A}$ does not.

From a geometrical viewpoint,  $\boldsymbol{\tilde{\phi}}_s$ can be seen as the projection of the vector $\boldsymbol{\tilde{z}}_s$ onto the constraint cone $\Omega_s$ defined by the irreducible matrix $\boldsymbol{A}_s$ as
\begin{equation} \label{eq:cone}
\Omega_s=\{\boldsymbol{\phi} \in \mathbb{R}^D: \boldsymbol{A}_s \boldsymbol{\phi \geq 0} \}.
\end{equation}
That is, $\boldsymbol{\tilde{\phi}}_s=\Pi(\boldsymbol{\tilde{z}}_s | \Omega_s)$, where $\Pi(\boldsymbol{u}|V)$ stands for the projection of $\boldsymbol{u}$ onto the space $V$. Further, the polar cone $\Omega_s^0$ \cite[p.~121]{rockafellar70}, which is the dual vector space of $\Omega_s$, is defined as 
\begin{equation} \label{eq:polarcone}
\Omega_s^0=\{ \boldsymbol{\rho} \in \mathbb{R}^D : \langle \boldsymbol{\rho}, \boldsymbol{\phi} \rangle  \leq 0, \; \; \forall \boldsymbol{\phi} \in \Omega_s  \},
\end{equation}
where $\langle \boldsymbol{u},\boldsymbol{v} \rangle =\boldsymbol{u}^{\top}\boldsymbol{v}$. Such a definition characterizes the polar cone as the set of vectors that form obtuse angles with all vectors in $\Omega_s$. \cite{meyer99} showed that the negative rows of an irreducible matrix are the \textsl{edges} (generators) of the polar cone, leading to the following characterization of the polar cone in Equation \ref{eq:polarcone}:
\begin{equation} \label{eq:polarconeedges}
\Omega_s^0=\{\boldsymbol{\rho} \in \mathbb{R}^D : \boldsymbol{\rho}=\sum_{j=1}^{m}a_j\boldsymbol{\gamma}_{s_j}, \; \; a_j \geq 0, \; \;  j=1,2,\dots,m \},
\end{equation}
where $\boldsymbol{\gamma}_{s_1},\boldsymbol{\gamma}_{s_2}, \dots, \boldsymbol{\gamma}_{s_m}$ are the rows of $-\boldsymbol{A}_s$. Equation \ref{eq:polarconeedges} shows that $\Omega_s^0$ is a finitely generated cone, which implies that it is a \textsl{polyhedral} cone. \citet[p.~17]{robertson88} established necessary and sufficient conditions for a vector $\boldsymbol{\tilde{\phi}}_s$ to be the projection of $\boldsymbol{\tilde{z}}_s$ onto $\Omega_s$. That is, $\boldsymbol{\tilde{\phi}}_s \in \Omega_s$ solves the constrained problem in Equation \ref{eq:opt2} if and only if
\begin{equation*}
\langle\boldsymbol{\tilde{z}}_s-\boldsymbol{\tilde{\phi}}_s, \boldsymbol{\tilde{\phi}}_s \rangle=0, \; \; \text{ and } \; \;  \langle \boldsymbol{\tilde{z}}_s-\boldsymbol{\tilde{\phi}}_s, \boldsymbol{\phi}\rangle \leq 0, \; \; \forall \boldsymbol{\phi} \in \Omega_s.
\end{equation*}  
Moreover, the above conditions can be adapted to the polar cone as follows: the vector $\boldsymbol{\tilde{\rho}}_s \in \Omega_s^0$ minimizes $||\boldsymbol{\tilde{z}}_s-\boldsymbol{\rho}||^2$ over $\Omega_s^0$ if and only if
\begin{equation} \label{eq:KKTpolar}
\langle\boldsymbol{\tilde{z}}_s-\boldsymbol{\tilde{\rho}}_s, \boldsymbol{\tilde{\rho}}_s \rangle=0, \; \; \text{ and } \; \;  \langle \boldsymbol{\tilde{z}}_s-\boldsymbol{\tilde{\rho}}_s, \boldsymbol{\gamma}_{s_j}\rangle \leq 0 \; \text{ for } \; j=1,2,\dots, m.
\end{equation}

Although the constrained problem in Equation \ref{eq:opt2} does not have a general closed form solution, there are some particular cases where this can be explicitly characterized. For instance, \citet[p.~23]{robertson88} demonstrated that, under partial ordering constraints, the solution $\boldsymbol{\tilde{\theta}}_s$ of the constrained problem in Equation \ref{eq:opt} takes the form
\begin{equation} \label{eq:partialorder}
\tilde{\theta}_{s_d}=\underset{U: d \in U}{\max} \; \; \underset{L: d \in L}{\min} \frac{\sum_{d \in L \cap U }\widehat{N}_d \tilde{y}_{s_d}}{\sum_{d \in L \cap U} \widehat{N}_d}, \; \; \; \text{ for } d=1,\dots, D;
\end{equation}
where $L$ and $U$ are lower and upper sets with respect to the partial ordering, respectively. Equation \ref{eq:partialorder} shows that the proposed constrained estimator is simply pooling neighboring domains in such a way that the imposed constraints are respected. Heuristically, this is an advantageous property for small domains, as it allows them to borrow strength from other domains.

One approach to computing $\boldsymbol{\tilde{\phi}}_s$ is based on the edges of the constraint cone $\Omega_s$. However, the number of edges can be considerably larger than the number of constraints for large values of $D$, especially for the case when there are more constraints than domains (see \citealp{meyer99}). Moreover, given the lack of a general closed form solution for the edges of $\Omega_s$ (when $m>D$), then the edges need to be computed numerically. This task can be a computationally demanding job, which makes this approach an inefficient way to compute $\boldsymbol{\tilde{\phi}}_s$. Fortunately, a more efficient algorithm based on computing the projection onto the polar cone has been developed: the Cone Projection Algorithm (CPA) \citep{meyer13b}. This alternative approach takes advantage of the easy-to-find edges $\boldsymbol{\gamma}_{s_j}$ of the polar cone, the conditions in Equation \ref{eq:KKTpolar}, and the fact that $\Pi(\boldsymbol{\tilde{z}}_s | \Omega_s) = \boldsymbol{\tilde{z}}_s - \Pi(\boldsymbol{\tilde{z}}_s | \Omega_s^0 )$. We remark that the latter fact is a key component on the proofs of the main theoretical results shown in this paper. CPA has been implemented in the software \texttt{R} into the \texttt{coneproj} package. See \cite{liao14} for further details.

\subsection{Variance estimation of $\tilde{\theta}_{s_d}$}

The conditions in Equation \ref{eq:KKTpolar} can be used to show that the projection of $\boldsymbol{\tilde{z}}_s$ onto the polar cone $\Omega_s^0$ coincides with the projection onto the linear space generated by the edges $\boldsymbol{\gamma}_{s_j}$ such that $\langle \boldsymbol{\tilde{z}}_s-\boldsymbol{\tilde{\rho}}_s, \boldsymbol{\gamma_{s_j}} \rangle = 0$. This set of edges could be empty, meaning that the projection onto $\Omega_s^0$ is equal to the projection onto the zero vector. Moreover, this set of edges might not be unique. To formalize this idea, denote $V_{s,J}=\{\boldsymbol{\gamma}_{s_j} : j \in J\}$ for any $J \subseteq \{1,2,\dots, m\}$. Define the set $\mathcal{\overline{F}}_{s,J}$ as, 
\begin{equation} \label{eq:faceclosure}
	\mathcal{\overline{F}}_{s,J}=\{\boldsymbol{\rho} \in \mathbb{R}^D : \boldsymbol{\rho} = \sum_{j\in J}a_j\boldsymbol{\gamma}_{s_j}, \; \;  a_j \geq 0, \; \;  j \in J \},
\end{equation}  
where $\mathcal{\overline{F}}_{s,\emptyset}=\boldsymbol{0}$ by convention. That is, $\mathcal{\overline{F}}_{s,J}$ is the polyhedral sub-cone of $\Omega_s^0$ that starts at the origin and is defined by the edges in $V_{s,J}$. Further, let $\mathcal{L}(V_{s,J})$ be the linear space generated by the vectors in $V_{s,J}$. Hence, projecting onto $\Omega_s^0$ is equivalent to projecting onto $\mathcal{L}(V_{s,J})$, for an appropriate set $J$. 

Estimating appropriately the variance of $\tilde{\theta}_{s_d}$ is a complicated task, derived from the fact that  the projection of $\boldsymbol{\tilde{z}}_s$ onto $\Omega_s^0$ (or onto $\Omega_s$) might not always land on the same linear space $\mathcal{L}(V_{s,J})$ for different samples $s$. To better understand that, define $\mathcal{\tilde{G}}_s$ to the set of all subsets $J \subseteq \{ 1, 2, \dots, m\}$ such that $\Pi(\boldsymbol{\tilde{z}}_s|\Omega_{s}^0) = \Pi(\boldsymbol{\tilde{z}}_s|\mathcal{L}(V_{s,J}))\in \mathcal{\overline{F}}_{s,J}$. The latter definition is motivated by the following non-efficient procedure to find $\boldsymbol{\tilde{\rho}}_s$: project $\boldsymbol{\tilde{z}}_s$ onto each of the $2^m$ linear spaces generated by the edges in $V_{s,J}$, and then check if such a projection lands inside the portion of the polar cone $\Omega_{s}^0$ defined by $V_{s,J}$ (that is, $\mathcal{\overline{F}}_{s,J}$) and that satisfies the conditions stated in Equation \ref{eq:KKTpolar}. Note that, for different samples $s$, the sets $\tilde{\mathcal{G}}_s$ might be different. In addition, the cardinality of $\tilde{\mathcal{G}}_s$ can be greater than one. That is, there could be different sets $J_1$ and $J_2$ such that the projection onto the polar cone $\Omega_s^0$ is equal to projecting onto either $\mathcal{L}(V_{s,{J_1}})$ or $\mathcal{L}(V_{s,{J_2}})$. However, independently of which set is chosen, the projection $\boldsymbol{\tilde{\rho}}_s$ is unique. For instance, consider the case where $m>D$, so the set of all edges $\boldsymbol{\gamma}_{s_j}$ constitutes a linear dependent set of vectors. Hence, there could exist different subsets $J_1, J_2$ that induce the same linear space such that $J_1,J_2 \in \mathcal{\tilde{G}}_s$. A different example where the cardinality of $\mathcal{\tilde{G}}_s$ is greater than 1 is based on the drawn sample. For illustration, consider monotone increasing restrictions with $D=3$. Suppose that $\tilde{y}_{s_1}=\tilde{y}_{s_2}<\tilde{y}_{s_3}$. As there are only 3 domains, the transformed vector $\boldsymbol{\tilde{z}}_s$ has elements of the form
\begin{equation*} \label{eq:example1goodsets}
\tilde{z}_{s_1}=\sqrt{\frac{\widehat{N}_1}{\widehat{N}}}\tilde{y}_{s_1}, \; \; \tilde{z}_{s_2}=\sqrt{\frac{\widehat{N}_2}{\widehat{N}}}\tilde{y}_{s_2}. \; \; \tilde{z}_{s_3}=\sqrt{\frac{\widehat{N}_3}{\widehat{N}}}\tilde{y}_{s_3}. 
\end{equation*}
In this setting, it is straightforward to see that $\Pi(\boldsymbol{\tilde{z}}_s | \Omega_s^0)=\boldsymbol{0}$. However, to compute it, we project $\boldsymbol{\tilde{z}}_s$ onto each of the $2^2=4$ linear spaces generated by the polar cone edges 
\begin{equation*}  \label{eq:example2goodsets}
\boldsymbol{\gamma}_{s_1}=\left(\sqrt{\frac{\widehat{N}}{\widehat{N}_1}},-\sqrt{\frac{\widehat{N}}{\widehat{N}_2}},0\right)^{\top}, \; \; \; \boldsymbol{\gamma}_{s_2}=\left(0,\sqrt{\frac{\widehat{N}}{\widehat{N}_2}},-\sqrt{\frac{\widehat{N}}{\widehat{N}_3}}\right)^{\top}.
\end{equation*}
Hence, it can be seen that the conditions $\Pi(\boldsymbol{\tilde{z}}_s|\Omega_{s}^0) = \boldsymbol{0} = \Pi(\boldsymbol{\tilde{z}}_s|\mathcal{L}(V_{s,J}))\in \mathcal{\overline{F}}_{s,J}$ are satisfied only for $J=\emptyset$ and $J=\{1\}$, which implies that $\tilde{\mathcal{G}}_s=\{ \emptyset, \{1\} \}$. Moreover, note that $V_{s, \emptyset}$ and $V_{s,\{1\}}$ do not span the same linear spaces, which is what complicates the variance estimation of $\tilde{\theta}_{s_d}$. In general, the set of sample vectors where these scenarios occur has measure zero. However, they cannot be excluded at the population level.

We propose a variance estimator for $\tilde{\theta}_{s_d}$ that relies on the sets in $\tilde{\mathcal{G}}_s$ and is based on linearization methods. Consider any $J \in \mathcal{\tilde{G}}_s$, and let $\boldsymbol{P}_{s,J}$ be the projection matrix corresponding to the linear space $\mathcal{L}(V_{s,J})$, where $\boldsymbol{P}_{s,\emptyset}$ is the matrix of zeros by convention. By the selection of $J$, then $\boldsymbol{\tilde{\rho}}_s$ can be expressed as $\boldsymbol{P}_{s,J}\boldsymbol{\tilde{z}}_s$, which implies that $\boldsymbol{\tilde{\theta}}_s$ can be written as $\boldsymbol{\tilde{\theta}}_{s,J}=\boldsymbol{\tilde{y}}_s - \boldsymbol{W}_s^{-1/2}\boldsymbol{P}_{s,J}\boldsymbol{W}_s^{1/2}\boldsymbol{\tilde{y}}_s$, where we add the subscript $J$ in $\boldsymbol{\tilde{\theta}}_s$ to be aware that the expression depends on the chosen $J$.


Now, observe that $\boldsymbol{\tilde{\theta}}_{s,J}$ is a smooth non-linear function of the $\widehat{t}_d$'s and the $\widehat{N}_d$'s, where $\widehat{t}_d$ is the HT estimator of $t_d=\sum_{k \in U_d}y_k$. Therefore, treating $J$ as fixed, we can approximate the variance of $\tilde{\theta}_{s_d,J}$ via Taylor linearization \citep[p.~175]{sarndal92} by
\begin{equation}\label{eq:vartheta}
AV(\tilde{\theta}_{s_d,J})=\sum_{k \in U} \sum_{l \in U} \Delta_{kl}\frac{u_k}{\pi_k}\frac{u_l}{\pi_l},
\end{equation}
where $\Delta_{kl}=\pi_{kl}-\pi_k\pi_l$, and 
\begin{equation*} 
u_k=\sum_{i=1}^D \alpha_iy_{k}1_{k \in U_i} + \sum_{i=1}^D \beta_i1_{k \in U_i} \; \; \text{ for } \; \;  k=1,2,\dots, N,
\end{equation*}
with $1_A$ being the indicator variable for the event $A$, and
\begin{equation*} 
\alpha_i = \frac{\partial \tilde{\theta}_{s_d,J}}{\partial \widehat{t}_i} \Bigr| _{\substack{(\widehat{t}_1, \dots, \widehat{t}_D, \widehat{N}_1, \dots, \widehat{N}_D)=(t_1,\dots, t_D, N_1, \dots, N_D)}}; \; \; \; 
\beta_i = \frac{\partial \tilde{\theta}_{s_d}}{\partial \widehat{N}_i} \Bigr| _{\substack{(\widehat{t}_1, \dots, \widehat{t}_D, \widehat{N}_1, \dots, \widehat{N}_D)=(t_1,\dots, t_D, N_1, \dots, N_D)}}.
\end{equation*}
In addition, a consistent estimator of the approximated variance in Equation \ref{eq:vartheta}, is given by
\begin{equation}\label{eq:varthetaest}
\widehat{V}(\tilde{\theta}_{s_d,J})=\sum_{k \in s} \sum_{l \in s} \frac{\Delta_{kl}}{\pi_{kl}}\frac{\widehat{u}_k}{\pi_k}\frac{\widehat{u}_l}{\pi_l},
\end{equation}
where 
\begin{equation*}
\widehat{u}_k=\sum_{i=1}^D \widehat{\alpha}_iy_{k}1_{ k \in s_i } + \sum_{i=1}^D \widehat{\beta}_i 1_{k \in s_i}  \; \; \text{ for } \; \;  k=1,2,\dots, N,
\end{equation*}
with $\widehat{\alpha}_i, \widehat{\beta}_i$ obtained from $\alpha_i, \beta_i$ by substituting the appropriate Horvitz-Thompson estimators for each total population. Thus, we propose the estimator in Equation \ref{eq:varthetaest} as a variance estimator of $\tilde{\theta}_{s_d}$.

\subsection{Some shape constraints of interest}

As it was mentioned before, irreducible matrices can be used to express a broad range of shape constraints. We include some scenarios of interest with the sole purpose of highlighting the potential utility of our proposed estimator. Several other restrictions can be also considered by our constrained methodology as long as they conform to an irreducible matrix.
\begin{itemize}
	\item \textbf{Double monotone}: domain means are expected to be monotone with respect to two covariates. For instance, average glucose level may increase with people's age, and decrease with mean weekly exercising time.
	\item \textbf{Tree-ordering}: there is one domain mean that is expected to be smaller (or larger) than the others. For example, a placebo effect could be expected to be smaller than treatment effects.
\end{itemize}

In general, combinations of the above shape scenarios could also be considered. For instance, Figure \ref{fig:mixmonotones} contains four different estimates of the population domains means in Figure \ref{fig:mixmonotones}(a): unconstrained estimates are shown in Figure \ref{fig:mixmonotones}(b), and two constrained estimates obtained from different shape restrictions on variables $x_1$ and $x_2$ are shown in Figure \ref{fig:mixmonotones}(c)-(d). Note that unconstrained estimates are wiggly and do not look closer to the population domain means, meanwhile constrained estimates seem to be a more reasonable choice. 
\begin{figure}[ht!] 
	\begin{center}
	\subfigure[Population domain means.]{\includegraphics[width=.49\textwidth]{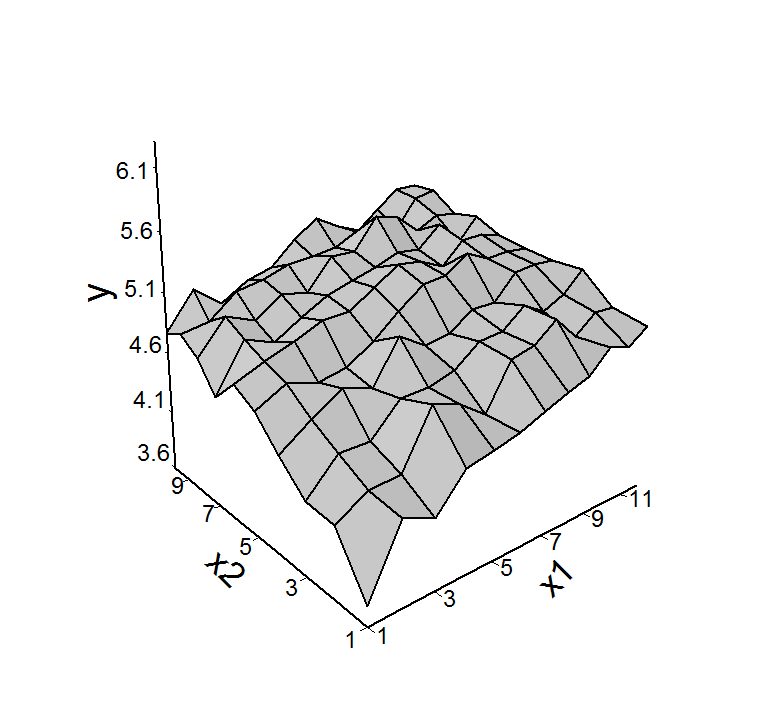}}\
	\subfigure[Unconstrained.]{\includegraphics[width=.49\textwidth]{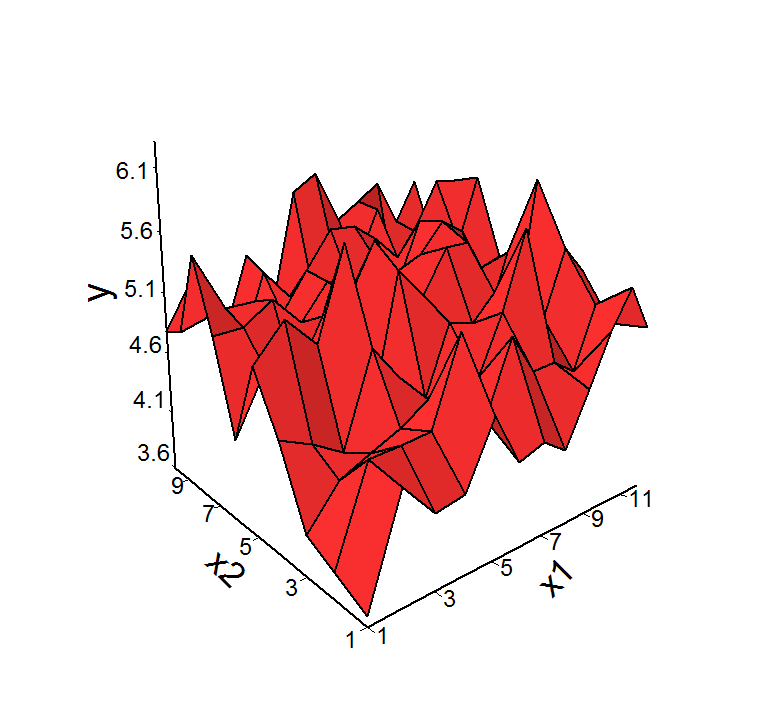}}\\
	\subfigure[$x_1$: monotone, $x_2$: monotone.]{\includegraphics[width=.49\textwidth]{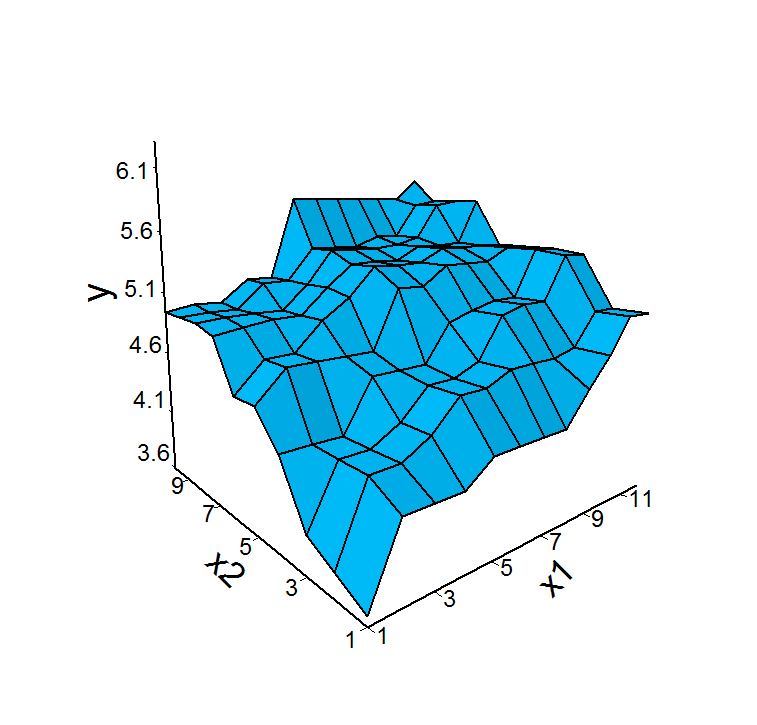}}\
	\subfigure[$x_1$: unconstrained, $x_2$: monotone.]{\includegraphics[width=.49\textwidth]{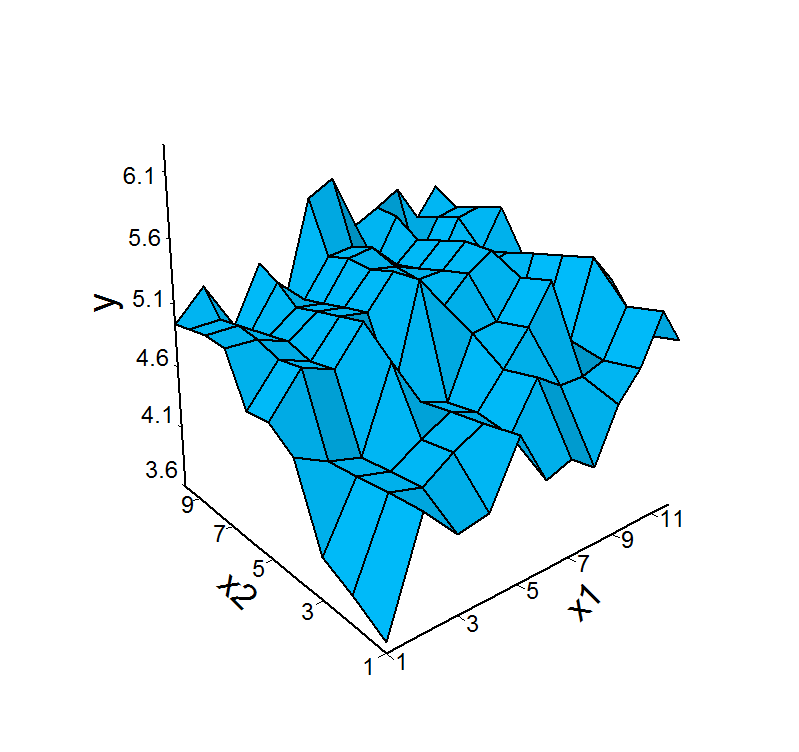}}\\
	\caption{Population domain means and unconstrained estimator (top). Constrained estimator under two different settings of shape constraints (bottom). }
	\label{fig:mixmonotones}
	\end{center}
\end{figure}

\section{Properties of the constrained estimator} \label{sec:properties}

\subsection{Assumptions}

To derive our theoretical results, we make assumptions on the asymptotic behavior of the population $U_N$ and the sampling design $p_N$. Such assumptions are:

\begin{itemize}
	\item[A1.] The number of domains $D$ is fixed.
	\item[A2.] $\underset{N \rightarrow \infty}{\limsup} \;  N^{-1}\sum_{k \in U}y_k^4 <\infty$.
	\item[A3.] There exist constants $\mu_d$ and $r_d>0$ such that $\overline{y}_{U_d, N}-\mu_d=O(N^{-1/2})$ and $N_{d,N}/N-r_d=O(N^{-1/2})$, for all $d$.
	\item[A4.] The sample size $n_N$ is non-random and satisfies $0<\lim_{N \rightarrow \infty} n_N/N <1$. In addition, there exists $\epsilon$, $0<\epsilon<1$, such that $n_{d,N} \geq \epsilon n_N/D$ for all $d$ and all $N$.
	\item[A5.] For all $N$, $\min_{k \in U_N} \pi_k \geq \lambda >0$, $\min_{k,l \in U_N} \pi_{kl} \geq \lambda^* >0$, and 
	\begin{equation*}
	\limsup_{N\rightarrow \infty}n_N \max_{k,l \in U_N : k\neq l} |\Delta_{kl}| <\infty 
	\end{equation*}
	where $\Delta_{kl}=\pi_{kl}-\pi_k\pi_l$.
	\item[A6.] For any vector of $q$ variables $\boldsymbol{x}$ with finite fourth population moment, 
		\begin{equation*}
		\text{var}_{p_N}(\boldsymbol{\widehat{x}}_{s_N})^{-1/2} (\boldsymbol{\widehat{x}}_{s_N}-\boldsymbol{\overline{x}}_{U_N}) \overset{d}{\rightarrow} \mathcal{N}(\boldsymbol{0}, \boldsymbol{I}_q),  
		\end{equation*}
	and
		\begin{equation*}
			\widehat{\text{var}}(\boldsymbol{\widehat{x}}_{s_N}) - \text{var}_{p_N}(\boldsymbol{\widehat{x}}_{s_N}) = o_p(n_N^{-1});  
		\end{equation*}
		
	where $\boldsymbol{\widehat{x}}_{s_N}$ is the HT estimator of $\boldsymbol{\overline{x}}_{U_N}=N^{-1}\sum_{k \in {U_N}}\boldsymbol{x}_k/\pi_k$, $\boldsymbol{I}_q$ is the identity matrix of dimension $q$, the design variance-covariance matrix $\text{var}_{p_N}(\boldsymbol{\widehat{x}}_{s_N})$ is positive definite, and $\widehat{\text{var}}(\boldsymbol{\widehat{x}}_{s_N})$ is the HT estimator of $\text{var}(\boldsymbol{\widehat{x}}_{s_N})$.
\end{itemize}

Assumption A1 establishes that the number of domains remains constant as the population size changes. The condition in Assumption A2 is made to have the property that the difference between design variances and their estimates are on the order of $o_p(n_N^{-1})$. In particular, note that this condition is satisfied when the variable $y$ is bounded, which can be naturally assumed for most types of survey data. Assumption A3 guarantees that the population domain means and sizes converge to the limiting values $\mu_d$ and $r_d$, respectively. Alternatively, the $\mu$ values can be thought as superpopulation parameters that generate the population elements $y_k$. In fact, our theoretical results depend on whether the assumed constraints hold for these superpopulation parameters and not for the population domain means. Although this might seem to be inappropriate given our interest on using constraints at the population level, Assumption A3 ensures that the shape of the domain means would be reasonable close to the shape of the superpopulation means. Assumption A4 states that the sample size in each domain cannot be smaller than a fraction of the ratio $n_N/D$, which would be obtained by dividing equally the sample size over all domains. This assumption aims to ensure that the moments of smooth functions of the $N^{-1}\widehat{t}_d$ and the $N^{-1}\widehat{N}_d$ are bounded. Also, it assumes that the sample size is non-random. However, this can be adapted to a random sample size by imposing certain conditions on the expected sample size $\mathbb{E}_{p_N}(n_N)$. Assumption A5 establishes non-zero lower bounds for both first and second order inclusion probabilities, and states that the design covariances $\Delta_{kl}$ must converge to zero at least as fast as $n_N^{-1}$. Assumption A6 ensures asymptotic normality for a general finite fourth moment vector of variables $\boldsymbol{x}$, which is needed to maintain normality properties on non-linear estimators. Moreover, it establishes consistency conditions on the variance-covariance estimator. 

\subsection{Main results}

Based on the property that $\Pi(\boldsymbol{\tilde{z}}_s | {\Omega}_s)=\boldsymbol{\tilde{z}}_s-\Pi(\boldsymbol{\tilde{z}}_s | {\Omega}_s^0)=\boldsymbol{\tilde{z}}_s-\boldsymbol{\tilde{\rho}}_s$, we derive some theoretical properties of the constrained estimator by focusing on the projection onto $\Omega_s^0$ instead of $\Omega_s$. Recall that the edges of the polar cone $\Omega_s^0$ are simply the $m$ rows of $-\boldsymbol{A}_s$, denoted by $\boldsymbol{\gamma}_{s_j}$; and that $\boldsymbol{\tilde{\rho}}_s$ can be described by the sets $J \in \mathcal{\tilde{G}}_s$. Being able to characterize the property that $J \in \mathcal{\tilde{G}}_s$ in terms of the vectors in $V_{s,J}$ allow us to obtain theoretical convergence rates, which are used to develop inference properties of the constrained estimator. When the set $J \in \mathcal{\tilde{G}}_s$ produces a set of linear independent vectors $V_{s,J}$, then it is straightforward that $\boldsymbol{\tilde{\rho}}_s$ can be written as $\boldsymbol{P}_{s,J}\boldsymbol{\tilde{z}}_s=\boldsymbol{A}_{s,J}^{\top} (\boldsymbol{A}_{s,J}\boldsymbol{A}_{s,J}^{\top})^{-1}\boldsymbol{A}_{s,J}\boldsymbol{\tilde{z}}_s$, where $\boldsymbol{A}_{s,J}$ denotes the matrix formed by the rows of $\boldsymbol{A}_s$ in positions $J$. Hence, based on the conditions in Equation \ref{eq:KKTpolar}, $J \in \mathcal{\tilde{G}}_s$ if and only if
\begin{equation} \label{eq:Jisgoodset}
\langle \boldsymbol{\tilde{z}}_s - \boldsymbol{P}_{s,J}\boldsymbol{\tilde{z}}_s , \boldsymbol{\gamma}_{s_j} \rangle \leq 0 \; \; \; \text{ for } j \notin J, \; \; \; \text{ and }  (\boldsymbol{A}_{s,J}\boldsymbol{A}_{s,J}^{\top})^{-1}\boldsymbol{A}_{s,J}\boldsymbol{\tilde{z}}_s \geq \boldsymbol{0}; 
\end{equation}
where the latter condition assures that $\Pi(\boldsymbol{\tilde{z}}_s|\mathcal{L}(V_{s,J})) \in \overline{\mathcal{F}}_{s,J}$. However, it is possible that the set $J \in \mathcal{\tilde{G}}_s$ produces a set of linearly dependent vectors $V_{s,J}$. In that case, Theorem \ref{theo:linspace} guarantees that it is always possible to find a subset $J^* \subset J$ such that $V_{s,J^*}$ is a linearly independent set that spans the same linear space as $V_{s,J}$, and also, that satisfies $J^* \in \mathcal{\tilde{G}}_s$. Thus, analogous conditions as in Equation \ref{eq:Jisgoodset} can be established using $J^*$ instead of $J$.

\begin{theorem} \label{theo:linspace}
	Let $\boldsymbol{A}$ be a $m \times D$ irreducible matrix with rows $-\boldsymbol{\gamma}_{j}$. Let $\Omega^0$ be its corresponding polar cone. For any set $J \subseteq \{1,2,\dots, m\}$, define $V_{J}=\{\boldsymbol{\gamma}_{j} : j \in J \}$. Further, denote $\mathcal{\overline{F}}_J$ to be the subcone of $\Omega^0$ generated by the edges given by the set $J$. For a vector $\boldsymbol{z}$, define its set $\mathcal{G}$ to be conformed by all sets $J \subseteq \{1,2,\dots, m\}$ such that $\Pi(\boldsymbol{z} | \Omega^0)=\Pi(\boldsymbol{z} | \mathcal{L}(V_J)) \in \mathcal{\overline{F}}_{J}$. 	Suppose $J$ is a non-empty set such that $V_{J}$ is a linearly dependent set and $J\in \mathcal{G}$. Then, there exists $J^* \subset J$ such that $V_{J^*}$ is a linearly independent set, $\mathcal{L}(V_{J^*})=\mathcal{L}(V_J)$, and $J^* \in \mathcal{G}$.
\end{theorem}

All different concepts that have been defined at the sample level, can be analogously defined at the superpopulation level. For instance, let $\mathcal{G}_{\mu}$ be the set of all subsets $J \subseteq \{1,\dots,m\}$ such that $\Pi(\boldsymbol{z}_\mu | \Omega_\mu^0)=\Pi(\boldsymbol{z}_\mu| \mathcal{L}(V_{\mu,J})) \in \mathcal{\overline{F}}_{\mu, J}$, where $\boldsymbol{z}_\mu$, $\Omega_\mu^0$, $V_{\mu,J}$ and $\overline{\mathcal{F}}_{\mu,J}$ are the analogous versions of $\boldsymbol{\tilde{z}}_s$, $\Omega_s^0$, $V_{s,J}$ and $\overline{\mathcal{F}}_{s,J}$ obtained by substituting $\boldsymbol{\tilde{y}}_s$ and $\boldsymbol{W}_s$ by $\boldsymbol{\mu}=(\mu_1, \dots, \mu_D)$ and $\boldsymbol{W}_\mu=\text{diag}(r_1, r_2, \dots, r_D)$. Moreover, necessary and sufficient conditions as in Equation $\ref{eq:KKTpolar}$ can be analogously established to characterize the vector $\boldsymbol{\rho}_\mu$ to be the projection onto $\Omega_\mu^0$.

Recall the set $\mathcal{\tilde{G}}_s$ could vary for different samples. Also, note that highly variable small samples are likely to choose sets $J \in \mathcal{\tilde{G}}_s$ that are not chosen in the `asymptotic true' $\mathcal{G}_\mu$. However, as the sample size increases, these wrong choices are less likely to occur since the sample domain means get closer to the limiting domain means. This intuitive idea is formalized in Theorem \ref{theo:badsets}, which states that sets that are not in $\mathcal{G}_\mu$ have an asymptotic zero probability of being chosen by the sample.

\begin{theorem} \label{theo:badsets}
	Consider any set $J \subseteq \{1,2,\dots,m\}$ such that $J \notin \mathcal{G}_{\mu}$. Then, $P(J \in \mathcal{\tilde{G}}_s)=O(n_N^{-1})$.
\end{theorem}

Theorem \ref{theo:normaldist} contains the main result of this paper, which permits the use of the constrained estimator $\boldsymbol{\tilde{\theta}}_s$ to make inference of the population domain means. This generalizes Theorem 2 of \cite{wu16}, where only monotone restrictions were considered. Note the presence of a bias term $B$ on the mean of the asymptotic normal distribution. We conjecture that this term arises as a consequence of using the estimated variance $\widehat{V}(\tilde{\theta}_{s_d,J})$, solely based on the $J$ chosen by the observed sample, which does not always converge to the asymptotic variance of $\tilde{\theta}_{s_d}$. This undesirable situation occurs when there is more than one set $J \in \mathcal{G}_\mu$ such that their corresponding edges in $V_{\mu,J}$ span different linear spaces, or equivalently, that the projection onto the polar cone $\Omega_\mu^0$ belongs to the intersection of those different linear spaces. In particular, note that the condition $\boldsymbol{A\mu} >  \boldsymbol{0}$ means that the vector $\boldsymbol{z}_\mu$ is strictly inside the constraint cone $\Omega_\mu$, and then, there is no set $J\neq \emptyset$ such that $\Pi(\boldsymbol{z}_\mu| \mathcal{L}(V_{\mu,J}))=\boldsymbol{0}$. Thus, in this case, the bias term vanishes. 

\begin{theorem} \label{theo:normaldist}
	Suppose that $\boldsymbol{\mu}$ satisfies $\boldsymbol{A\mu}\geq \boldsymbol{0}$. Consider any set $J$ such that $J \in \mathcal{\tilde{G}}_s$. Then
	\begin{equation*}
	\widehat{V}(\tilde{\theta}_{s_d,J})^{-1/2}(\tilde{\theta}_{s_d}-\overline{y}_{U_d}) \overset{\mathcal{L}}{\rightarrow} \mathcal{N}(B, 1),
	\end{equation*}
	for any $d=1,2,\dots,D$, where $B=O(\sqrt{\frac{n_N}{N}})$ is a bias term that vanishes when $\boldsymbol{A\mu} > \boldsymbol{0}$.
\end{theorem}

Note that Theorem \ref{theo:normaldist} relies on the fact that the assumed shape constraints hold for the vector of limiting domain means $\boldsymbol{\mu}$ instead of for the vector of population domain means $\boldsymbol{\overline{y}}_U$. In the next section, we show through simulations that the constrained estimator improves both estimation and variability when the population domains are approximately close to the assumed shape, in comparison with unconstrained estimators. 

\section{Performance of constrained estimator} \label{sec:simulations}

\subsection{Simulations}

We run simulation experiments to measure the performance of the proposed methodology to carry out estimation and inference of population domain means. Given a pair of natural numbers $D_1$ and $D_2$, we generate the limiting domain means $\mu_d$ from the monotone bivariate function $\mu(x_1,x_2)$ given by
\begin{equation*} \label{eq:mus}
\mu(x_1,x_2)= \sqrt{1+4x_1/D_1} + \frac{4\exp(0.5+2x_2/D_2)}{1+\exp(0.5+2x_2/D_2)}.
\end{equation*}
The $\mu_d$'s are created by evaluating $\mu(x_1,x_2)$ at every combination of $x_1=1, 2,\dots, D_1$ and $x_2=1, 2, \dots, D_2$, producing a total number of domains equal to $D=D_1D_2$. We set $D_1=6$ and $D_2=4$. Note that although the function $\mu(x_1,x_2)$ produces a matrix rather than a vector of domain means, it can be vectorized in order to represent the limiting domain means as the vector $\boldsymbol{\mu}$. For each domain $d$, we generate its $N_d=N/D=400$ elements by adding i.i.d. normally distributed noise with mean $0$ and variance $\sigma^2$ to the $\mu_d$. Once the elements of the population have been simulated, then the population domain means $\boldsymbol{\overline{y}}_U$ are computed. The population domain means used for simulations when $\sigma=1$ are displayed in Figure \ref{fig:truesim}. Observe that these domain means are reasonably (not strictly) monotone with respect to $x_1$ and $x_2$. 
\begin{figure}[ht!] 
	\begin{center}
		\subfigure{\includegraphics[width=.5\textwidth]{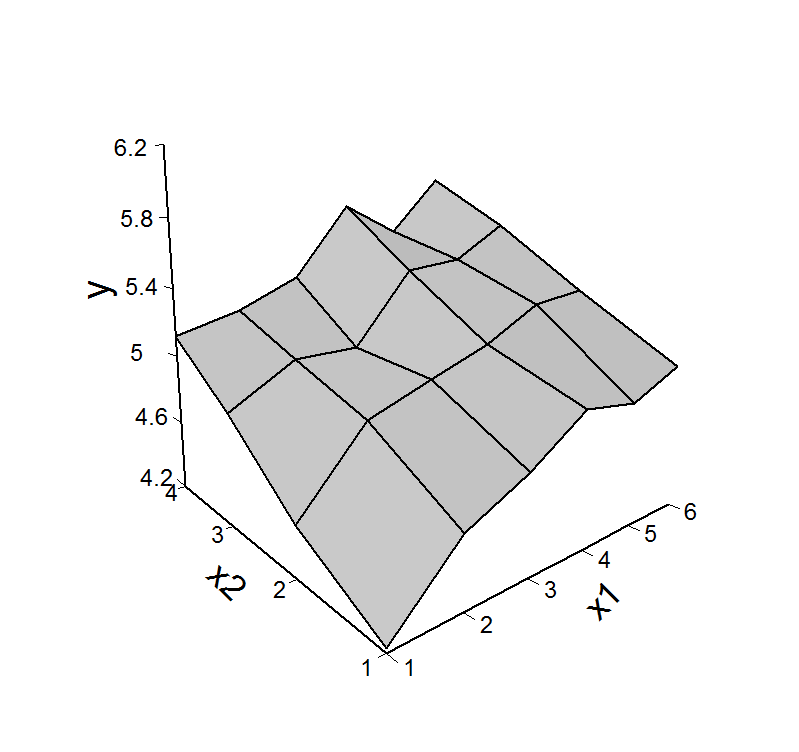}} \
		\caption{Population domain means for simulations when $\sigma=1$.}
		\label{fig:truesim}
	\end{center}
\end{figure}

Samples are drawn from a stratified sampling design without replacement, with $4$ strata that cut across the $D$ domains. Strata are constructed using an auxiliary variable $\nu$ that is correlated with the variable of interest $y$. The vector $\nu$ is created by adding i.i.d. standard normal distributed noise to $\sigma d/D$, for each element in domain $d$. Then, stratum membership is assigned by ranking the vector $\nu$, and creating $4$ blocks of $N/4=2400$ elements each based on such ranks. To make the design informative, we sample $n_N=480$ elements divided across strata in (60, 120, 120, 180).  This probability sampling design is similar to the one described in \cite{wu16}.

We consider 4 different scenarios obtained from the combination of two possible types of shape constraints and $\sigma=1$ or $2$. The first type of constraints assumes the population domain means are monotone increasing with respect to both $x_1$ and $x_2$ (double monotone), while the second type of constraints assumes monotonicity only with respect to $x_1$ (only $x_1$ monotone). Moreover, for a fixed $\sigma$, the exact same population is considered for the two possible types of constraints. For each scenario, the unconstrained $\boldsymbol{\tilde{y}}_s$ and constrained $\boldsymbol{\tilde{\theta}}_s$ estimates are computed along with their linearization-based variance estimates (Equation \ref{eq:varthetaest}). Constrained estimates are computed using the CPA, and their variance estimates are computed by relying on the sample-selected set $J \in \mathcal{\tilde{G}}_s$. In addition, 95\% Wald confidence intervals based on the normal distribution are constructed for both estimators. The lengths of these confidence intervals are omitted because they have the same behavior (up to the constant 1.96) as the variance estimates.  

To measure the precision of $\boldsymbol{\tilde{y}}_s$ and $\boldsymbol{\tilde{\theta}}_s$ as estimators of the population domain means $\boldsymbol{\overline{y}}_U$, we consider the Weighted Mean Squared Error (WMSE) given by
\begin{equation*}\label{eq:mse}
\text{WMSE}(\boldsymbol{\tilde{\varphi}}_s)=\mathbb{E} \left[ (\boldsymbol{\tilde{\varphi}}_s-\boldsymbol{\overline{y}}_U )^{\top} \boldsymbol{W}_U (\boldsymbol{\tilde{\varphi}}_s-\boldsymbol{\overline{y}}_U ) \right],
\end{equation*} 
where $\boldsymbol{\tilde{\varphi}}_s$ could be either the unconstrained or constrained estimator, $\boldsymbol{W}_U$ is the diagonal matrix with elements $N_d/N$, $d=1,\dots, D$.The WMSE values are approximated by simulations.

Simulation results are summarized in Figures \ref{fig:doublemonsigma1} - \ref{fig:onevarsigma2}, and are based on $R=10000$ replications. These display the 24 domains divided in groups of 6, where each is assumed to be monotone. For the double monotone scenario, similar plots with groups of 4 monotone domains each can be also pictured. From the fitting one sample plots, it can be seen that the constrained estimates can be exactly equal to the unconstrained estimates for some domains. In those cases, their variance estimates are also equal. Also, confidence intervals for the constrained estimator tend to be tighter in comparison with those for the unconstrained estimator. On average, the constrained estimator behaves slightly differently than the population domain means, due to their non-strict monotonicity. As an advantage, the percentiles for the constrained estimator are narrower, demonstrating the distribution of the proposed estimator is tighter than the distribution of the unconstrained estimator. For small values of $\sigma$, unconstrained estimates are closer to satisfy the assumed restrictions, which leads to small improvements on the constrained estimator over the unconstrained. In contrast, shape assumptions tend to be more severely violated in unconstrained estimates for larger values of $\sigma$, allowing the proposed estimator to gain much more efficiency on these cases. This latter property can be noted by observing that the constrained estimator percentile band gets farther away from the unconstrained estimator band as $\sigma$ increases.

In terms of variability, the constrained estimator has the smaller variance of the two estimators. However, on average, it gets overestimated by its corresponding linearization-based variance estimate. This might be a direct consequence of estimating the variance based only on the set $J \in \mathcal{\tilde{G}}_s$, which is actually a random set that might change from sample to sample. In contrast, the variance estimate of the unconstrained estimator underestimates the true variance, on average. Although it would be ideal to improve both of these variance estimates, we consider it to be less alarming to produce greater variance estimates, at least for inference purposes. In addition, confidence intervals for both estimators demonstrate a similar good coverage rate when $\sigma=1$, meanwhile such coverage gets slightly improved by the constrained estimator when $\sigma=2$.  

\begin{figure}[ht!] 
	\begin{center}
		\subfigure{\includegraphics[width=.49\textwidth]{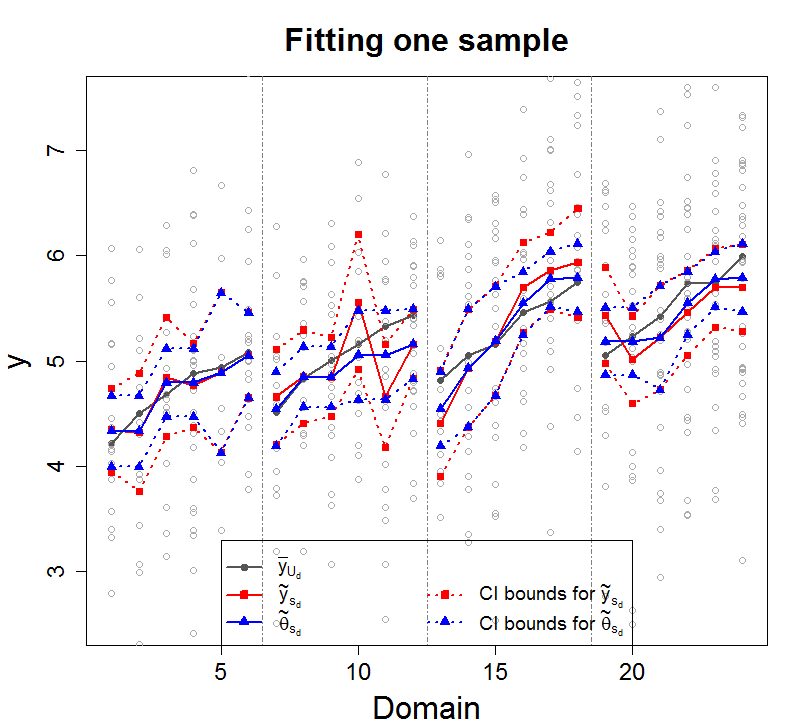}}\
		\subfigure{\includegraphics[width=.49\textwidth]{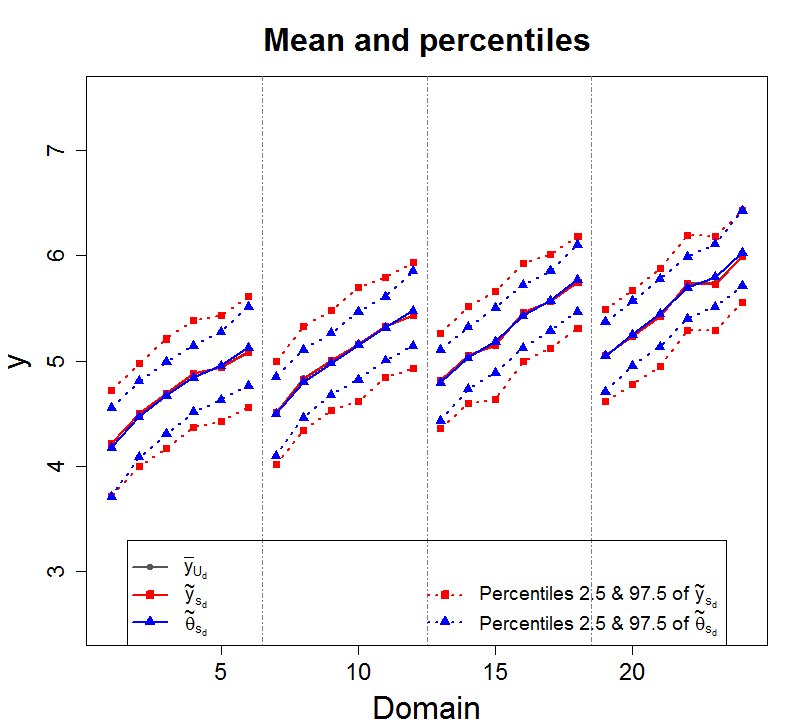}}\\
		\subfigure{\includegraphics[width=.49\textwidth]{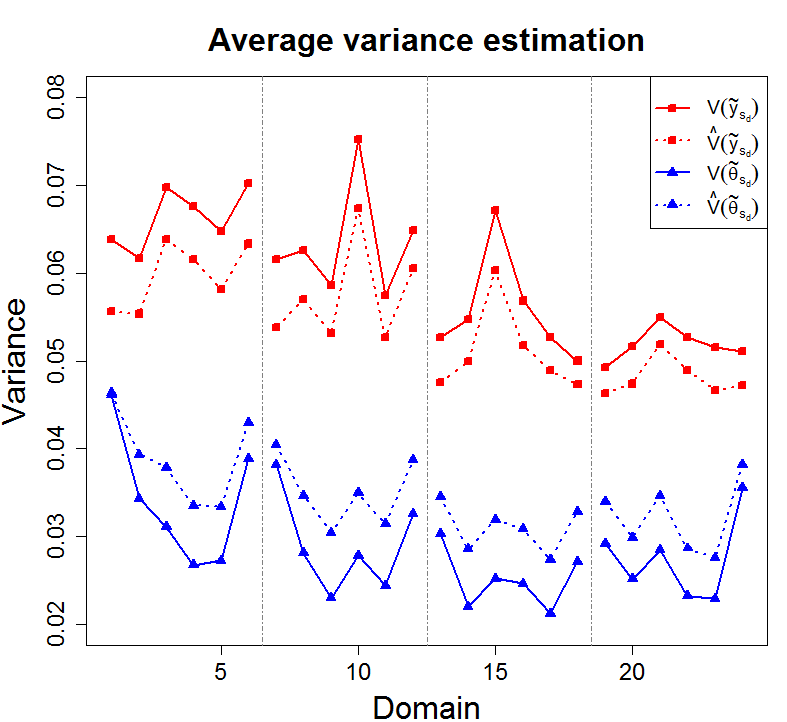}}\
		\subfigure{\includegraphics[width=.49\textwidth]{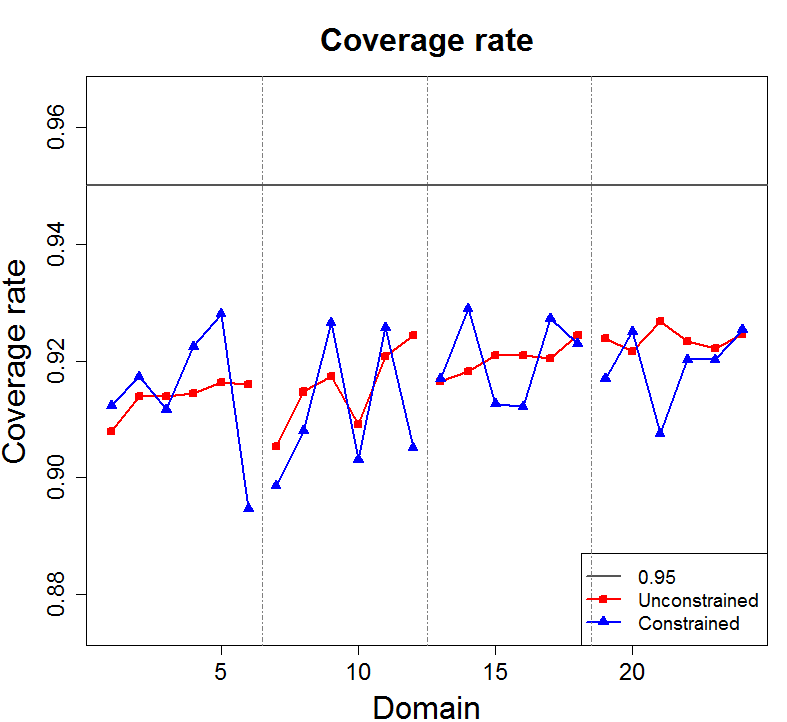}}
		\caption{Plots of simulation results for the unconstrained and constrained estimators under the double monotone scenario with $\sigma=1$, based on $10000$ replications.}
		\label{fig:doublemonsigma1}
	\end{center}
\end{figure} 

\begin{figure}[ht!] 
		\subfigure{\includegraphics[width=.49\textwidth]{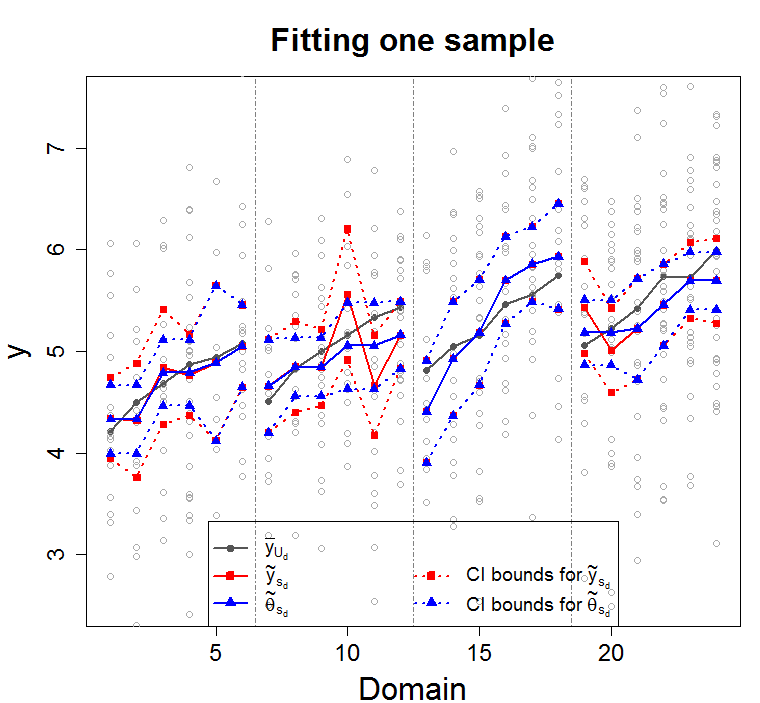}}\
		\subfigure{\includegraphics[width=.49\textwidth]{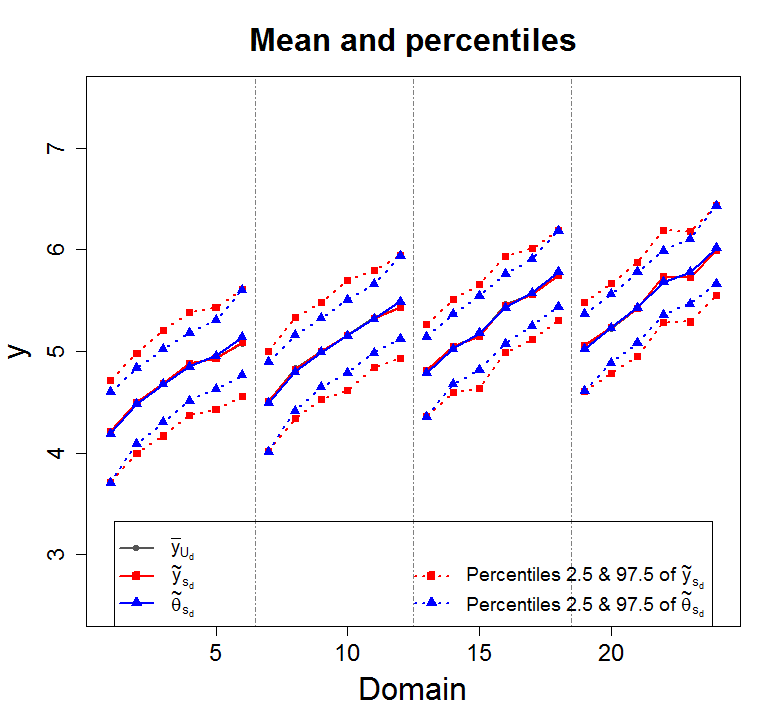}}\\
		\subfigure{\includegraphics[width=.49\textwidth]{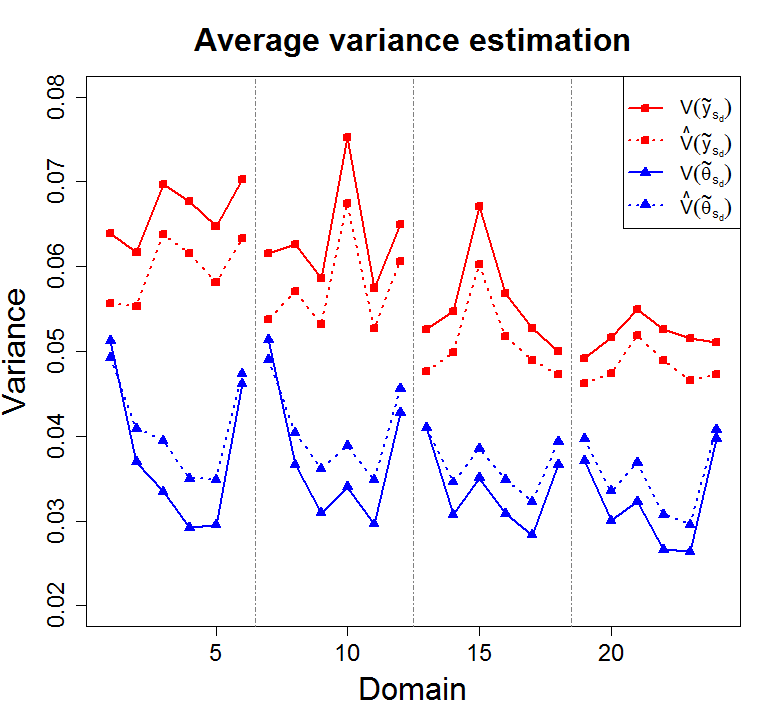}}\
		\subfigure{\includegraphics[width=.49\textwidth]{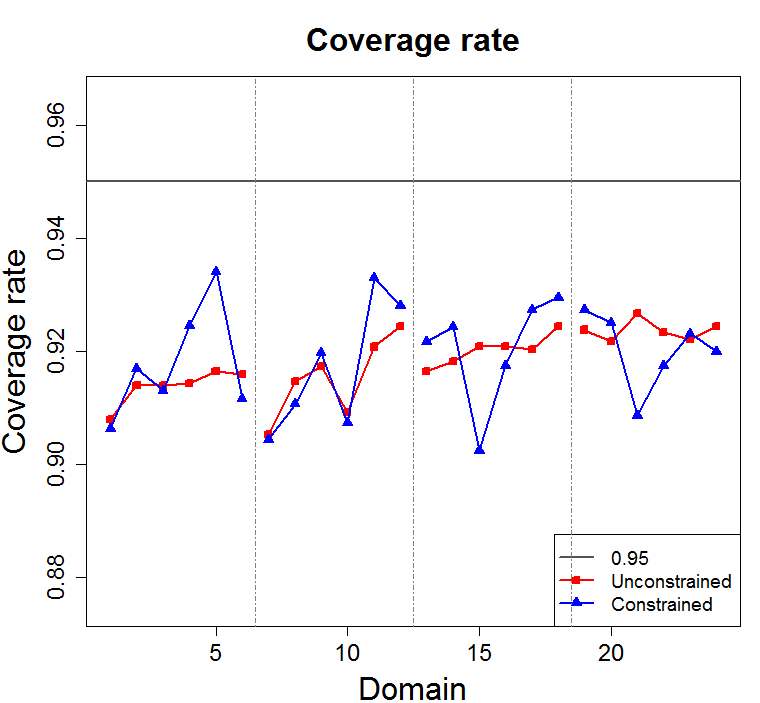}}
		\caption{Plots of simulation results for the unconstrained and constrained estimators under the only $x_1$ monotone scenario with $\sigma=1$, based on $10000$ replications.}
		\label{fig:onevarsigma1}
		\begin{center}
	\end{center}
\end{figure} 

\begin{figure}[ht!] 
	\begin{center}
		\subfigure{\includegraphics[width=.49\textwidth]{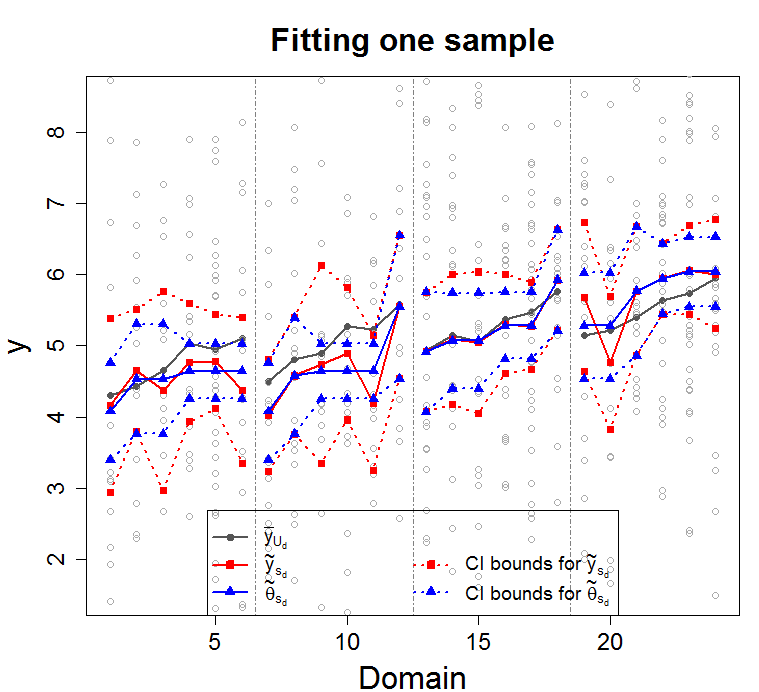}}\
		\subfigure{\includegraphics[width=.49\textwidth]{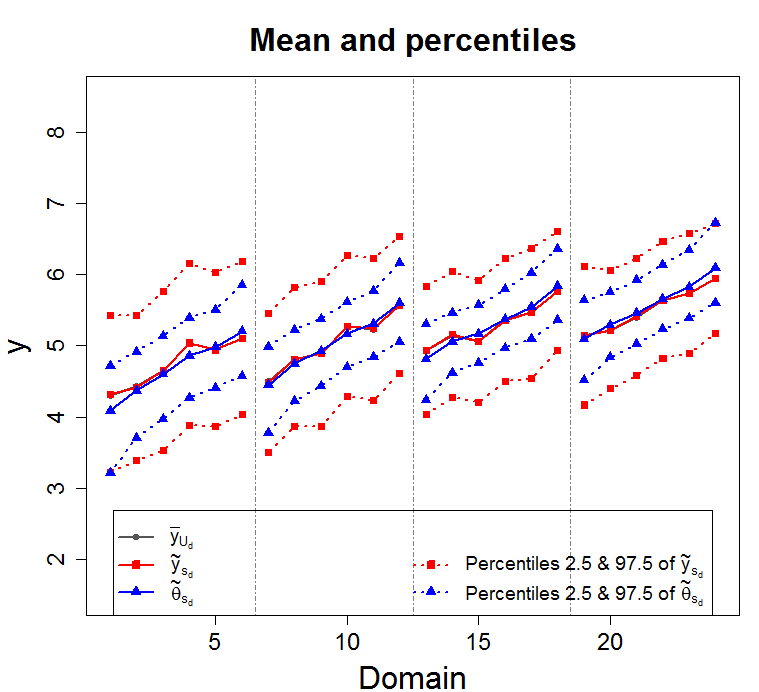}}\\
		\subfigure{\includegraphics[width=.49\textwidth]{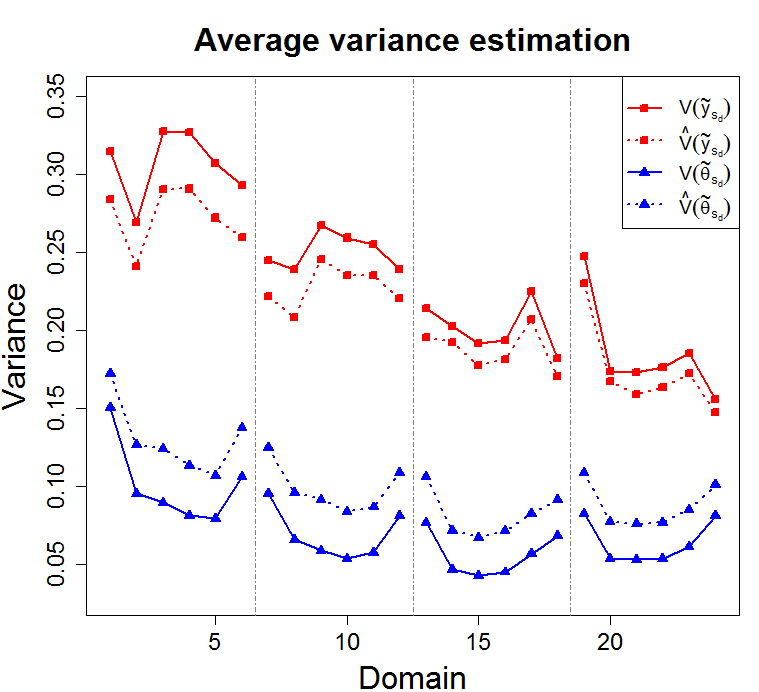}}\
		\subfigure{\includegraphics[width=.49\textwidth]{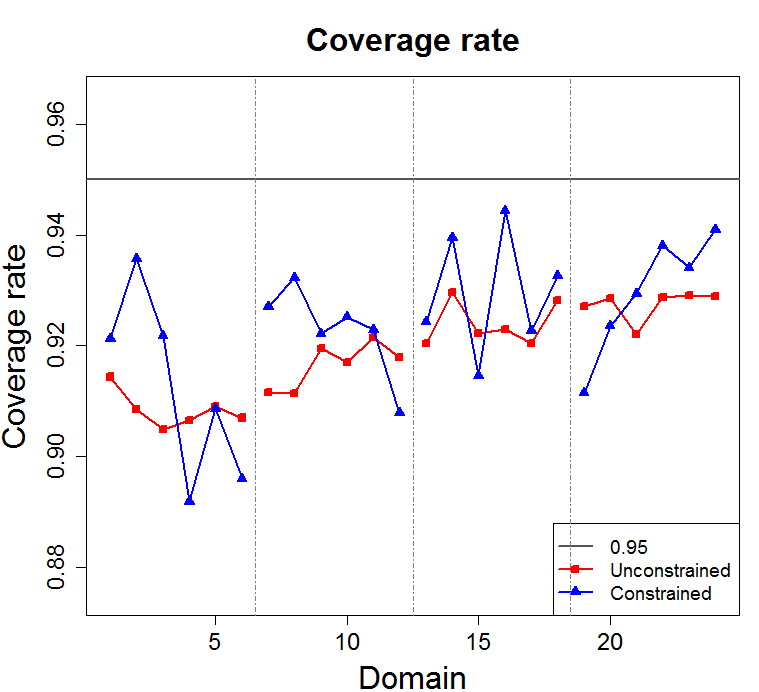}}
		\caption{Plots of simulation results for the unconstrained and constrained estimators under the double monotone scenario with $\sigma=2$, based on $10000$ replications.}
		\label{fig:doublemonsigma2}
	\end{center}
\end{figure} 

\begin{figure}[ht!] 
	\begin{center}
		\subfigure{\includegraphics[width=.49\textwidth]{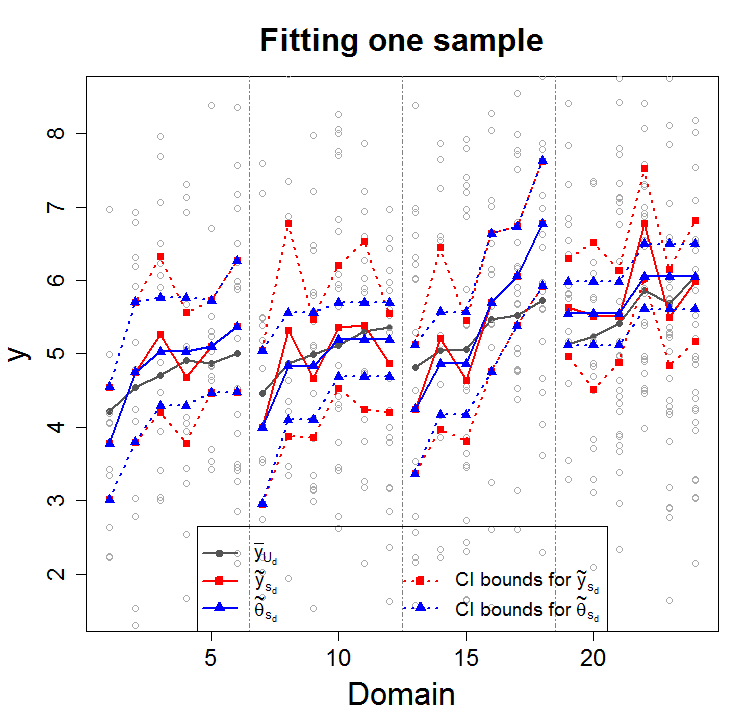}}\
		\subfigure{\includegraphics[width=.49\textwidth]{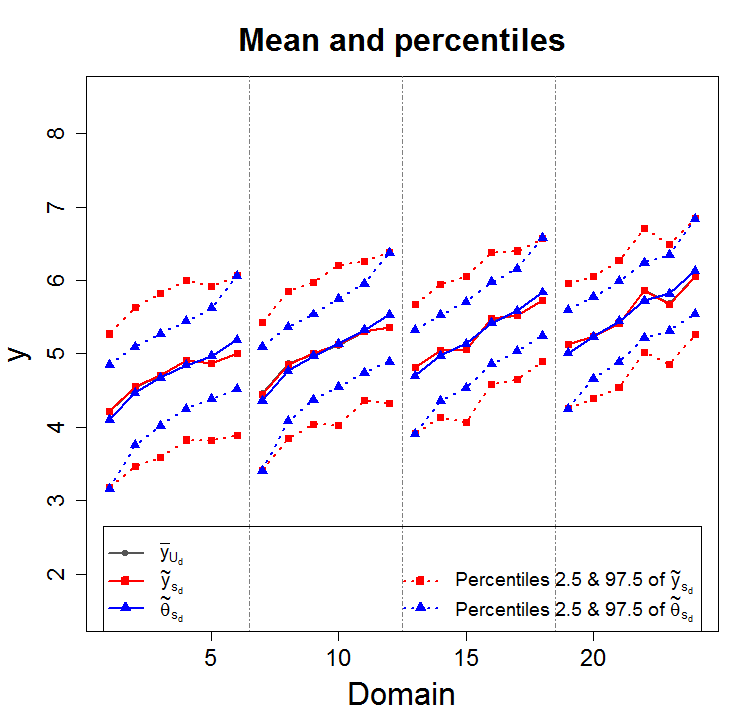}}\\
		\subfigure{\includegraphics[width=.49\textwidth]{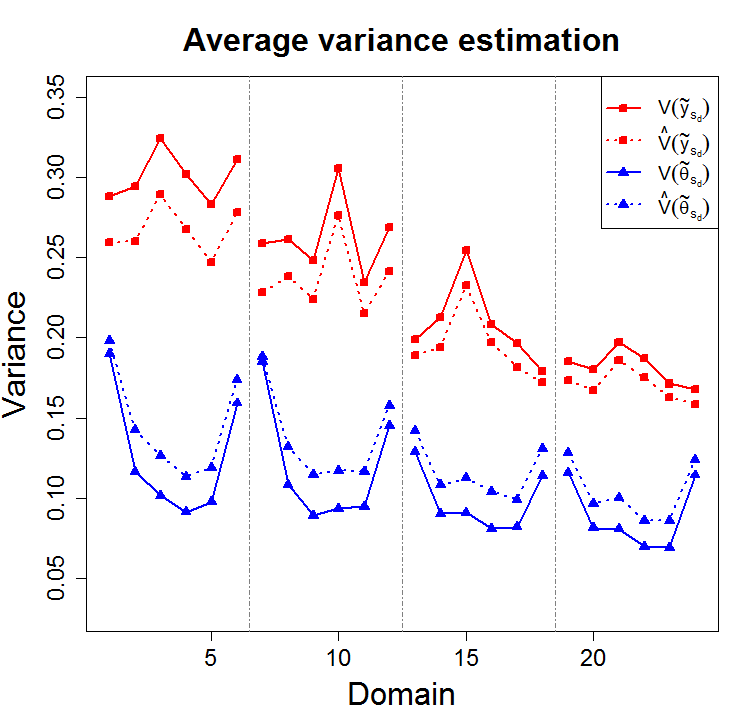}}\
		\subfigure{\includegraphics[width=.49\textwidth]{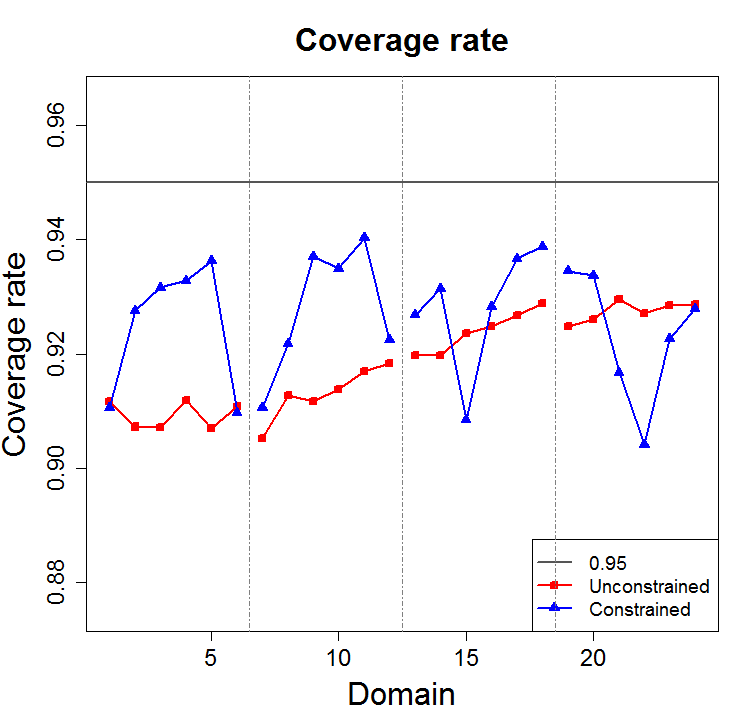}}
		\caption{Plots of simulation results for the unconstrained and constrained estimators under the only $x_1$ monotone scenario with $\sigma=2$, based on $10000$ replications.}
		\label{fig:onevarsigma2}
	\end{center}
\end{figure} 

Table \ref{tab:mse} shows that the constrained estimator is more precise on average than the unconstrained estimator, even though the population domain means are not strictly monotone with respect to $x_1$ and $x_2$. Moreover, the precision of the constrained estimator gets improved when the monotonicity with respect to the two variables is assumed, instead of only with respect to $x_1$. This can be translated on stating that the precision of the proposed estimator is benefited by taking into account the most appropriate shape assumptions.

\begin{table}[ht!]
	\begin{center}
		\begin{tabular}{c|c|c|c}
		& \text{Unconstrained} & \text{Only $x_1$ monotone} & \text{Double monotone} \\ \hline
		$\sigma=1$ & 0.0593 & 0.0362 & 0.0298 \\
		$\sigma=2$ & 0.2384 & 0.1175 & 0.0832 	
		\end{tabular}
	\end{center}
	\caption{WMSE values.}
	\label{tab:mse}
\end{table}	

\FloatBarrier

\subsection{Replication methods for variance estimation}

Recently, it is more common that large-scale surveys make use of replication-based methods for variance estimation. Some examples of such surveys are the last editions of the NHANES and the National Survey of College Graduates (NSCG), the latter sponsored by the National Science Foundation (NSF). To study the performance of replication-based variance estimators under the proposed constrained methodology, we carry out simulation studies based on the delete-a-group Jackknife (DAGJK) variance estimator proposed by \cite{kott2001}.

We perform replication-based simulation experiments using the setting described in Section 4.1. To compute the DAGJK variance estimator, we first randomly create $G$ equal-sized groups within each of the $4$ strata. Then, for each possible $g$, we delete the $g$-th group in each of the strata, adjust the remaining weights by $w_{k}^{(g)}=(\frac{G}{G-1})w_k$, where $w_k=\pi_k^{-1}$; and compute the replicate constrained estimate $\boldsymbol{\tilde{\theta}}_{s}^{(g)}$ using the adjusted weights. Hence, the DAGJK variance estimate of $\tilde{\theta}_{s_d}$, $\widehat{V}_{JK}({\tilde{\theta}_{s_d}})$, is obtained by calculating
\begin{equation*} \label{eq:repbased}
\widehat{V}_{JK}({\tilde{\theta}_{s_d}}) = \frac{G-1}{G}\sum_{g=1}^G \left( \tilde{\theta}_{s_d}^{(g)} - \tilde{\theta}_{s_d} \right)^2.
\end{equation*}  
Analogously, a replication-based variance estimator of $\tilde{y}_{s_d}$ can be derived by substituting the role of $\boldsymbol{\tilde{\theta}}_s$ by $\boldsymbol{\tilde{y}}_s$.

Our simulations consider only the double monotone scenario, with $\sigma=1$ or $2$, and $G=10,20$ or $30$. Moreover, the sample size is set to either $n_N=480$ or $n_N=960$, where the latter case is obtained by doubling the original sample size in each strata. Figures \ref{fig:repbasedsigma1} - \ref{fig:repbasedsigma2sample960} contain our replication-based simulation results based on $10000$ replications. From these, it can be noted that the DAGJK estimates tend to overestimate the variance of the unconstrained estimator, meanwhile the linearization-based variance estimate has an underestimating behavior. In contrast, both replication-based and linearization-based variance estimates of the constrained estimator overestimate the true variance. Moreover, note that as the number of groups $G$ increases, DAGJK estimates tend to be greater, especially for small values of $\sigma$. Such increments on DAGJK estimates have the direct consequence of increasing the coverage rate as $G$ gets larger. In addition, the coverage rate for both estimators is improved (closer to 0.95) when the sample size is increased. As a general conclusion in terms of the constrained estimator, DAGJK variance estimators have a similar behavior than linearization-based estimators. Thus, it seems appropriate to adapt the proposed constrained methodology to allow the use of replication-based variance estimation methods.

\begin{figure}[ht!]
	\begin{center}
	\subfigure{\includegraphics[width=.49\textwidth]{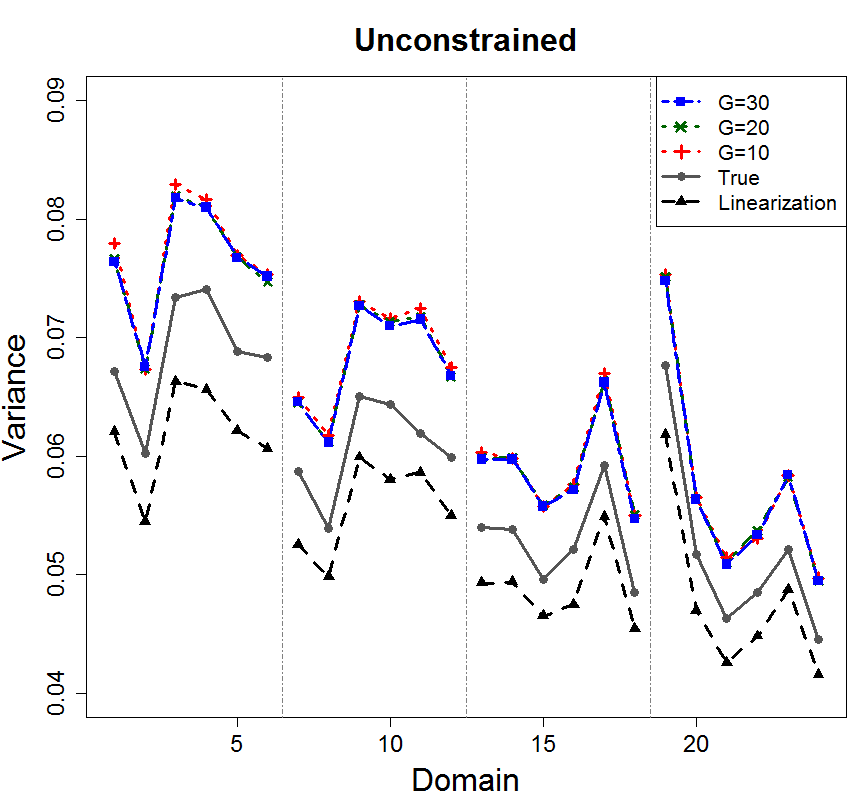}}\
	\subfigure{\includegraphics[width=.49\textwidth]{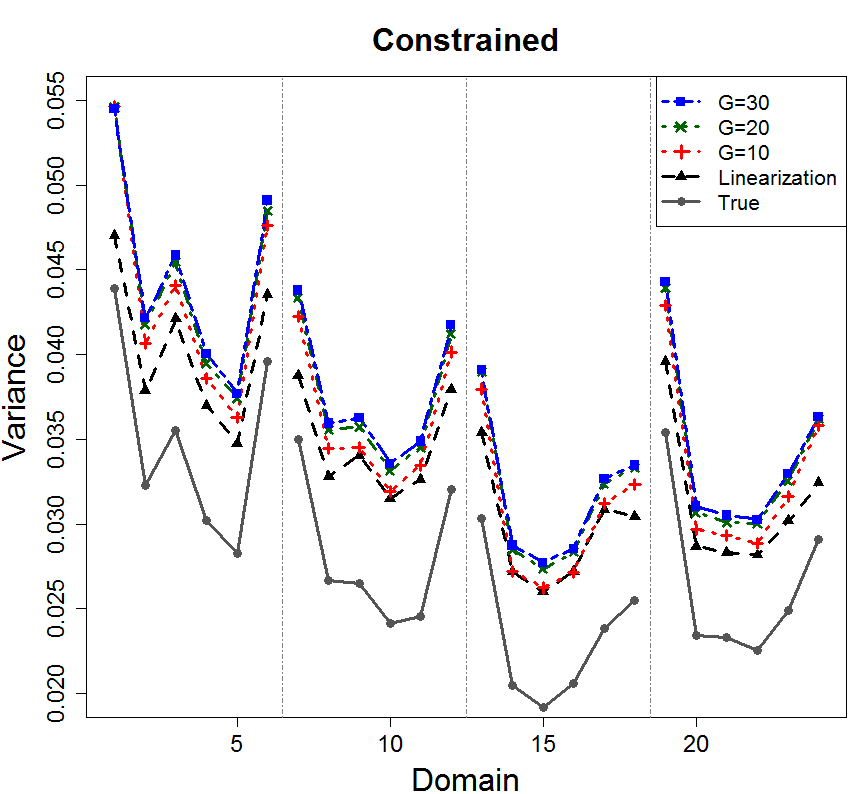}}\\
	\subfigure{\includegraphics[width=.49\textwidth]{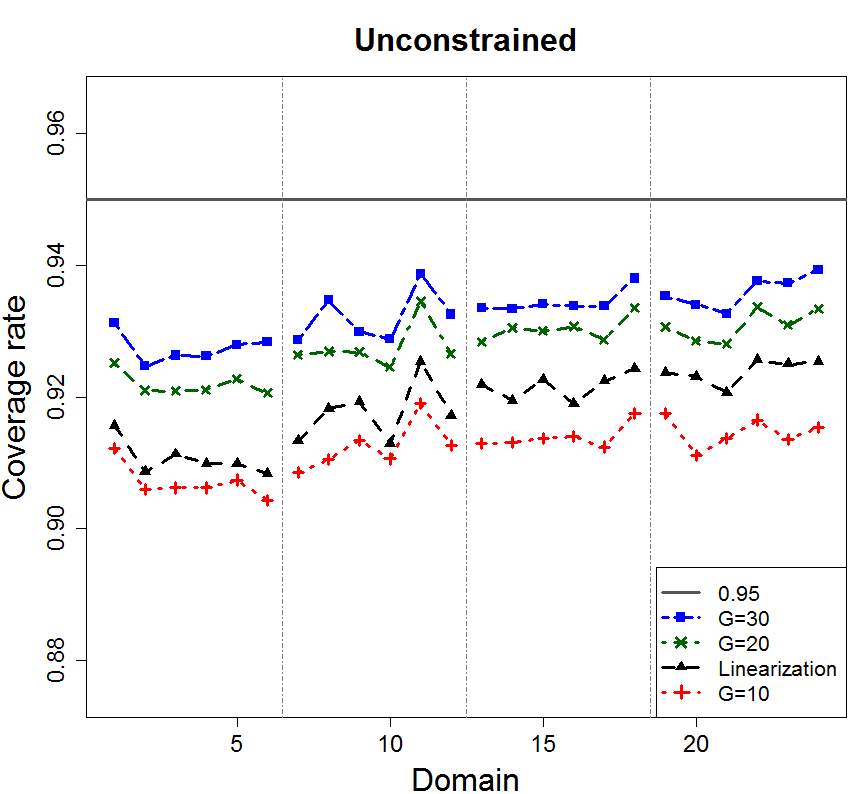}}\
	\subfigure{\includegraphics[width=.49\textwidth]{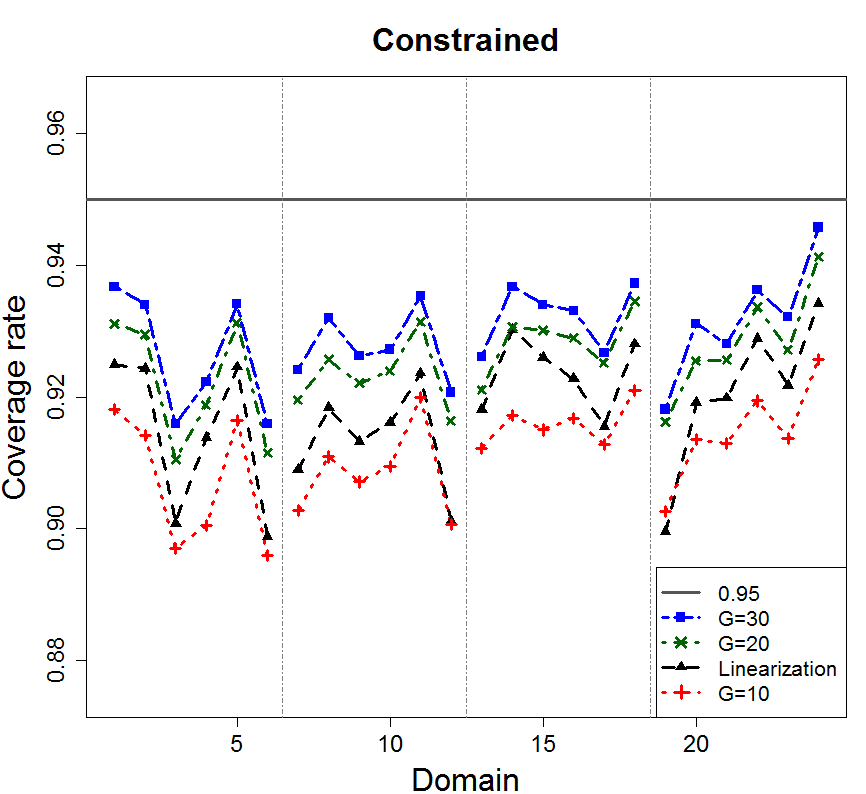}}
	\caption{Variance estimation (top) and coverage rate (bottom) simulation results based on linearization and DAGJK methods for the unconstrained (left) and constrained (right) estimators, under the double monotone scenario with $n_N=480$ and $\sigma=1$.}
	 \label{fig:repbasedsigma1}
	\end{center}
\end{figure} 

\begin{figure}[ht!] 
	\begin{center}
		\subfigure{\includegraphics[width=.49\textwidth]{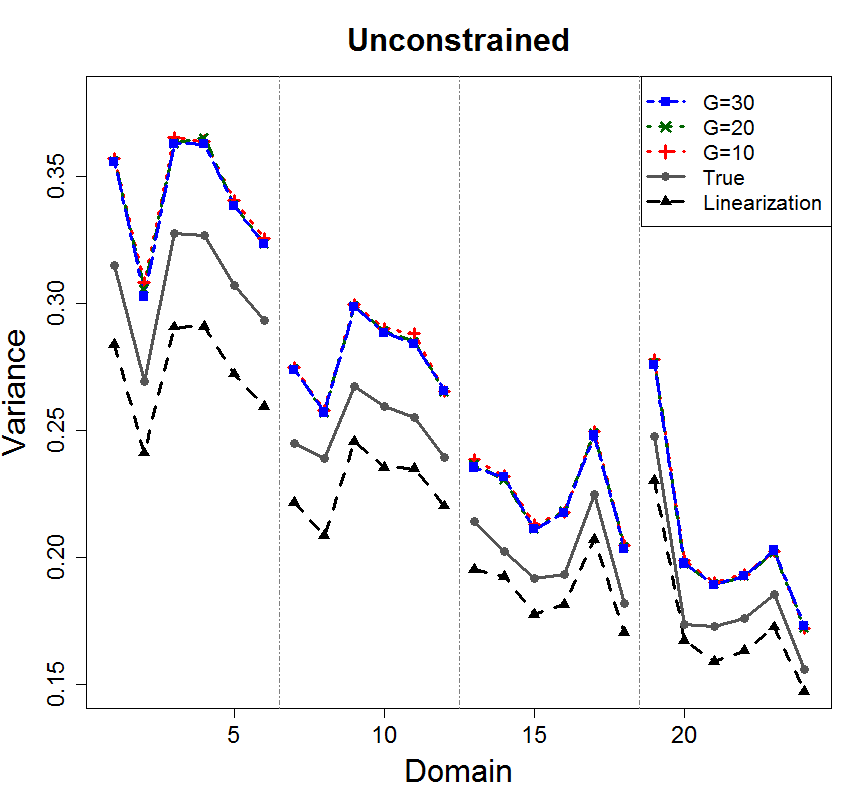}}\
		\subfigure{\includegraphics[width=.49\textwidth]{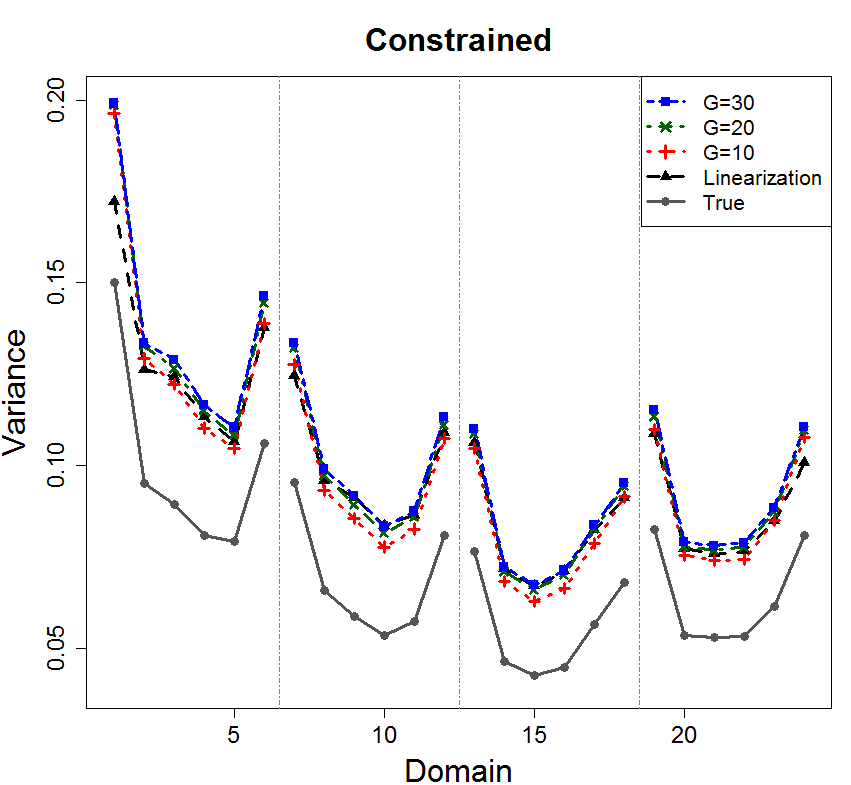}}\\
		\subfigure{\includegraphics[width=.49\textwidth]{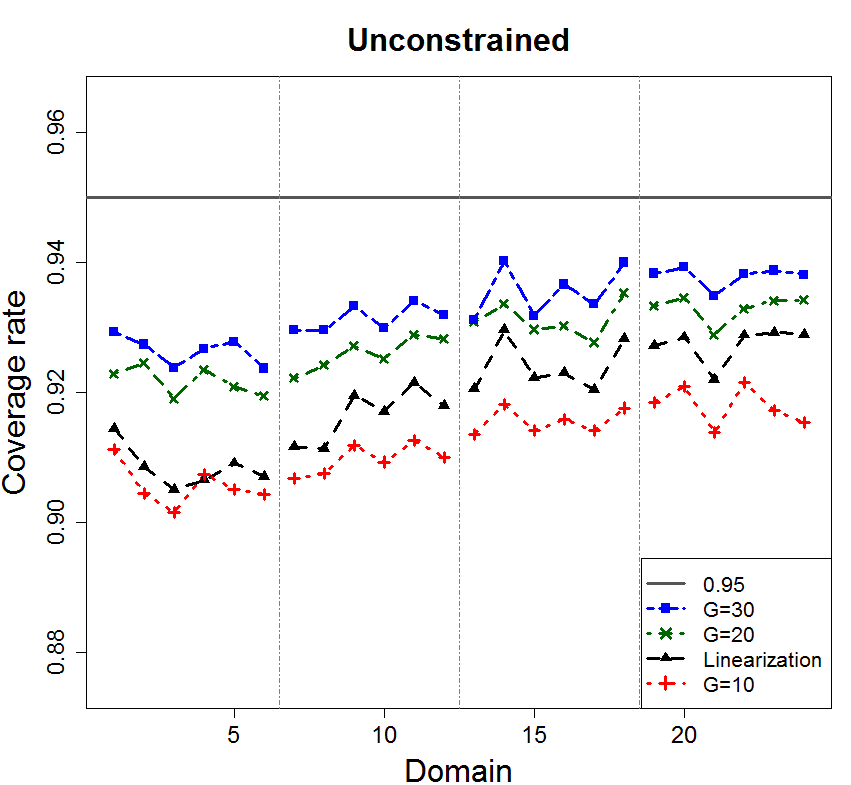}}\
		\subfigure{\includegraphics[width=.49\textwidth]{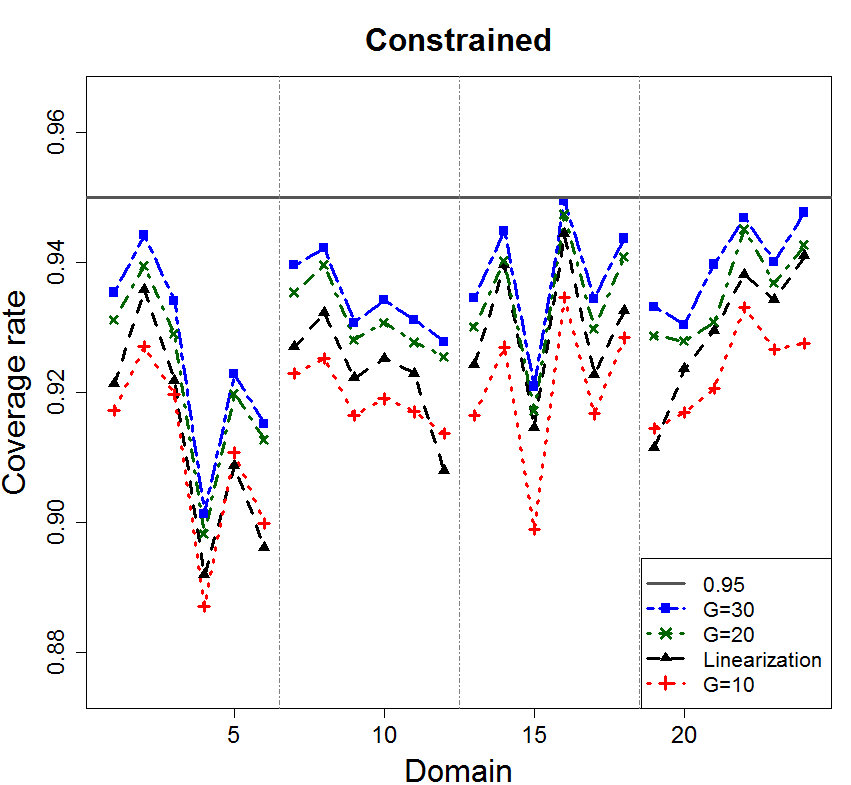}}
		\caption{Variance estimation (top) and coverage rate (bottom) simulation results based on linearization and DAGJK methods for the unconstrained (left) and constrained (right) estimators, under the double monotone scenario with $n_N=480$ and $\sigma=2$.}
		\label{fig:repbasedsigma2}
	\end{center}
\end{figure} 

\begin{figure}[ht!] 
	\begin{center}
		\subfigure{\includegraphics[width=.49\textwidth]{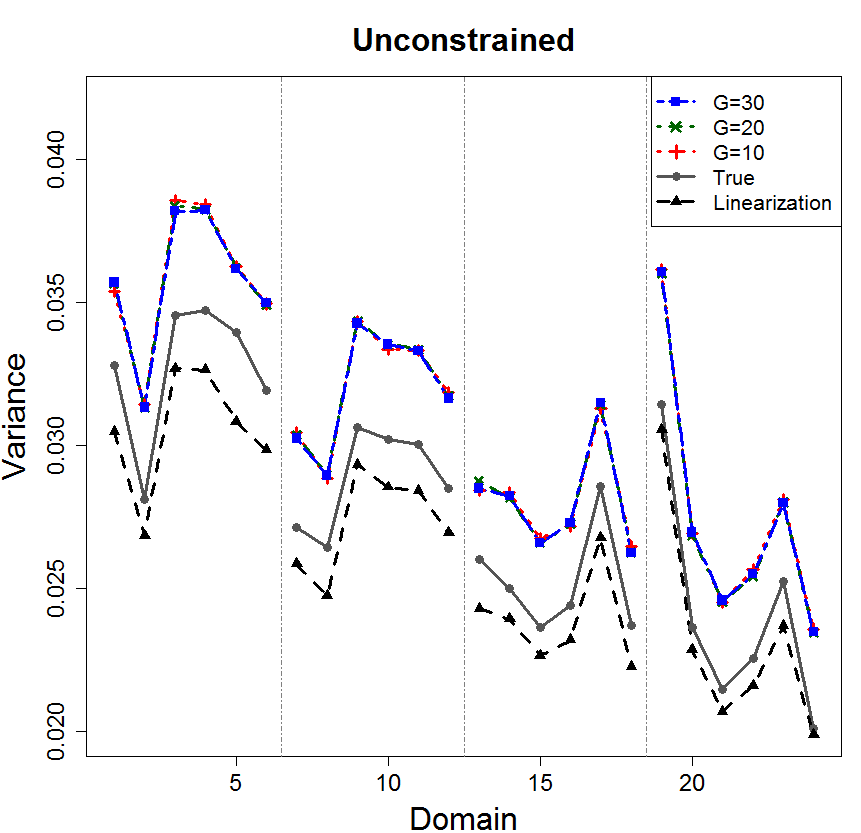}}\
		\subfigure{\includegraphics[width=.49\textwidth]{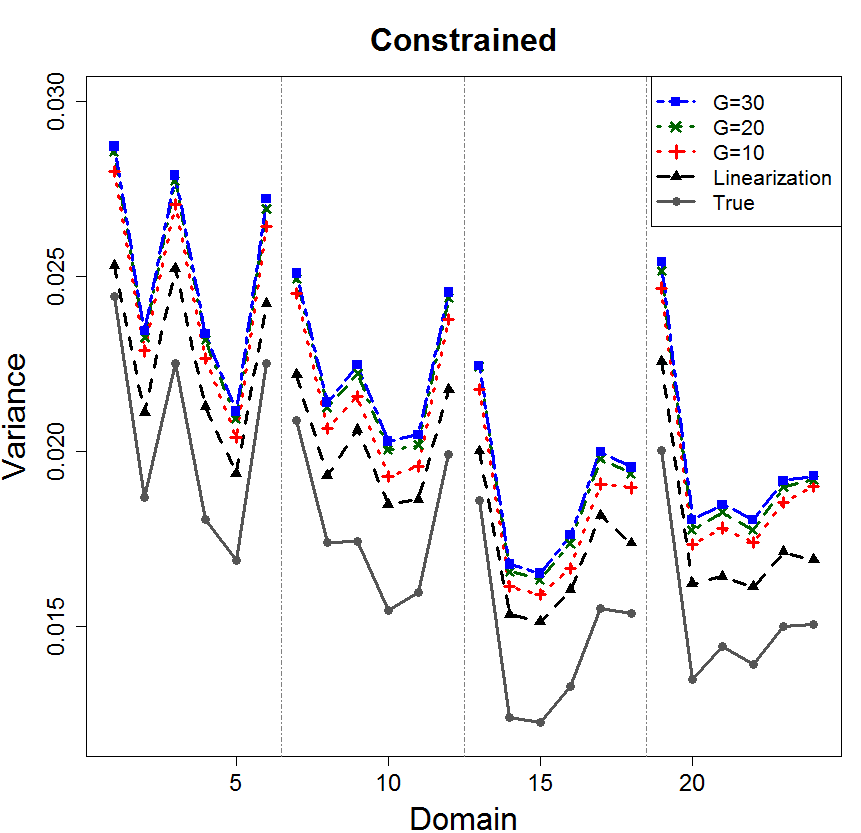}}\\
		\subfigure{\includegraphics[width=.49\textwidth]{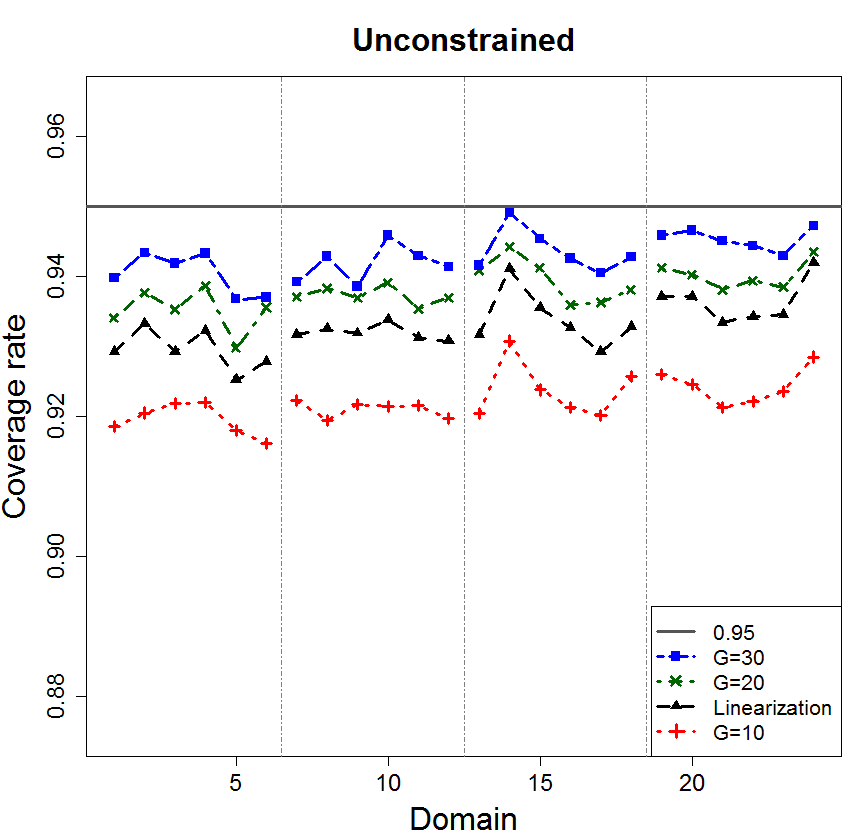}}\
		\subfigure{\includegraphics[width=.49\textwidth]{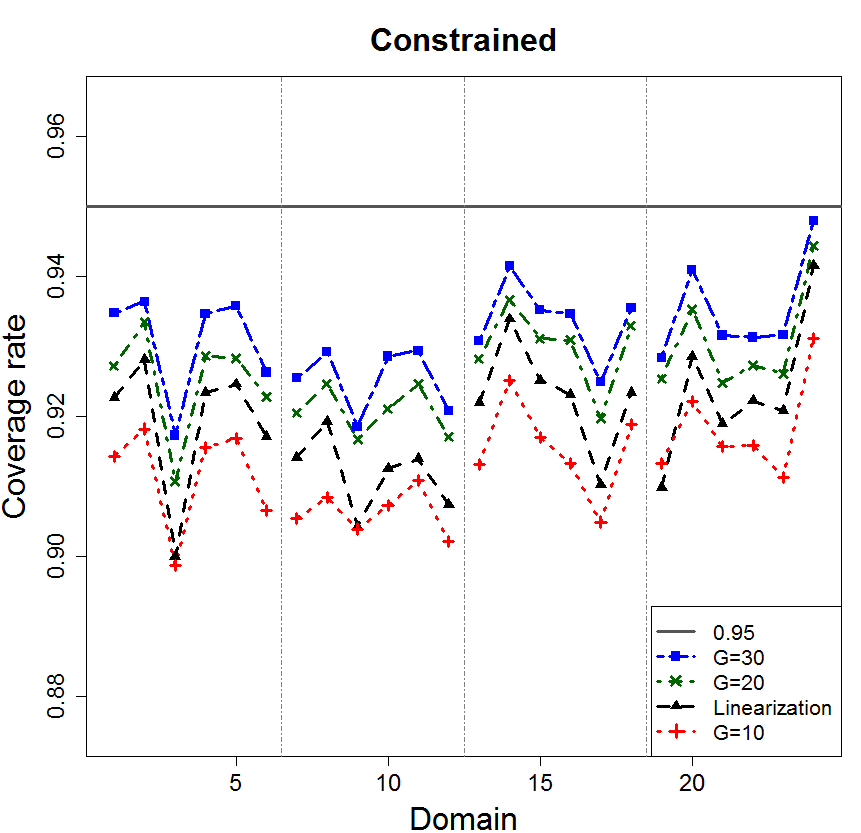}}
		\caption{Variance estimation (top) and coverage rate (bottom) simulation results based on linearization and DAGJK methods for the unconstrained (left) and constrained (right) estimators, under the double monotone scenario with $n_N=960$ and $\sigma=1$.}
		\label{fig:repbasedsigma1sample960}
	\end{center}
\end{figure} 

\begin{figure}[ht!] 
	\begin{center}
		\subfigure{\includegraphics[width=.49\textwidth]{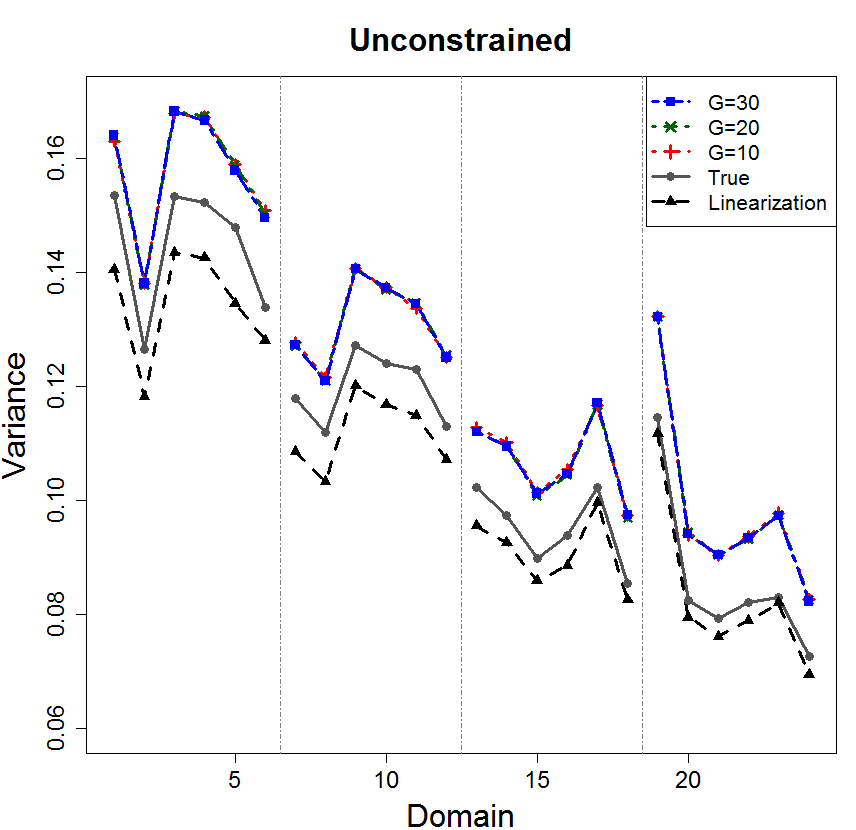}}\
		\subfigure{\includegraphics[width=.49\textwidth]{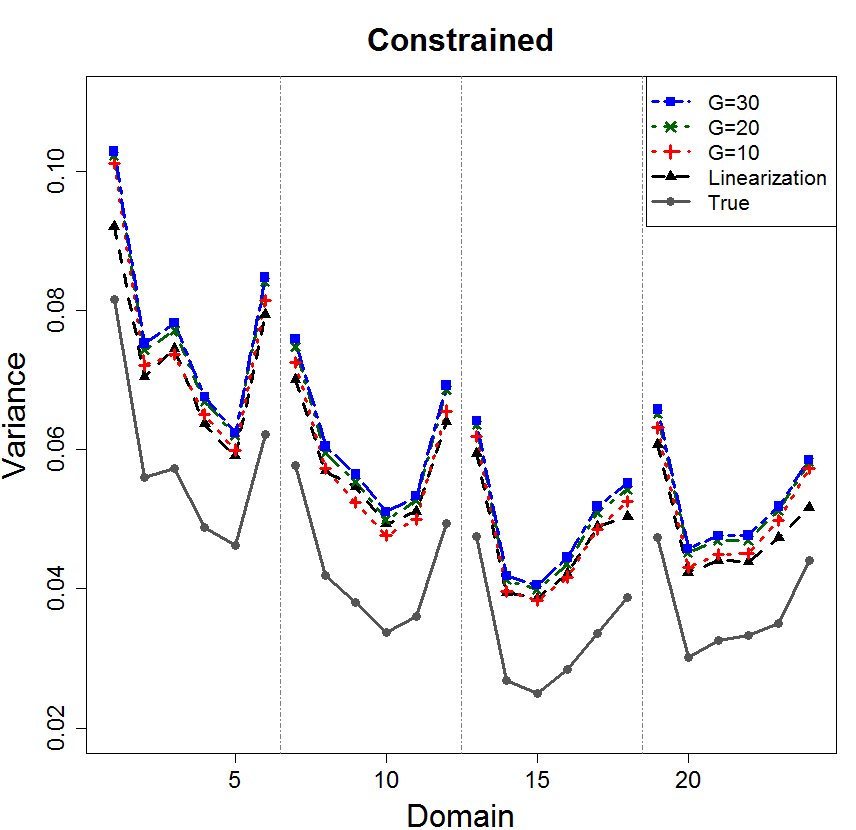}}\\
		\subfigure{\includegraphics[width=.49\textwidth]{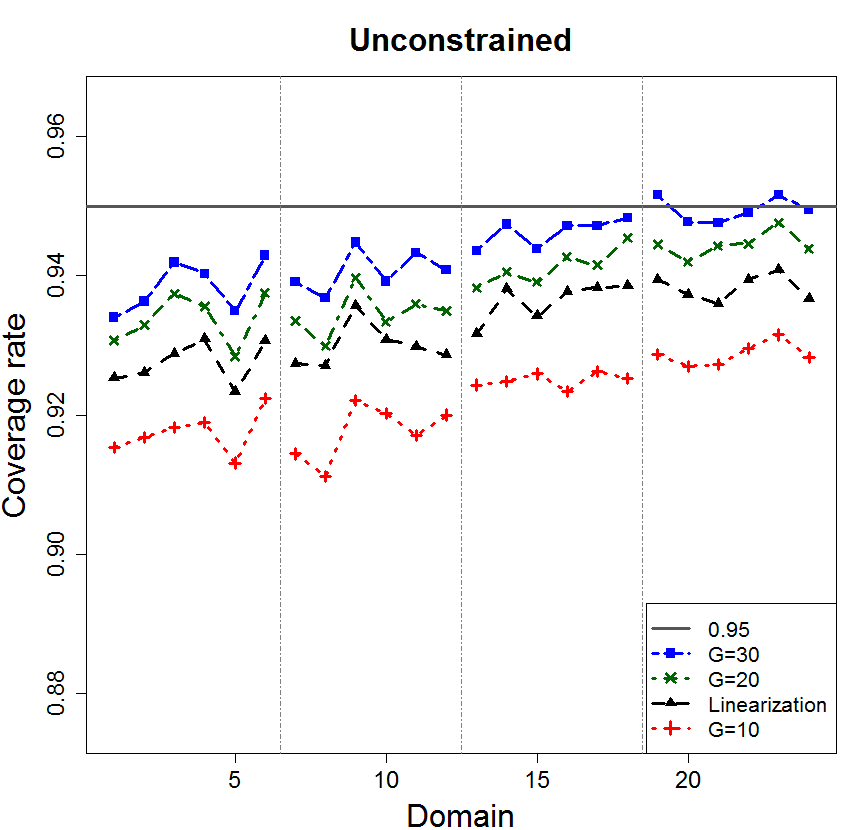}}\
		\subfigure{\includegraphics[width=.49\textwidth]{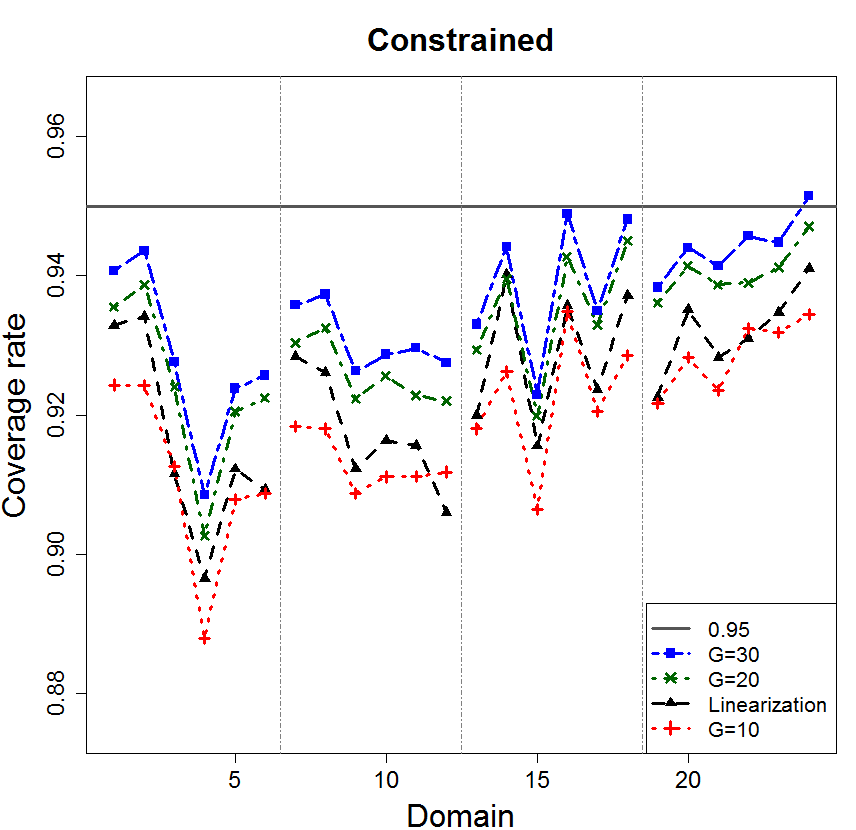}}
		\caption{Variance estimation (top) and coverage rate (bottom) simulation results based on linearization and DAGJK methods for the unconstrained (left) and constrained (right) estimators, under the double monotone scenario with $n_N=960$ and $\sigma=2$.}
		\label{fig:repbasedsigma2sample960}
	\end{center}
\end{figure}

\section{Application of constrained estimator to NSCG 2015} \label{sec:nscg}

To demonstrate the utility of the proposed constrained methodology in real survey data, we consider the 2015 National Survey of College Graduates (NSCG), which is sponsored by the National Center for Science and Engineering Statistics (NCSES) within the National Science Foundation, and is conducted by the U.S. Census Bureau. The 2015 NSCG data and documentation are openly available on the NSF website (\url{www.nsf.gov/statistics/srvygrads}). The purpose of the NSCG is to provide data on the characteristics of U.S. college graduates, with particular focus on those in the science and engineering workforce.

We set the total earned income before deductions in previous year (2014) to be the variable of interest (denoted by EARN). To avoid the high skewness of this variable, a $\log$ transformation is performed. Moreover, we take into account only those who reported a positive earning amount. A total of $76,389$ observations was considered in our analysis. In addition, $252$ domains are considered. These are determined by the cross-classification of four predictor variables. Such variables and their assumed constraints are:

\begin{itemize}
	\item \textbf{Time since highest degree}. This ordinal variable defines the year category of award of highest degree. The period from $2015$ to $1959$ is divided into 9 categories, where the first $8$ categories (denoted by 1-8) are of $6$ years each, and the last category (denoted by 9) is of $9$ years. \textsl{Constraint:} given the other predictors, the average total earned income increases with respect to the time since highest degree from year category 1 to 7. No assumption is made with respect to categories 8 and 9, as those people are likely to be retired (at least 42 years since their highest degree).  
	\item \textbf{Field category}. This nominal variable defines the field of study for highest degree, based on a major group categorization provided within the 2015 NSCG. The $7$ categories for this variable are:
	\begin{enumerate}[1:]
		\item Computer and mathematical sciences,
		\item Biological, agricultural and environmental life sciences,
		\item Physical and related sciences,
		\item Social and related sciences,
		\item Engineering,
		\item S\&E-related fields,
		\item Non-S\&E fields.
	\end{enumerate}
	\textsl{Constraint:} given the other predictors, the average total earned income for each of the fields 2 and 4 is less than for the fields 1, 3 and 5. No assumption is made with respect to categories 6 and 7, as they cover many fields for which a reasonable order restriction might be complicated to impose.
	
	\item \textbf{Postgrad.} This binary variable defines whether the highest degree is of the postgraduate level (YES) or of the Bachelor's level (NO). \textsl{Constraint:} given the other predictors, the average total earned income is higher for those with postgraduate studies.
	
	\item \textbf{Supervise.} This binary variable defines whether supervising others is a responsibility in the principal job (YES) or not (NO). \textsl{Constraint:} given the other predictors, the average total earned income is higher for those who supervise others in their principal job.
\end{itemize}  

Figures \ref{fig:nscgunccon1} and \ref{fig:nscgunccon2} contain the unconstrained and constrained estimates for each of the four groups obtained from the cross-classification of the Postgrad and Supervise binary variables. Note that since the assumed constraints constitute a partial ordering, then the constrained estimates are obtained by pooling domains. These figures show that the constrained estimator has a smoother behavior than the unconstrained. Moreover, it tends to correct for the large spike domains produced by the unconstrained estimator, which are usually a consequence of a very small sample size.

\begin{figure}[ht!] 
	\begin{center}
		\subfigure[Supervise$=$YES (unconstrained).]{\includegraphics[width=.49\textwidth]{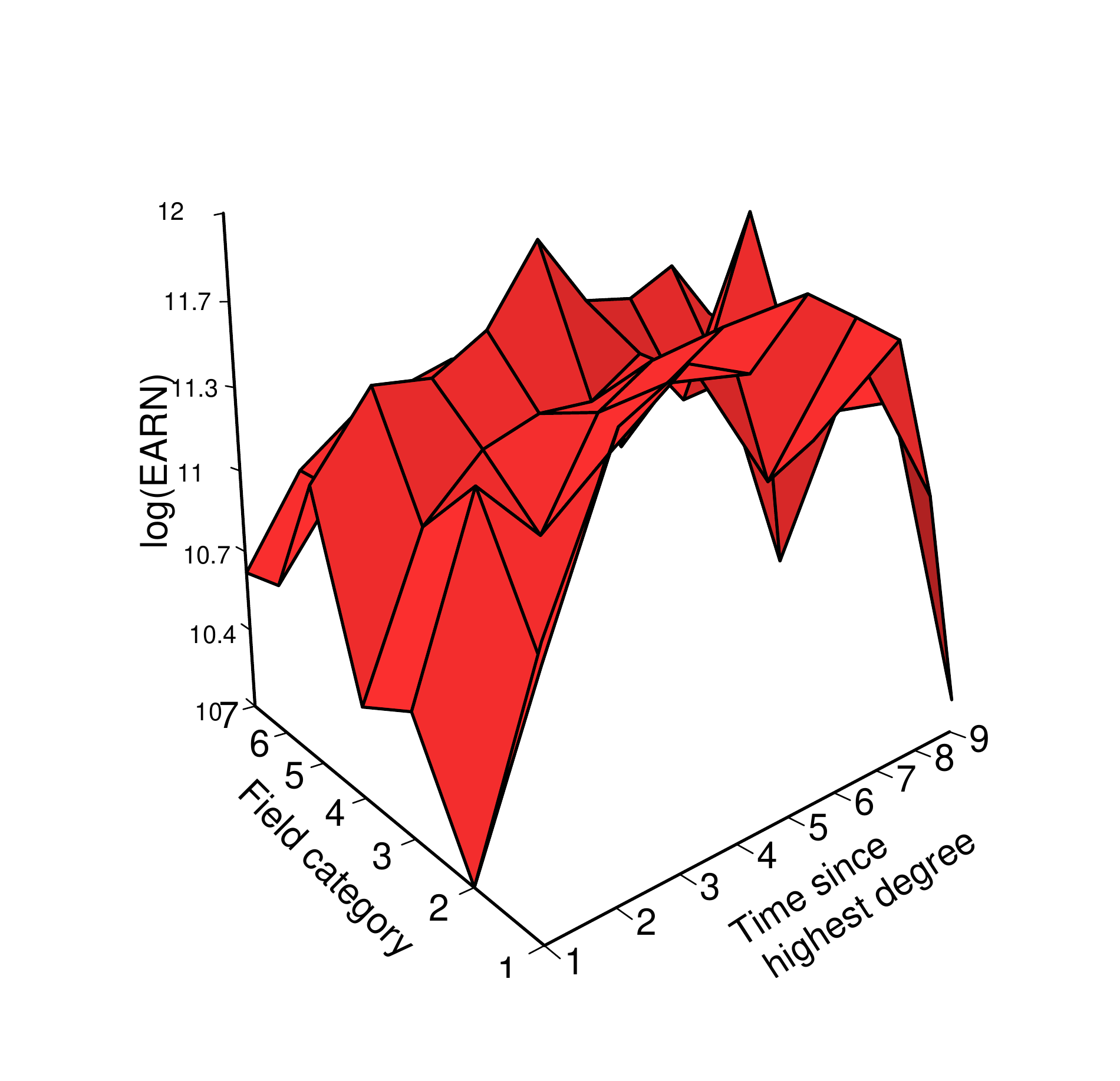}}\
		\subfigure[Supervise$=$YES (constrained).]{\includegraphics[width=.49\textwidth]{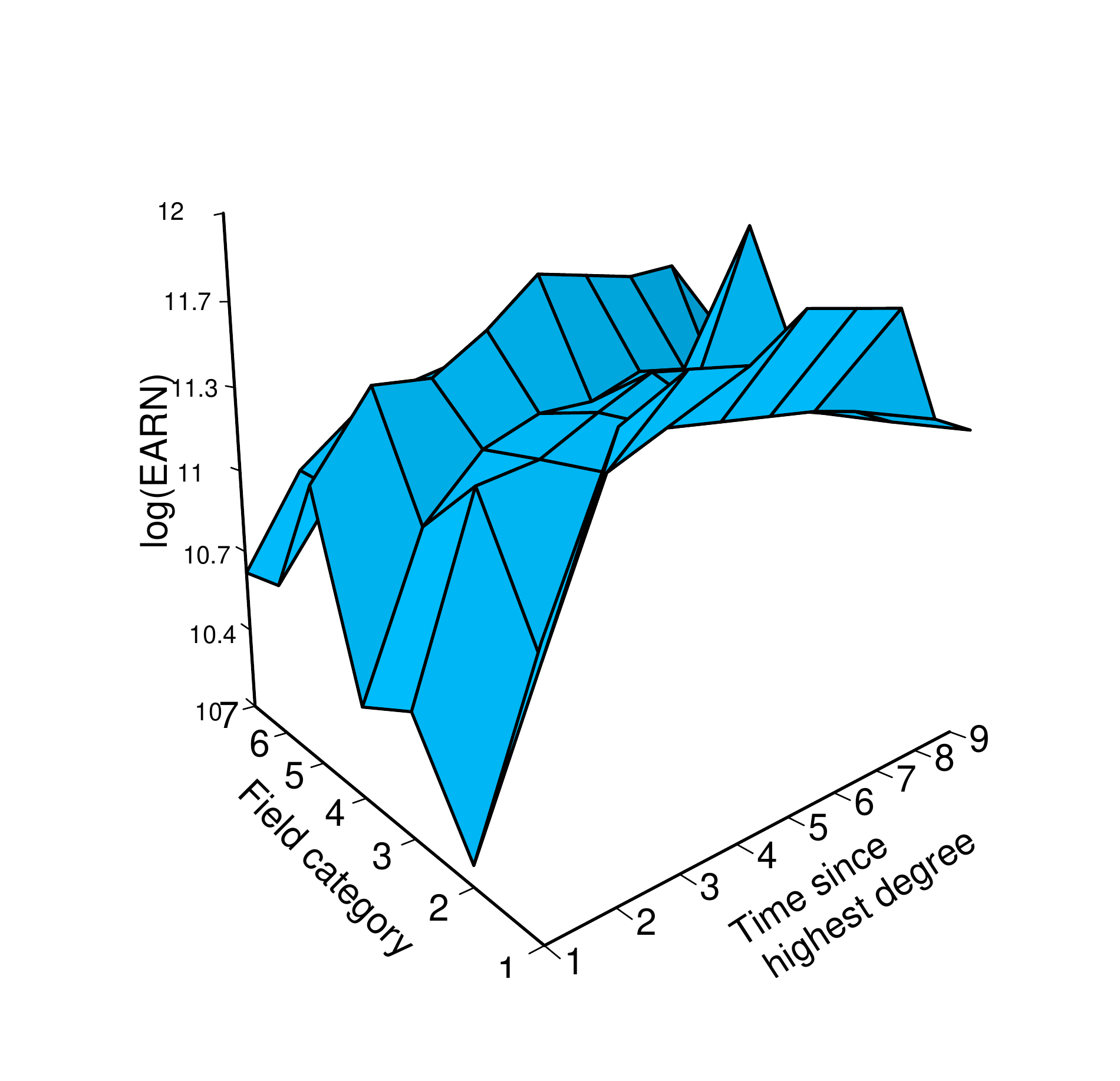}}\\
		\subfigure[Supervise$=$NO (unconstrained).]{\includegraphics[width=.49\textwidth]{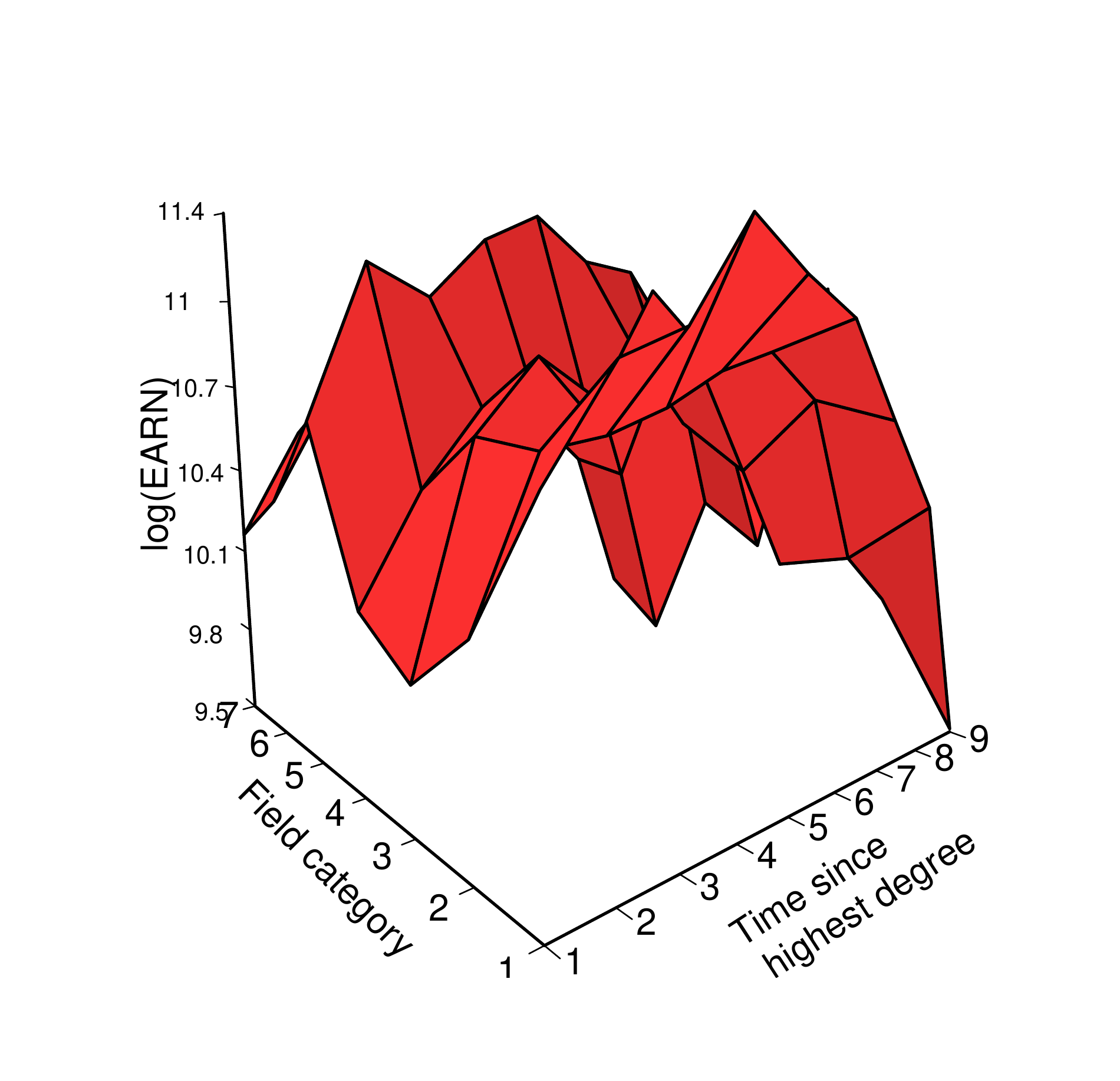}}\
		\subfigure[Supervise$=$NO (constrained).]{\includegraphics[width=.49\textwidth]{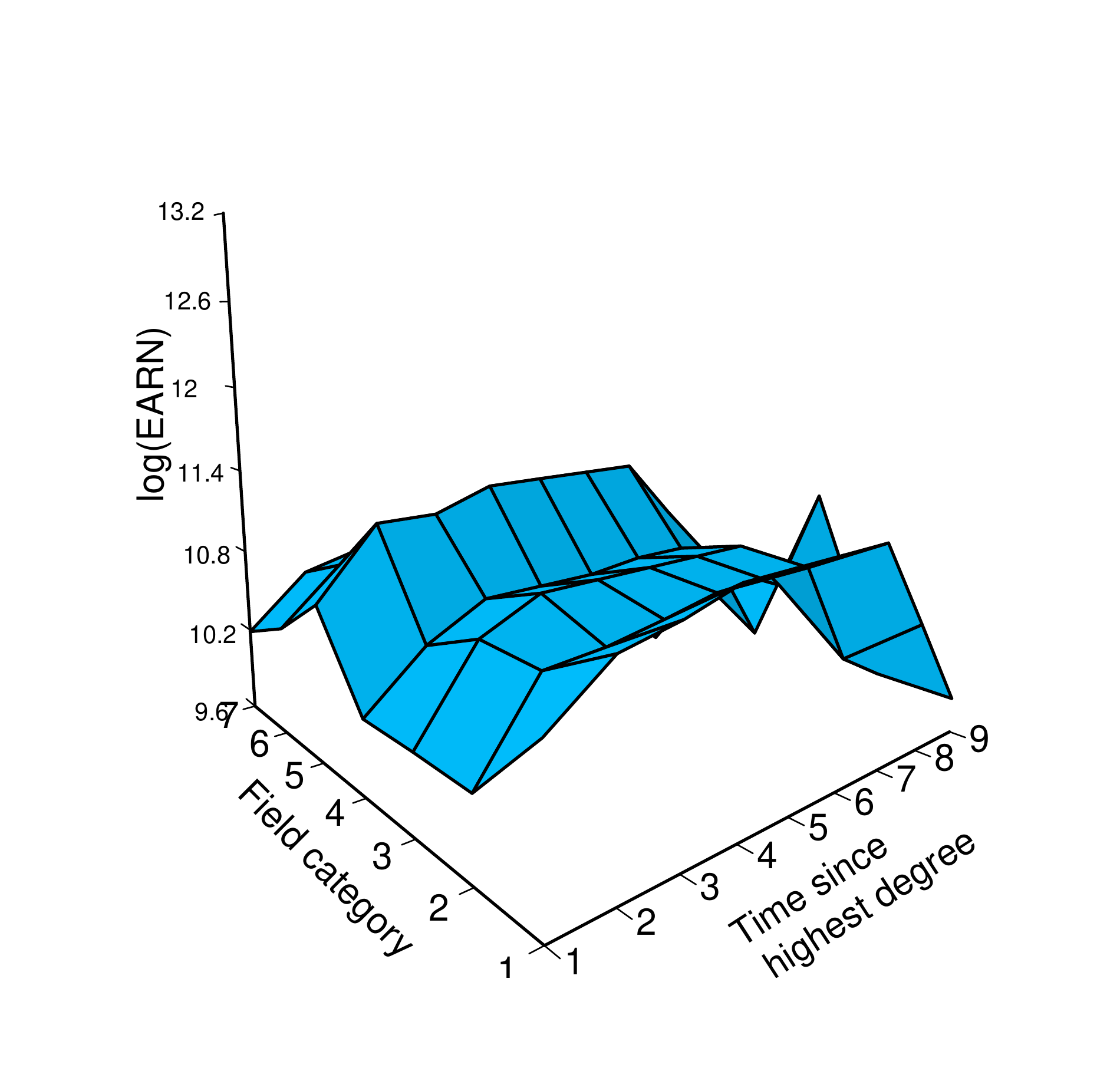}}
		\caption{Unconstrained (left) and constrained (right) domain mean estimates for the 2015 NSCG data, given that Postgrad$=$NO is fixed.}
		\label{fig:nscgunccon1}
	\end{center}
\end{figure} 

\begin{figure}[ht!] 
	\begin{center}
		\subfigure[Supervise$=$YES (unconstrained).]{\includegraphics[width=.49\textwidth]{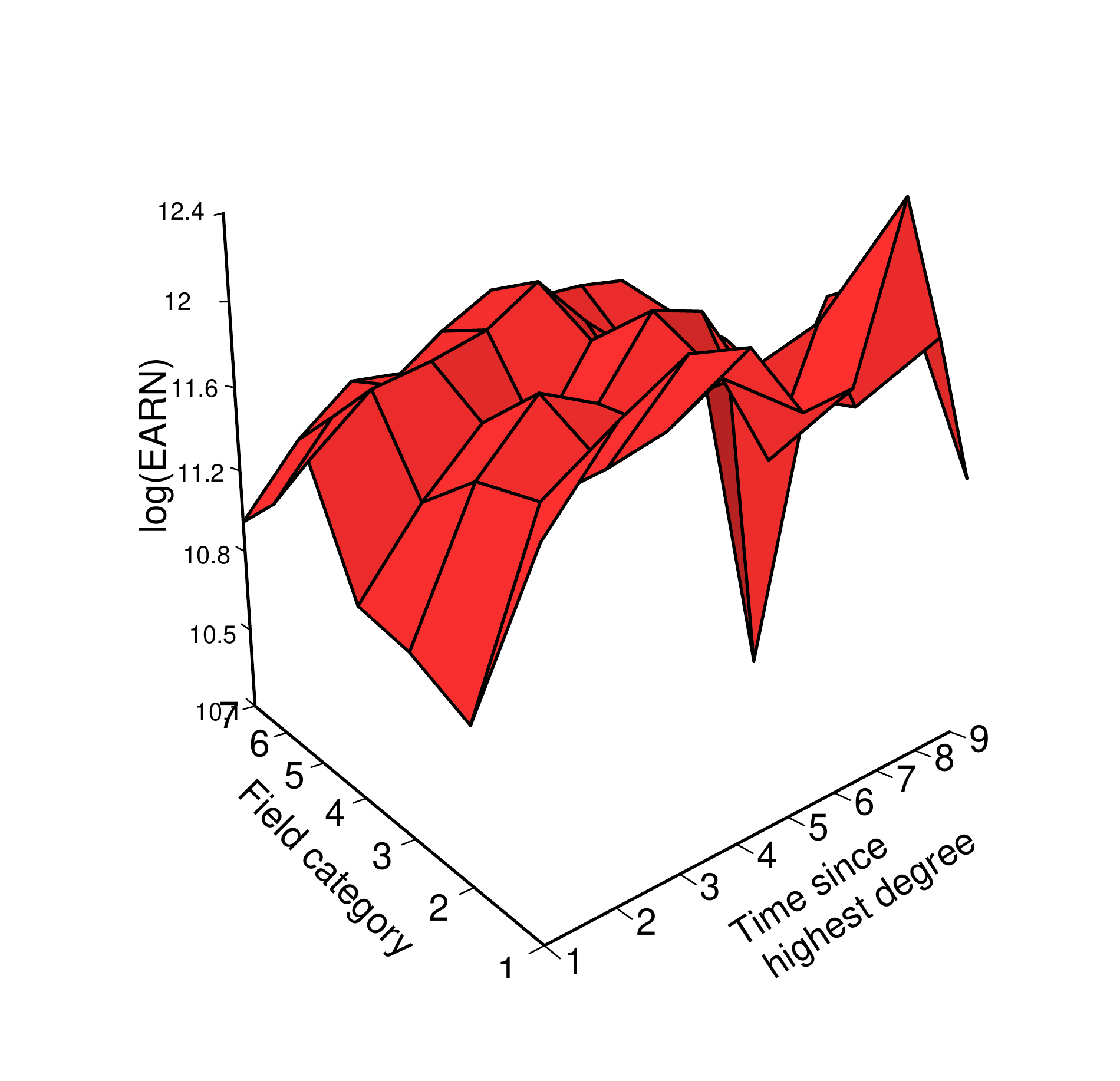}}\
		\subfigure[Supervise$=$YES (constrained).]{\includegraphics[width=.49\textwidth]{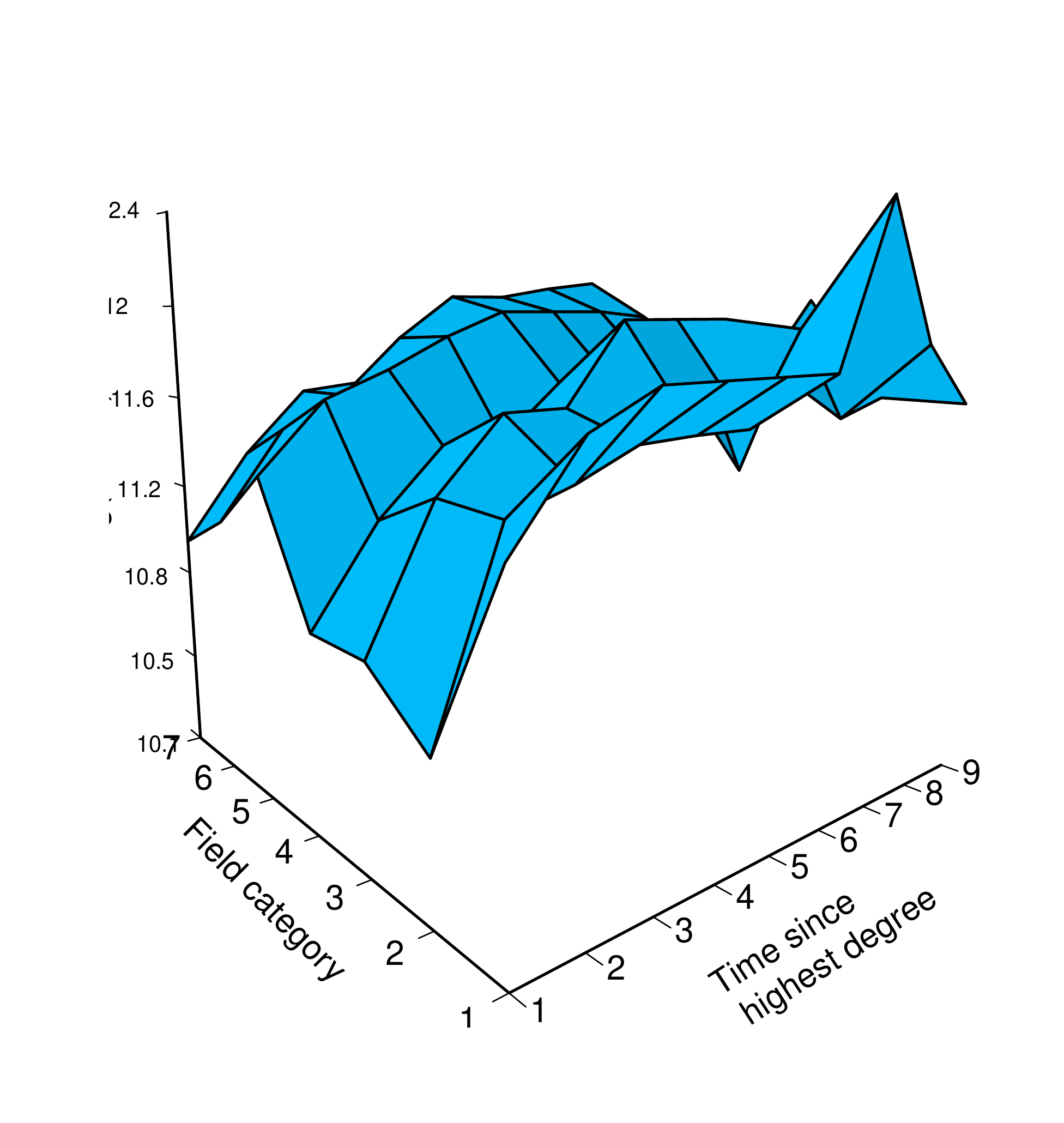}}\\
		\subfigure[Supervise$=$NO (unconstrained).]{\includegraphics[width=.49\textwidth]{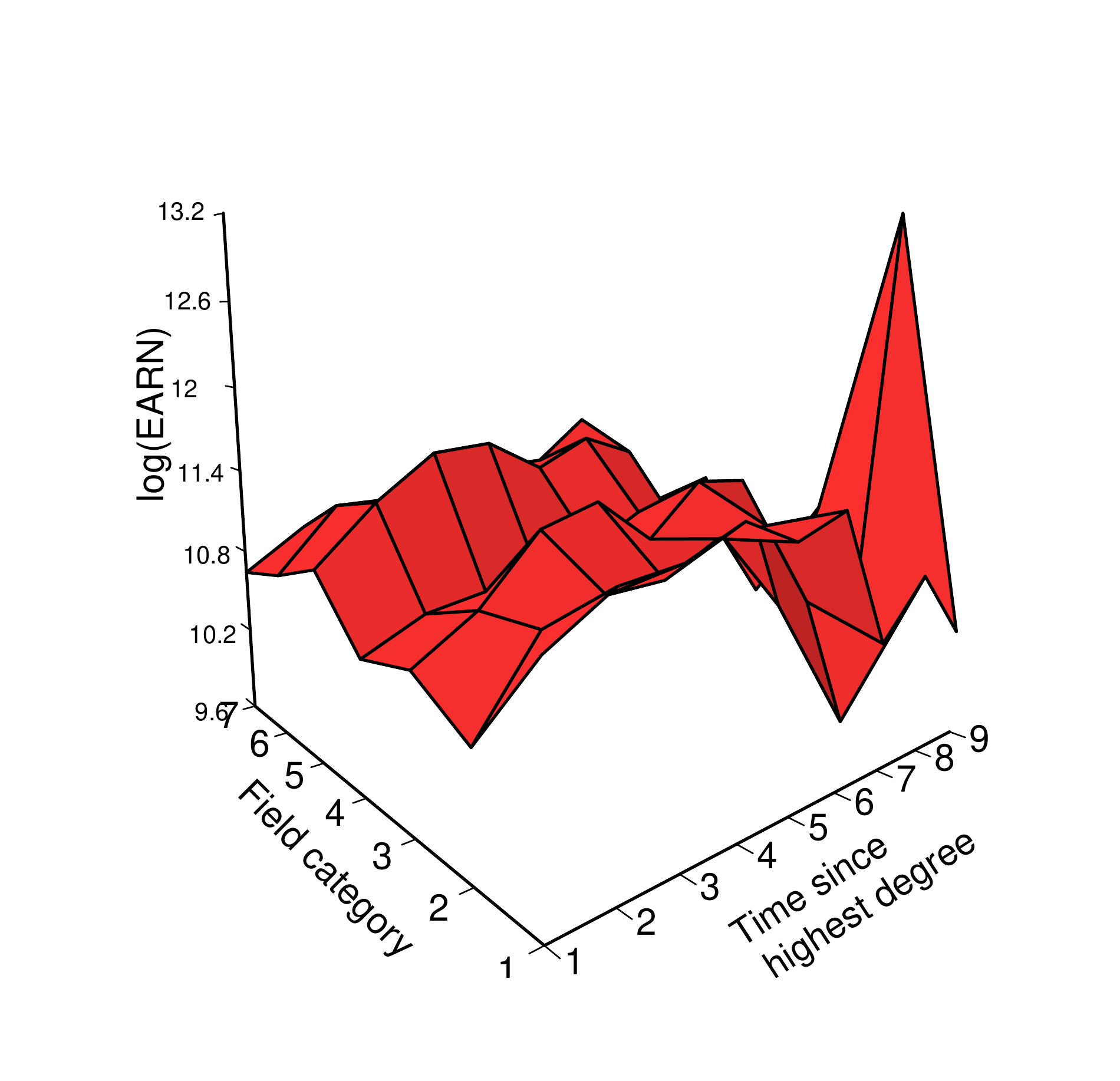}}\
		\subfigure[Supervise$=$NO (constrained).]{\includegraphics[width=.49\textwidth]{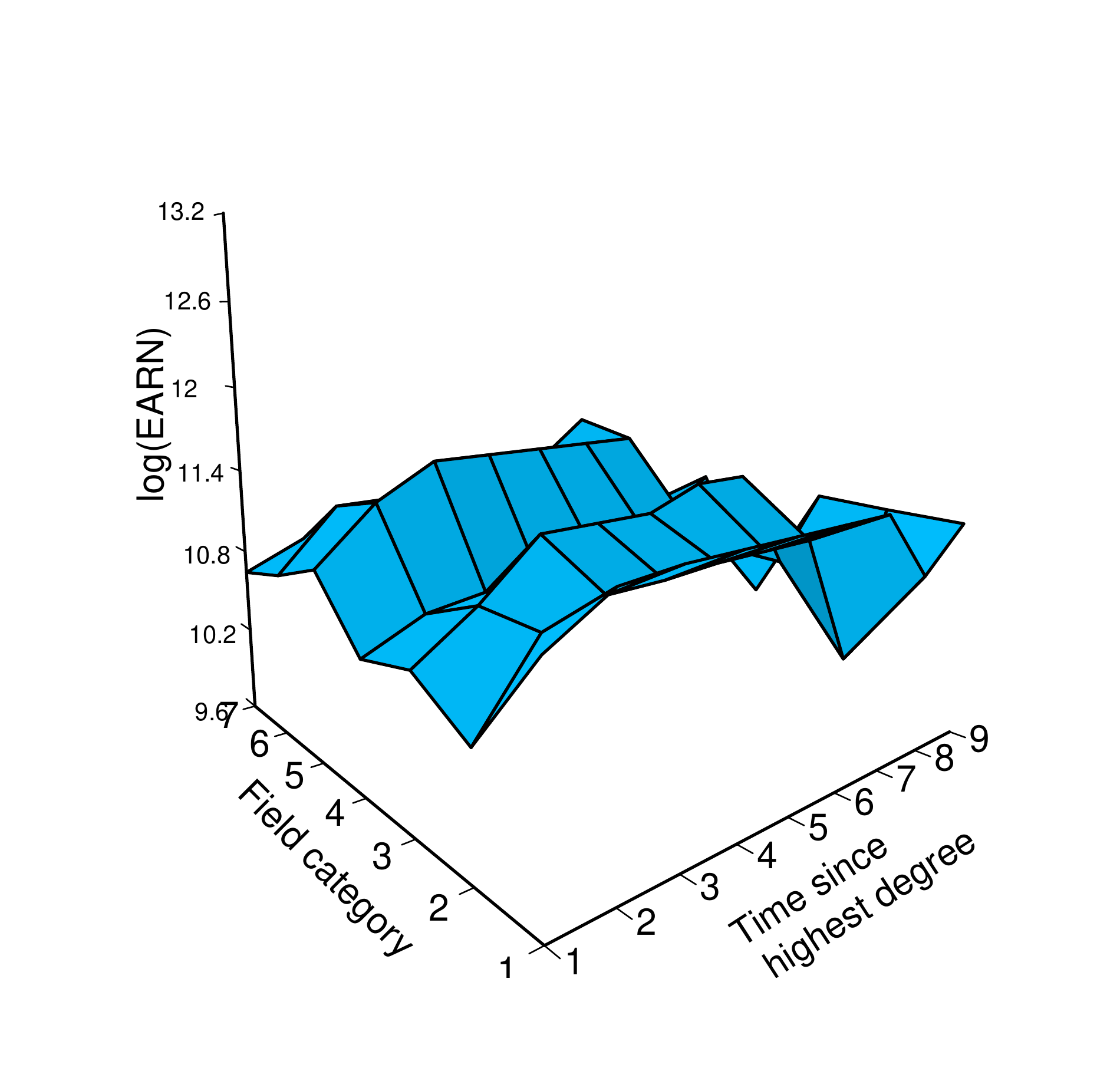}}
		\caption{Unconstrained (left) and constrained (right) domain mean estimates for the 2015 NSCG data, given that Postgrad=YES is fixed.}
		\label{fig:nscgunccon2}
	\end{center}
\end{figure} 

Standard errors for both unconstrained and constrained estimates are computed using the 2015 NSCG replicate weights, which are based on Successive Difference \citep{opsomer16} and Jackknife replication methods. Both the replicate weights and adjustment factors were provided by the Program Director of the Human Resources Statistics Program from the NCSES and are available upon request. 

Figure \ref{fig:nscg_sdratio} displays the ratio of these estimates for each of the $252$ domains. Note that in the vast majority of cases, the standard error estimates of the proposed estimator are lower than those for the unconstrained estimator, with improvements of as much as 7 times smaller. However, there are some cases where the opposite behavior occurs. These are explored in Figure \ref{fig:nscg_slices}, which shows plots of two different slices: one with respect to the Time since highest degree variable and other with respect to Field category. These plots include unconstrained and constrained estimates, Wald confidence intervals and sample sizes. Further, each of these two slices contain one of the two domains that can be easily identified in Figure \ref{fig:nscg_sdratio} to have the smallest ratios. The first of these domains is displayed in Figure \ref{fig:nscg_slices_1}, indexed by $5$. Here, the confidence interval is narrower for unconstrained estimates, which is as a direct consequence of having smaller standard deviation estimates. Note that the unconstrained estimates for the domains indexed by $5$ and $6$ violate the monotonicity assumption, and thus, are being pooled to obtain the constrained estimates. In contrast, Figure \ref{fig:nscg_slices_3} shows that the samples sizes on these domains are considerably large, meaning that the noticed violation might be in fact true. Therefore, as the imposed restrictions are enforcing these two domains to get pooled, then domain indexed by $5$ ends up producing a larger standard deviation on its constrained estimate. 
The second domain where unconstrained estimates produce smaller standard deviation estimates is displayed in Figure \ref{fig:nscg_slices_2}, indexed by $1$. Here, this domain is being pooled with its consecutive domain to obtain the constrained estimate. However, as these two domains have very low sample sizes (Figure \ref{fig:nscg_slices_4}), they produce a constrained estimate that is based on a very small `effective' sample size. Therefore, both the unconstrained and constrained estimates might be considered as unreliable, given the small sample circumstances.

\begin{figure}[ht!] 
	\begin{center}
		\subfigure{\includegraphics[width=.50\textwidth]{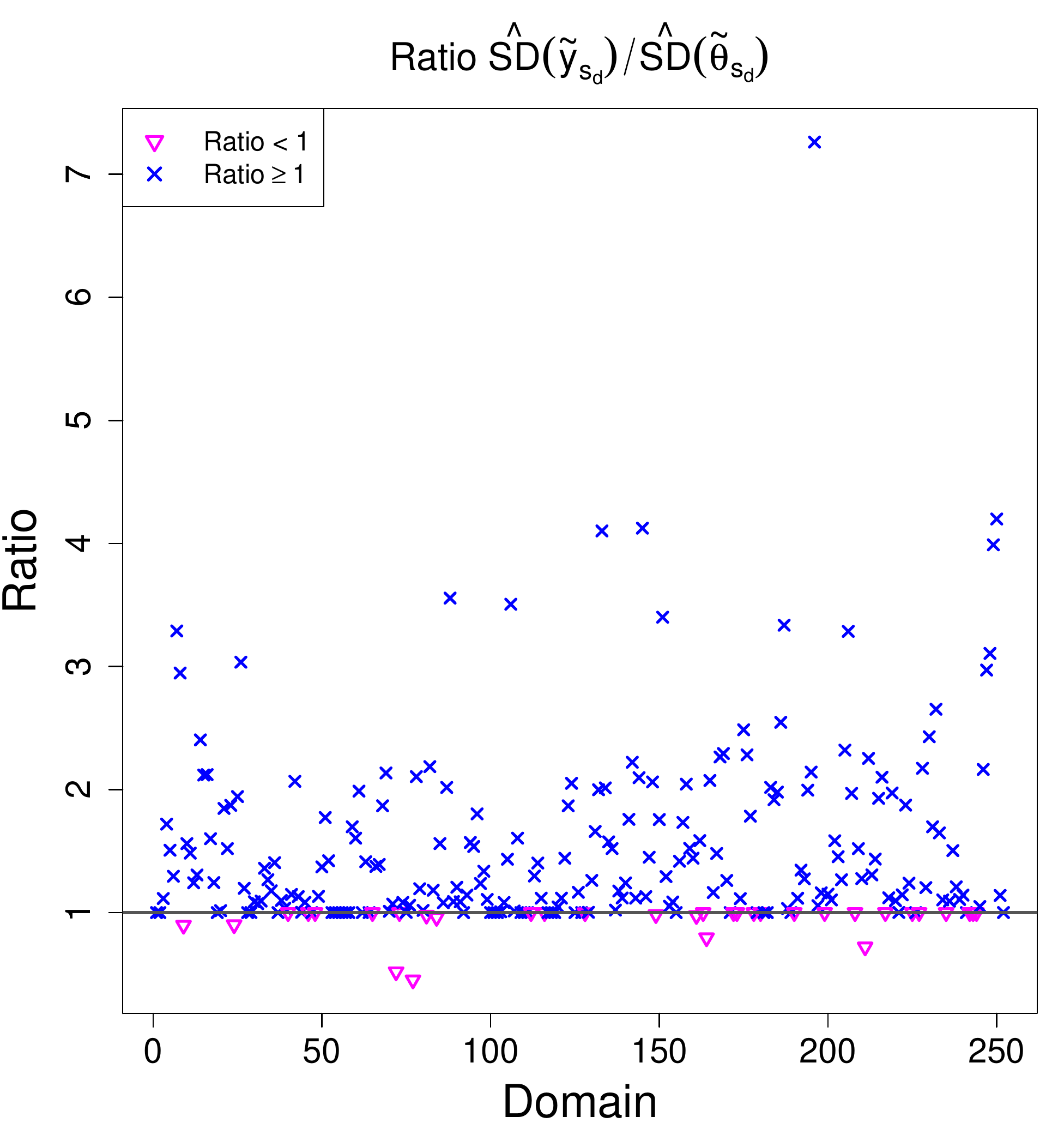}}
		\caption{Ratio of the estimated standard errors of unconstrained estimates over those for constrained estimates for the 2015 NSCG data.}
		\label{fig:nscg_sdratio}
	\end{center}
\end{figure} 

\begin{figure}[ht!] 
	\begin{center}
		\subfigure[Field category=2.]{\label{fig:nscg_slices_1}\includegraphics[width=.49\textwidth]{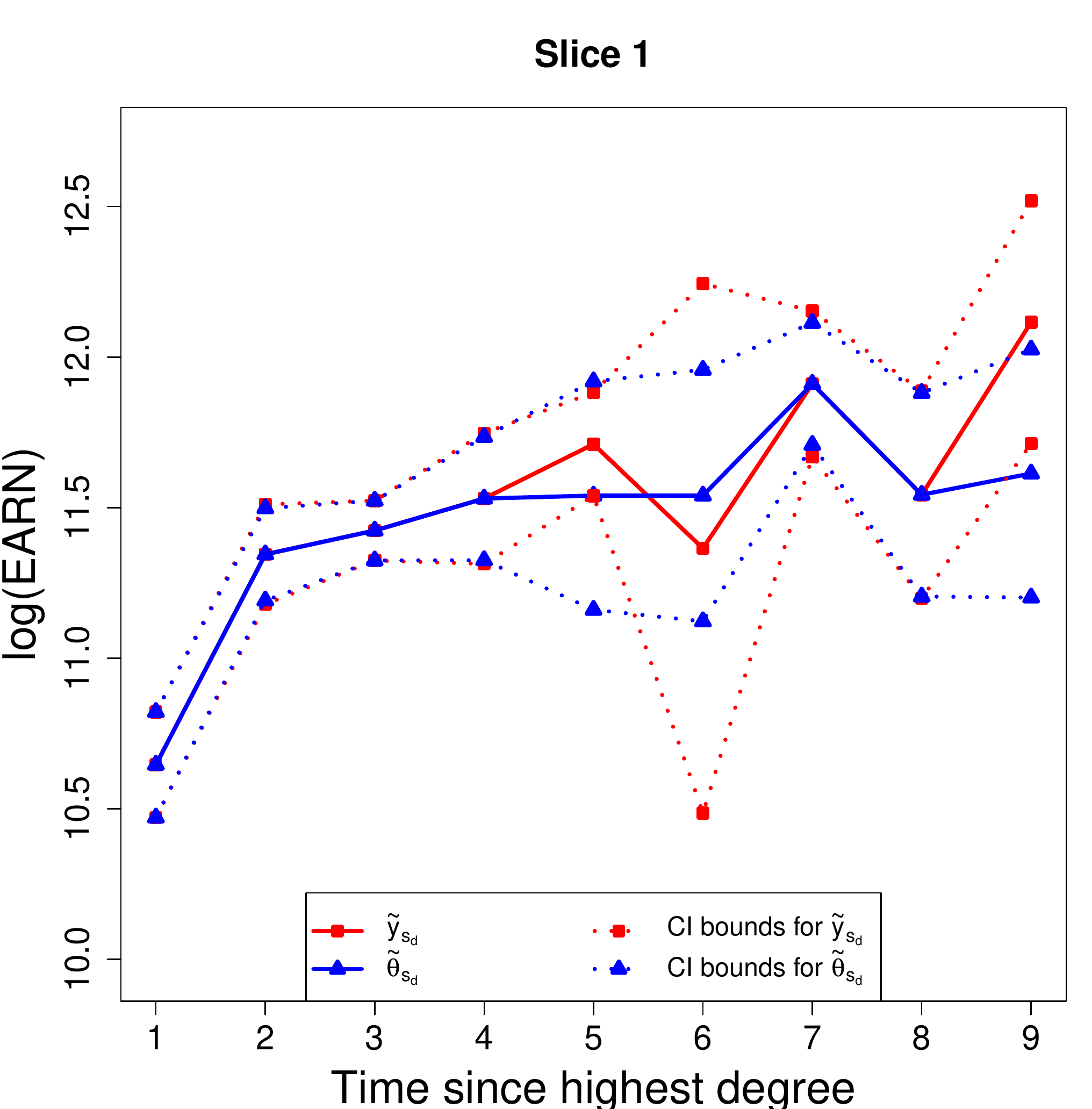}}\
		\subfigure[Time since highest degree=9.]{\label{fig:nscg_slices_2}\includegraphics[width=.49\textwidth]{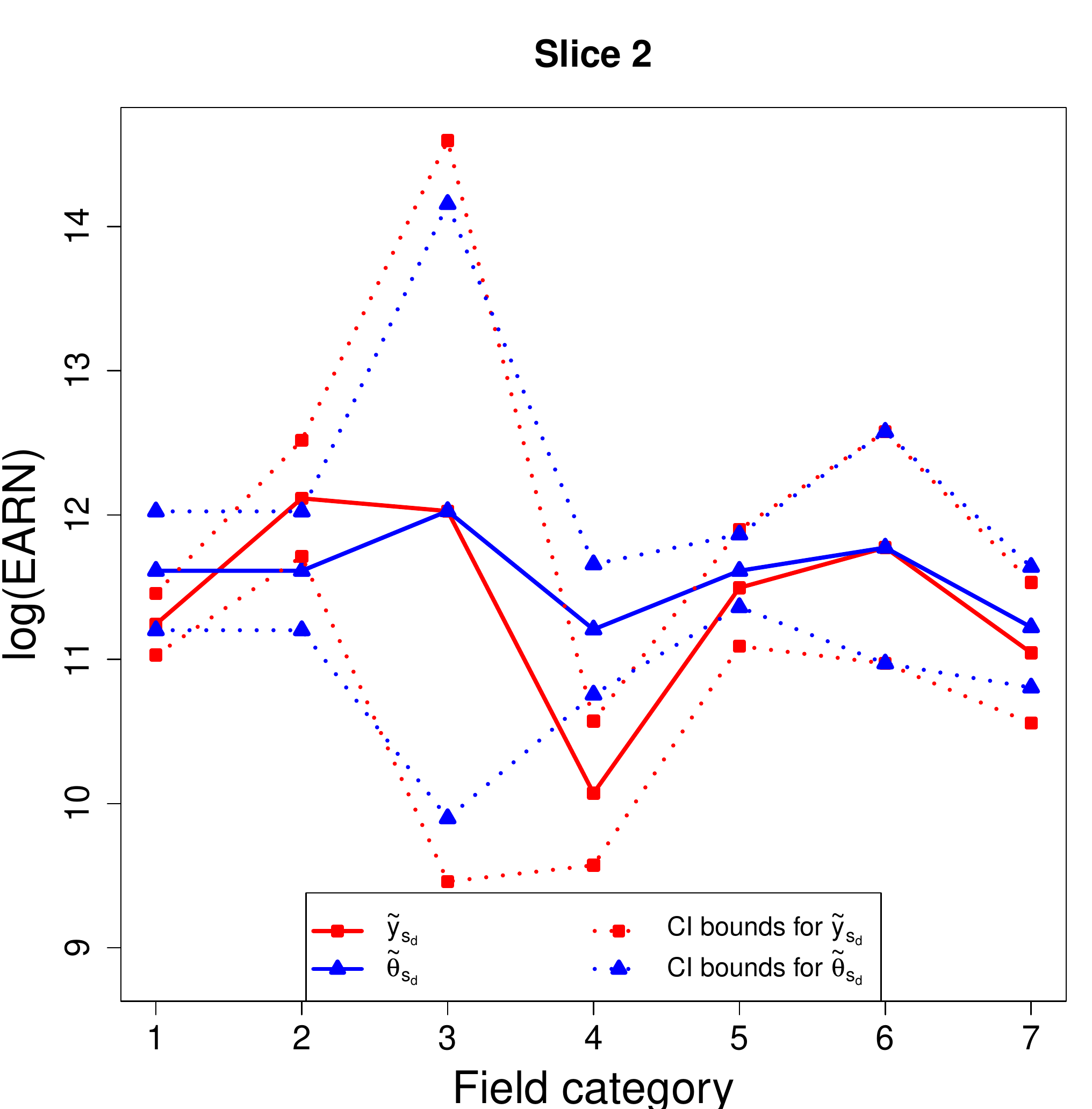}} \\
		\subfigure[Field category=2.]{\label{fig:nscg_slices_3}\includegraphics[width=.49\textwidth]{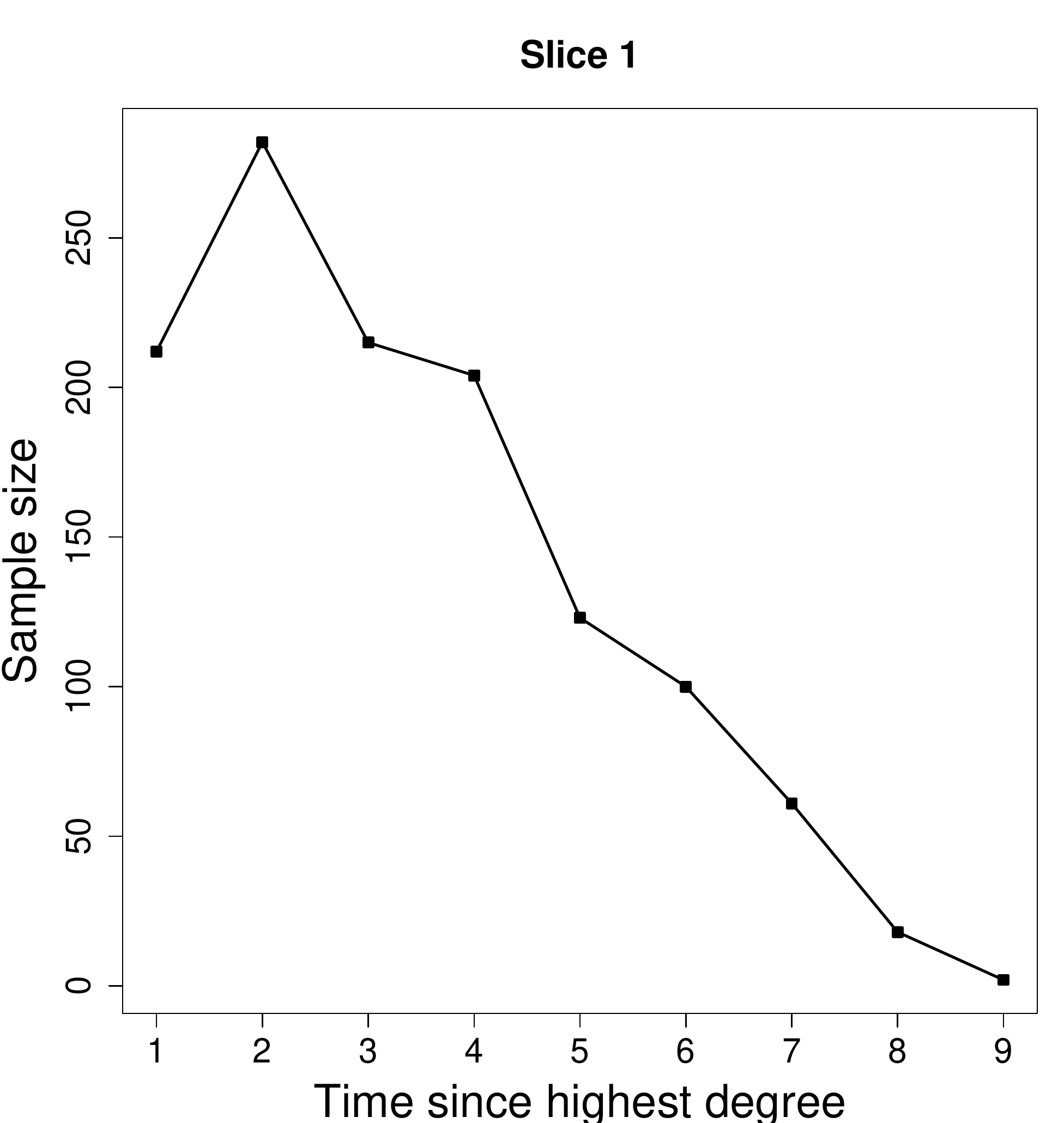}}\
		\subfigure[Time since highest degree=9.]{\label{fig:nscg_slices_4}\includegraphics[width=.49\textwidth]{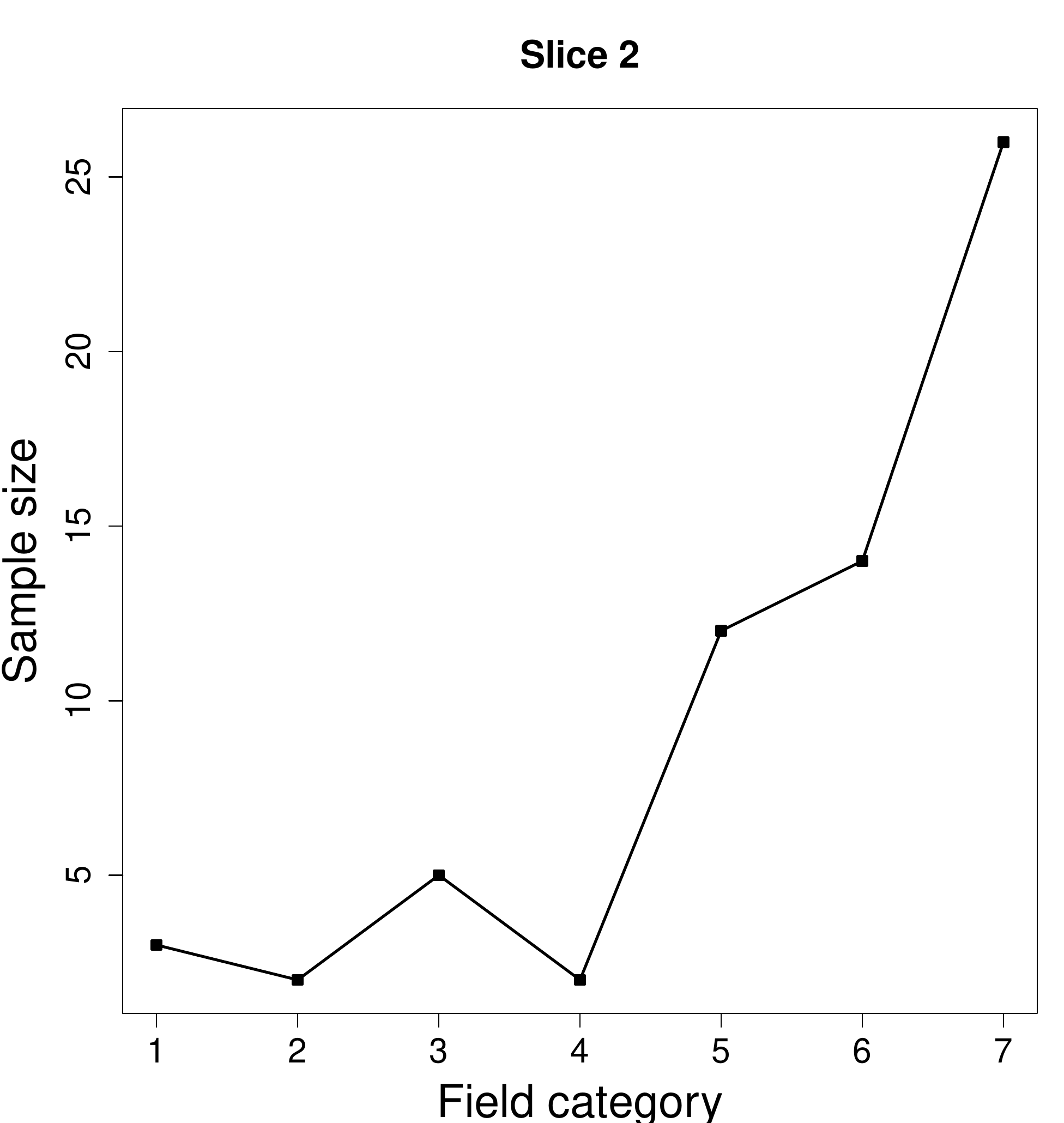}} 
		\caption{Slice plots of: unconstrained and constrained estimates with Wald confidence intervals (top) and sample sizes (bottom) for the 2015 NSCG data, given that Postgrad=YES and Supervise=YES.}
		\label{fig:nscg_slices}
	\end{center}
\end{figure}

\section{Conclusions} \label{sec:conclusions}

We proposed a methodology to estimate domain means which takes into account both design-based estimators and reasonable shape restrictions, and it was shown to largely improve their estimation and variability, especially on small domains. As this new methodology covers a broad range of shape assumptions beyond univariate monotonicity, it aims to jointly take advantage of several types of qualitative information that arises naturally for survey data. We also proposed a design-based variance estimation method of the estimator. However, as this method depends solely on the set $J$ that represents the linear space where the constrained estimator lands, then it tends to overestimate the variance. Replication-based methods are shown to behave similarly. Hence, further research might be carried out to develop variance estimation methods that do not ignore the randomness associated to the set $J$. From the computational side, it is based on the Cone Projection Algorithm which is efficiently implemented in the package \texttt{coneproj}. Thus, it is presented as an easy-to-implement attractive alternative for small area estimation.

We identify some possible direct implications of our proposed methodology. For cases of missing data, our methodology has the potential of `bounding' domains with no observations, which will provide some (instead of none) knowledge regard those domain means. Further, if population-level information is available, then a model-assisted based estimator that makes use of our proposed methodology could be developed. Under partial orderings, such estimator would be equivalent to a poststratified estimator, which uses the sample-selected pooling as the post strata.

Among some research extensions of interest, monotone restrictions might be relaxed, as these can be very strict assumptions for certain populations. Further, shape selection tools for survey data might be developed. As an immediate consequence of these tools, data-driven methods that selects the most appropriate amount of relaxation will be available. In addition, the presented methodology may be adapted to allow for covariates, leading to the development of methods that are analogous to partial linear additive models with shape restrictions.


\bibliography{CristianOliva_bib}{}
\bibliographystyle{apalike}

\appendix

\section{Appendix} \label{sec:appendix2}

The first part of this appendix contains all lemmas (with proofs) used to prove the theoretical results discussed in this paper. Complete proofs of these results are included at the end of this appendix. The proof of Lemma \ref{lem:lincomb} can be also found in \citet[Ch. 1]{fenchel53}.

\begin{lemma} \label{lem:lincomb}
	If a non-zero vector can be written as the positive linear combination of  linearly dependent vectors, then it can be expressed as the positive linear combination of a linearly independent subset of these.   
\end{lemma}
\begin{proof}
	Let $\boldsymbol{v}$ be a non-zero vector such that it can be written as $\boldsymbol{v}=\sum_{i=1}^{k}a_i\boldsymbol{\ell_i}$; where $a_i>0$ for $i=1,2,\dots, k$, and $\{ \boldsymbol{\ell}_1, \boldsymbol{\ell}_2, \dots, \boldsymbol{\ell}_k \}$ is a set of linearly dependent vectors. Since this set of vectors is not linearly independent, then there exists constants $b_k$ (not all different than zero)  such that $\sum_{i=1}^{k}b_i \boldsymbol{\ell}_i =\boldsymbol{0}$. Without loss of generality, assume that there is at least one $b_i$ that is positive. Now, let $I_0$ be the set of indexes given by
	\begin{equation*}
		I_0=\underset{i \; : \; b_i>0 }{\arg \min } \frac{a_i}{b_i}.
	\end{equation*}
	Note that $I_0$ cannot contain all indexes $\{1,2,\dots, k\}$ because $\boldsymbol{v}$ is a non-zero vector. Hence, for any index $i_0 \in I_0$,
	\begin{equation*}
	\boldsymbol{v}=\sum_{i=1}^{k} \left(a_i-\frac{a_{i_0}}{b_{i_0}}b_i \right)\boldsymbol{\ell}_i=\sum_{i\notin I_0} \left(a_i-\frac{a_{i_0}}{b_{i_0}}b_i \right)\boldsymbol{\ell}_i
	\end{equation*}
	which means that the vector $\boldsymbol{v}$ can be also written as a positive linear combination of a proper subset of $\{ \boldsymbol{\ell}_1, \boldsymbol{\ell}_2, \dots, \boldsymbol{\ell}_k \}$. Finally, note that we can repeat the above argument until it is not possible to find constants $b_i\neq 0$ such that $\sum_{i}b_i \boldsymbol{\ell}_i =\boldsymbol{0}$. Thus, the resulting subset of vectors of $\{ \boldsymbol{\ell}_1, \boldsymbol{\ell}_2, \dots, \boldsymbol{\ell}_k \}$ has to be linearly independent, and $\boldsymbol{v}$ can be written as a positive linear combination of them. 
\end{proof}

\begin{lemma} \label{lem:irred}
	If $\boldsymbol{A}$ is a $m \times D$ irreducible matrix and $\boldsymbol{S}$ is a $D\times D$ diagonal matrix, then $\boldsymbol{AS}$ is also irreducible. 
\end{lemma}
\begin{proof}
	This is an immediate result derived from the fact that $\boldsymbol{S}$ is non-singular.
\end{proof}

\begin{lemma} \label{lem:linequiv}
	Let $\boldsymbol{A}$ be a $m \times D$ matrix. Also, let $\boldsymbol{S}_1$ and $\boldsymbol{S}_2$ be $D \times D$ diagonal matrices. For any set $J \subseteq \{1,2,\dots, m\}$, denote $V_{i,J}$ to be the set of vectors in rows $J$ of $\boldsymbol{A}_{i}=\boldsymbol{AS}_i$, $i=1,2$. Then, for any $J^* \subseteq J$,
	\begin{equation*} 
		\mathcal{L}(V_{1,J^*}) = \mathcal{L}(V_{1,J}) \iff \mathcal{L}(V_{2,J^*}) = \mathcal{L}(V_{2,J}).
	\end{equation*}
\end{lemma}
\begin{proof}
	Let $\boldsymbol{A}_{i,J}=\boldsymbol{A}_J\boldsymbol{S}_i$, $i=1,2$; where $\boldsymbol{A}_J$ denotes the submatrix of $\boldsymbol{A}$ that contains the rows in positions $J$. First, assume that $\mathcal{L}(V_{1,J^*}) = \mathcal{L}(V_{1,J})$. Since $J^* \subseteq J$, it is straightforward that $\mathcal{L}(V_{2,J^*}) \subseteq \mathcal{L}(V_{2,J})$. Now, consider any $\boldsymbol{v} \in \mathcal{L}(V_{2,J})$. Hence, $\boldsymbol{v}=\boldsymbol{A}_{2,J}^{\top}\boldsymbol{a}=\boldsymbol{S}_2\boldsymbol{A}_{J}^{\top}\boldsymbol{a}$ for some vector $\boldsymbol{a}$. Then, we have $\boldsymbol{S}_1\boldsymbol{S}_2^{-1}\boldsymbol{v}=\boldsymbol{S}_1\boldsymbol{A}_{J}^{\top}\boldsymbol{a} \in \mathcal{L}(V_{1,J})$. By assumption, there exists a vector $\boldsymbol{b}$ such that $\boldsymbol{S}_1\boldsymbol{S}_2^{-1}\boldsymbol{v}=\boldsymbol{S}_1\boldsymbol{A}_{J^*}^{\top}\boldsymbol{b}$. Therefore, $\boldsymbol{v}=\boldsymbol{S}_2\boldsymbol{A}_{J^*}^{\top}\boldsymbol{b} \in \mathcal{L}(V_{2,J^*})$. Thus, $\mathcal{L}(V_{2,J}) \subseteq \mathcal{L}(V_{2,J^*})$. Analogously, we can prove that $\mathcal{L}(V_{2,J^*}) = \mathcal{L}(V_{2,J})$ implies $\mathcal{L}(V_{1,J^*}) = \mathcal{L}(V_{1,J})$.	
\end{proof}

\begin{lemma} \label{lem:bounded}
	Under Assumptions A1-A5, then:
	\begin{enumerate}[(i)]
		\item The $N^{-1}\widehat{t}_d$ are uniformly bounded in $s_N$.
		\item The $N^{-1}\widehat{N}_d$ are uniformly bounded above and uniformly bounded away from zero in $s_N$.
		\item var$(N^{-1}\widehat{t}_d)=O(n_N^{-1})$ and var$(N^{-1}\widehat{N}_d)=O(n_N^{-1})$
		\item $\mathbb{E}[(N^{-1}\widehat{t}_d-r_d\mu_d)^2]=O(n_N^{-1})$ and $\mathbb{E}[(N^{-1}\widehat{N}_d-r_d)^2]=O(n_N^{-1})$.
	\end{enumerate}
\end{lemma}
\begin{proof}
	\begin{enumerate}[(i)]
		\item Note that
		\begin{equation*}
		\frac{|\widehat{t}_d|}{N}=\left| \frac{\sum_{k \in s_d} y_k/\pi_k}{N} \right| \leq \frac{\sum_{k \in U} |y_k|}{\lambda N}
		\end{equation*}
		which does not depend on $s_N$, and is bounded independently of $N$ by Assumption A2. 
		\item From Assumptions A4 and A5, note that 
		\begin{equation*}
		\frac{\epsilon n_N}{DN} \leq \frac{n_d}{N} \leq \frac{\widehat{N}_d}{N}= N^{-1}\sum_{k \in s_d} 1/\pi_k  \leq \lambda^{-1} N^{-1}N_d \leq \lambda^{-1},
		\end{equation*}
		where both lower and upper bounds do not depend on $s_N$, and are bounded for all $N$ by Assumptions A1, A2 and A4.
		\item Note that		
		\begin{equation*}
		n_N\text{var}(N^{-1}\widehat{t}_d)=n_N\text{var}\left(N^{-1}\sum_{k \in s_d} y_k/\pi_k\right) \leq \frac{\sum_{k \in U_d}y_k^2}{\lambda^2 N}\left(  \frac{n_N}{N}+n_N\max_{k,l \in U_d: k\neq l } |\Delta_{kl}| \right)
		\end{equation*}
		which is bounded by Assumptions A2, A4 and A5. Setting $y_k \equiv 1$ and following an analogous argument, it can be shown that $n_N\text{var}(N^{-1}\widehat{q}_d )=O(1)$.
		\item Since
		\begin{equation*}
		\mathbb{E}\left[\left(N^{-1}\widehat{t}_d-r_d\mu_d\right)^2\right] = \text{var}\left(N^{-1}\widehat{t}_d\right) + \left(\frac{N_d}{N}\overline{y}_{U_d} - r_d\mu_d\right)^2,
		\end{equation*}
		then Assumption A3 and (iii) lead to the desired conclusion. Analogously, $\mathbb{E}\left[\left(N^{-1}\widehat{q}_d-r_d\right)^2\right]=O(n_N^{-1})$.
	\end{enumerate} 
\end{proof}

\begin{proof}[Proof of Theorem \ref{theo:linspace}.]
	First, suppose that $\Pi(\boldsymbol{z} | \Omega^0)=\Pi(\boldsymbol{z} | \mathcal{L}(V_J))=\boldsymbol{0}$. In that case, any subset $J^* \subset J$
	such that $V_{J}$ is linearly independent will satisfy $\Pi(\boldsymbol{z} | \mathcal{L}(V_{J^*}))=\boldsymbol{0} \in \overline{\mathcal{F}}_{J*}$. Hence, it is enough to choose $J^* \subset J$ such that $V_{J^*}$ is linearly independent and spans $\mathcal{L}(V_J)$. Now, suppose that $\Pi(\boldsymbol{z} | \Omega^0) \neq \boldsymbol{0}$. Since $\Pi(\boldsymbol{z} | \Omega^0)=\Pi(\boldsymbol{z} | \mathcal{L}(V_J)) \in \mathcal{\overline{F}}_{J}$, then $\Pi(\boldsymbol{z}| \mathcal{L}(V_J))$ can be written as the positive linear combination of vectors $\boldsymbol{\gamma}_j$, $j \in J$. Moreover, $\langle \boldsymbol{z}-\Pi(\boldsymbol{z}| \mathcal{L}(V_J)), \boldsymbol{\gamma}_j \rangle = 0$ for $j \in J$, and $\langle \boldsymbol{z}-\Pi(\boldsymbol{z}| \mathcal{L}(V_J)), \boldsymbol{\gamma}_j \rangle \leq 0$, for $j \notin J$. From Lemma \ref{lem:lincomb}, there exists $J_0 \subset J$ such that $V_{J_0}$ is linearly independent and $\Pi(\boldsymbol{z}| \mathcal{L}(V_J))$ can be written as a positive linear combination of the vectors in $V_{J_0}$, which implies that $\Pi(\boldsymbol{z} | \mathcal{L}(V_J)) \in \mathcal{\overline{F}}_{J_0}$. In addition, since $\langle \boldsymbol{z}-\Pi(\boldsymbol{z}| \mathcal{L}(V_J)), \boldsymbol{\gamma}_j \rangle = 0$ for $j\in J_0$, then $\Pi(\boldsymbol{z}| \mathcal{L}(V_{J_0}))=\Pi(\boldsymbol{z}| \mathcal{L}(V_J))$. Thus, $\Pi(\boldsymbol{z}| \Omega^0)=\Pi(\boldsymbol{z}| \mathcal{L}(V_{J_0}))$. If $\mathcal{L}(V_{J_0})=\mathcal{L}(V_J)$ then $J^*=J_0$ satifies all required conditions. Now, assume that $\mathcal{L}(V_{J_0}) \subset \mathcal{L}(V_J)$. The fact that $\Pi(\boldsymbol{z}| \mathcal{L}(V_{J_0}))=\Pi(\boldsymbol{z}| \mathcal{L}(V_{J}))$ implies that $\Pi(\boldsymbol{z}| \mathcal{L}(V_{J_1}))=\Pi(\boldsymbol{z}| \mathcal{L}(V_{J_0}))$ for any set $J_1$ such that $J_0 \subseteq J_1 \subseteq J$. Further, since $\Pi(\boldsymbol{z}| \mathcal{L}(V_{J_0})) \in \mathcal{\overline{F}}_{J_0}$ then  $\Pi(\boldsymbol{z}| \mathcal{L}(V_{J_1})) \in \mathcal{\overline{F}}_{J_1}$. Thus, it is enough to choose the set $J^*$ such that $J_0 \subset J^* \subset J$ and $V_{J^*}$ is a linearly independent set that spans $\mathcal{L}(V_{J})$. This concludes the proof.   
\end{proof}

\begin{proof}[Proof of Theorem \ref{theo:badsets}.]
	Let $\boldsymbol{A}_\mu$, $\boldsymbol{A}_{\mu,J}$ and $\boldsymbol{\gamma}_{\mu_d}$ be the analogous versions of $\boldsymbol{A}_s$, $\boldsymbol{A}_{s,J}$ and $\boldsymbol{\gamma}_{s_d}$ obtained by substituting $\boldsymbol{\tilde{y}}_s$ and $\boldsymbol{W}_s$ by $\boldsymbol{\mu}$ and $\boldsymbol{W}_\mu$, respectively. Further, note that Lemma \ref{lem:irred} assures that both $\boldsymbol{A}_{s}$ and $\boldsymbol{A}_{\mu}$ are irreducible since $\boldsymbol{A}$ is.
	
	First, suppose $\emptyset \notin \mathcal{G}_\mu$ and let $J=\emptyset$. Then, from conditions in  Equation \ref{eq:KKTpolar}, $\emptyset \in \mathcal{\tilde{G}}_s$ if and only if $\langle \boldsymbol{\tilde{z}}_s, \boldsymbol{\gamma}_{s_j}\rangle \leq 0$ for $j=1,2,\dots, m$. In contrast, suppose that $\langle \boldsymbol{z}_{\mu}, \boldsymbol{\gamma}_{\mu_j}\rangle \leq 0$ for $j=1,2,\dots, m$. Hence, $\emptyset \in \mathcal{G}_{\mu}$, which contradicts our choice of $J$. Therefore, there exists $j_0$ such that $\langle \boldsymbol{z}_{\mu}, \boldsymbol{\gamma}_{\mu_{j_0}}\rangle > 0$. Then, we have
	\begin{align*}
	P\left(\emptyset \in \mathcal{\tilde{G}}_s  \right) & \leq P\left(0 \geq \langle \boldsymbol{\tilde{z}}_s, \boldsymbol{\gamma}_{s_{j_0}}\rangle \right) \\
	&= P\left( \langle \boldsymbol{z}_\mu, \boldsymbol{\gamma}_{\mu_{j_0}}\rangle - \langle \boldsymbol{\tilde{z}}_s, \boldsymbol{\gamma}_{s_{j_0}}\rangle \geq \langle \boldsymbol{z}_\mu, \boldsymbol{\gamma}_{\mu_{j_0}}\rangle \right) \\
	&\leq \frac{1}{\langle \boldsymbol{z}_\mu, \boldsymbol{\gamma}_{\mu_{j_0}}\rangle ^2 }  \mathbb{E} \left[ \left(\langle \boldsymbol{\tilde{z}}_s, \boldsymbol{\gamma}_{s_{j_0}}\rangle -  \langle \boldsymbol{z}_\mu,\boldsymbol{\gamma}_{\mu_{j_0}} \rangle \right)^2  \right] 
	\end{align*}
	where the last inequality is obtained by an application of Chebyshev's inequality. We show now that the expected value in the last term is $O(n_N^{-1})$. Note that $\langle \boldsymbol{\tilde{z}}_s, \boldsymbol{\gamma}_{s_{j_0}}\rangle$ is a function of the $N^{-1}\widehat{t}_d$ and the $N^{-1}\widehat{N}_d$. Let $f_1(N^{-1}\widehat{t}_1, \dots,N^{-1}\widehat{t}_D, N^{-1}\widehat{N}_1, \dots, N^{-1}\widehat{N}_D)$ be such a function. An application of the Mean Value Theorem to the continuous function $f_1(\cdot)$ (and to its first and second  derivative functions) along with Lemma $\ref{lem:bounded}$ (i)-(ii), lead to the conclusion that $|f_1(\cdot)|$ and its first and second derivative functions are uniformly bounded for all $N$. Moreover, $f_1(N^{-1}\widehat{t}_1, \dots, N^{-1}\widehat{t}_D,$ $N^{-1}\widehat{N}_1, \dots, N^{-1}\widehat{N}_D)$ and its first and second derivative functions, evaluated at $N^{-1}\widehat{t}_d=r_d\mu_d$ and $N^{-1}\widehat{N}_d=r_d$, are uniformly bounded for all $N$. By defining $g_1(\cdot)$ to the function $g_1(\cdot)=[f_1(\cdot)-f_1(r_1\mu_1,\dots, r_D\mu_D, r_1, \dots,r_D)]^2=[f_1(\cdot)-\langle \boldsymbol{z}_\mu, \boldsymbol{\gamma}_{\mu_{j_0}} \rangle]^2$, we can make use of Lemma \ref{lem:bounded} (iv) to fulfill the assumptions of Theorem 5.4.3 in \cite{fuller96} with $\alpha=1$, $s=2$, and $a_N=O(N^{-1/2})$. Therefore, $\mathbb{E} \left[ \left(\langle \boldsymbol{\tilde{z}}_s, \boldsymbol{\gamma}_{s_{j_0}}\rangle -  \langle \boldsymbol{z}_\mu,\boldsymbol{\gamma}_{\mu_{j_0}} \rangle \right)^2  \right]=O(n_N^{-1})$, since $g_1(\cdot)$ and its first derivative with respect to the $N^{-1}\widehat{t}_d$ and the $N^{-1}\widehat{N}_d$ evaluate to zero when $N^{-1}\widehat{t}_d=r_d\mu_d$, $N^{-1}\widehat{N}_d=r_d$.
	
	Now, take $J \neq \emptyset$ where $J \notin \mathcal{G}_\mu$. Assume that $J \in \mathcal{\tilde{G}}_s$. Theorem \ref{theo:linspace} guarantees that we can always choose a subset $J^* \subseteq J$ such that $J^* \in \mathcal{\tilde{G}}_s$, $V_{s,J^*}$ is linearly independent, and $\mathcal{L}(V_{s,J^*})=\mathcal{L}(V_{s,J})$. Note that $\Pi(\boldsymbol{\tilde{z}}_s | \mathcal{L}(V_{s,J^*}))=\boldsymbol{A}_{s,J^*}^{\top}(\boldsymbol{A}_{s,J^*}\boldsymbol{A}_{s,J^*}^{\top})^{-1}\boldsymbol{A}_{s,J^*}\boldsymbol{\tilde{z}}_s$. Let $\boldsymbol{\tilde{b}}_{s,J^*}=(\boldsymbol{A}_{s,J^*}\boldsymbol{A}_{s,J^*}^{\top})^{-1}\boldsymbol{A}_{s,J^*}\boldsymbol{\tilde{z}}_s$. Hence, from conditions in Equation \ref{eq:KKTpolar}, we have that $J \in \mathcal{\tilde{G}}_s$ implies both $\boldsymbol{\tilde{b}}_{s,J^*} \geq \boldsymbol{0}$, and $\langle \boldsymbol{\tilde{z}}_s - \boldsymbol{A}_{s,J^*}^{\top}\boldsymbol{\tilde{b}}_{s,J^*} , \boldsymbol{\gamma}_{s_j}\rangle \leq 0$ for $j =1,2,\dots, m$. Now, assume that $\boldsymbol{b}_{\mu,J^*}=(\boldsymbol{A}_{\mu,J^*}\boldsymbol{A}_{\mu,J^*}^{\top})^{-1}\boldsymbol{A}_{\mu,J^*}\boldsymbol{z}_\mu \geq \boldsymbol{0}$, and $\langle \boldsymbol{z}_\mu - \boldsymbol{A}_{\mu,J^*}^{\top}\boldsymbol{b}_{\mu,J^*} , \boldsymbol{\gamma}_{\mu_j} \rangle \leq 0$ for $j =1,2,\dots, m$. These conditions imply that $J^* \in \mathcal{G}_{\mu}$ which contradicts the original assumption that $J \notin \mathcal{G}_\mu$, since  $\mathcal{L}(V_{\mu,J^*})=\mathcal{L}(V_{\mu,J})$ by Lemma \ref{lem:linequiv}. Therefore, either there is an element of $\boldsymbol{b}_{\mu,J^*}$ that is strictly negative or there exists $j_0$ such that $\langle \boldsymbol{z}_\mu - \boldsymbol{A}_{\mu,J^*}^{\top}\boldsymbol{b}_{\mu,J^*} , \boldsymbol{\gamma}_{\mu_{j_0}}\rangle > 0$. Hence, proving that $P(J_t \in \mathcal{\tilde{G}}_s)=O(n_N^{-1})$ in any case will conclude the proof.
	 
	First, suppose the $j_0$-th element of $\boldsymbol{b}_{\mu,J^*}$ is strictly negative. That is, $\boldsymbol{e}_{j_0}^{\top}\boldsymbol{b}_{\mu,J^*} <0$, where $\boldsymbol{e}_j$ denotes the indicator vector that is $1$ for entry $j$ and $0$ otherwise. Then, we have
	\begin{align*}
	P\left( J \in \mathcal{\tilde{G}}_s  \right) &\leq P\left(\boldsymbol{e}_{j_0}^{\top}\boldsymbol{\tilde{b}}_{s,J^*} \geq 0 \right)\\
	&= P\left( \boldsymbol{e}_{j_0}^{\top}\boldsymbol{\tilde{b}}_{s,J^*} - \boldsymbol{e}_{j_0}^{\top}\boldsymbol{b}_{\mu,J^*} \geq  -\boldsymbol{e}_{j_0}^{\top}\boldsymbol{b}_{\mu,J^*} \right) \\
	&\leq \frac{1}{(\boldsymbol{e}_{j_0}^{\top}\boldsymbol{b}_{\mu,J^*})^2} \mathbb{E}\left[ (\boldsymbol{e}_{j_0}^{\top}\boldsymbol{\tilde{b}}_{s,J^*} - \boldsymbol{e}_{j_0}^{\top}\boldsymbol{b}_{\mu,J^*})^2 \right]
	\end{align*}
	where the last inequality is obtained by an application of Chebyshev's inequality. Let $f_2(N^{-1}\widehat{t}_1, \dots,$  $N^{-1}\widehat{t}_D, N^{-1}\widehat{N}_1,\dots, N^{-1}\widehat{N}_D)=\boldsymbol{e}_{j_0}^{\top}\boldsymbol{\tilde{b}}_{s,J^*}$ and $g_2(\cdot)=[f_2(\cdot)-\boldsymbol{e}_{j_0}^{\top}\boldsymbol{b}_{\mu,J^*}]^2$. An analogous argument than the one used before to the smooth functions $f_1$ and $g_1$ can be applied to the smooth functions $f_2$ and $g_2$, to conclude that the expected value of the last term of the inequality is $O(n_N^{-1})$. 
	
	Lastly, suppose that there exists $j_0$ such that $
	\kappa_{\boldsymbol{z}_\mu,j_0} = \langle \boldsymbol{z}_\mu - \boldsymbol{A}_{\mu,J_t^*}^{\top}\boldsymbol{b}_{\mu,J^*} , \boldsymbol{\gamma}_{\mu_{j_0}}\rangle >0$. Also, denote $\kappa_{\boldsymbol{\tilde{z}}_s,j_0} = \langle \boldsymbol{\tilde{z}}_s - \boldsymbol{A}_{s,J_t^*}^{\top}\boldsymbol{\tilde{b}}_{s,J^*} , \boldsymbol{\gamma}_{s_{j_0}}\rangle$. Then, we have
	\begin{align*}
	P\left( J \in \mathcal{\tilde{G}}_{s}  \right) &\leq P\left(0 \geq \kappa_{\boldsymbol{\tilde{z}}_s,j_0} \right)\\
	&= P\left( \kappa_{\boldsymbol{z}_\mu,j_0} - \kappa_{\boldsymbol{\tilde{z}}_s,j_0} \geq  \kappa_{\boldsymbol{z}_\mu,j_0} \right) \\
	&\leq \frac{1}{\kappa_{\boldsymbol{z}_\mu,j_0}^2} \mathbb{E}\left[ (\kappa_{\boldsymbol{\tilde{z}}_s,j_0} - \kappa_{\boldsymbol{z}_\mu,j_0} )^2 \right]
	\end{align*}
	where the last inequality is an application of the Chebyshev's inequality. By applying an analogous argument than before to the smooth functions $f_3(N^{-1}\widehat{t}_1, \dots, N^{-1}\widehat{t}_D, N^{-1}\widehat{N}_1, \dots, N^{-1}\widehat{N}_D)=\kappa_{\boldsymbol{\tilde{z}}_s,j_0}$ and $g_3(\cdot)=[f_3(\cdot)- \kappa_{\boldsymbol{z}_\mu,j_0}]^2$, we conclude that $\mathbb{E}\left[ (\kappa_{\boldsymbol{\tilde{z}}_s,j_0} - \kappa_{\boldsymbol{z}_\mu,j_0} )^2 \right]=O(n_N^{-1})$.
\end{proof}

\begin{proof}[Proof of Theorem \ref{theo:normaldist}.]
	Take any $J \in \mathcal{\tilde{G}}_s$ and any domain $d$. Note that the condition $\boldsymbol{A\mu}\geq \boldsymbol{0}$ implies that $\emptyset \in \mathcal{G}_\mu$. Then, we can write $\tilde{\theta}_{s_d}-\overline{y}_{U_d}$ as
	\begin{equation*}
	\tilde{\theta}_{s_d}-\overline{y}_{U_d} = (\tilde{y}_{s_d}-\overline{y}_{U_d})1_{J=\emptyset} + \sum_{J_G \in \mathcal{G}_\mu \setminus \emptyset} (\tilde{\theta}_{s_d,J_G}-\overline{y}_{U_d})1_{J_G=J} + \sum_{J_G \in \mathcal{G}_\mu^c} (\tilde{\theta}_{s_d,J_G}-\overline{y}_{U_d})1_{J_G=J}, 
	\end{equation*}
	where we used that $\tilde{\theta}_{s_d,\emptyset}=\tilde{y}_{s_d}$. Now, note that the unfeasible variance estimator $AV(\tilde{\theta}_{s_d,J})$ can be written as
	\begin{equation*}
	AV(\tilde{\theta}_{s_d,J})=AV(\tilde{y}_{s_d})1_{J=\emptyset} + \sum_{J_G \in \mathcal{G}_\mu \setminus \emptyset} AV(\tilde{\theta}_{s_d,J_G})1_{J = J_G} +\sum_{J_G \in \mathcal{G}_\mu^c } AV(\tilde{\theta}_{s_d,J_G})1_{J = J_G}.
	\end{equation*}
	Hence, 
	\begin{align*}
	& AV(\tilde{\theta}_{s_d,J})^{-1/2}(\tilde{\theta}_{s_d}-\overline{y}_{U_d}) = AV(\tilde{y}_{s_d})^{-1/2}(\tilde{y}_s-\overline{y}_{U_d})1_{J=\emptyset} \\
	&+ \sum_{J_G \in \mathcal{G}_\mu \setminus \emptyset}AV(\tilde{\theta}_{s_d,J_G})^{-1/2}(\tilde{\theta}_{s_d,J_G} - \overline{y}_{U_d})1_{J=J_G} 
	+ \sum_{J_G \in \mathcal{G}_\mu^c} AV(\tilde{\theta}_{s_d,J_G})^{-1/2}(\tilde{\theta}_{s_d,J_G} - \overline{y}_{U_d})1_{J=J_G} \\
	&= \left[ AV(\tilde{y}_{s_d})^{-1/2}(\tilde{y}_s-\overline{y}_{U_d})1_{J=\emptyset} 
	+ \sum_{J_G \in \mathcal{G}_\mu \setminus \emptyset}AV(\tilde{\theta}_{s_d,J_G})^{-1/2}(\tilde{\theta}_{s_d,J_G} - \theta_{U_d,J_G})1_{J=J_G} \right. \\
	&+ \left. \sum_{J_G \in \mathcal{G}_\mu^c} AV(\tilde{\theta}_{s_d,J_G})^{-1/2}(\tilde{\theta}_{s_d,J_G} - \theta_{U_d,J_G})1_{J=J_G} \right]	
	+ \left[ \sum_{J_G \in \mathcal{G}_\mu \setminus \emptyset}AV(\tilde{\theta}_{s_d,J_G})^{-1/2}(\theta_{U_d,J_G} - \overline{y}_{U_d})1_{J=J_G} \right] \\
	&+ \left[ \sum_{J_G \in \mathcal{G}_\mu^c} AV(\tilde{\theta}_{s_d,J_G})^{-1/2}(\theta_{U_d,J_G} - \overline{y}_{U_d})1_{J=J_G} \right] \\
	&=c_{1N}+c_{2N}+c_{3N},
	\end{align*} 
	where $\theta_{U_d,J_G}$ is the population version of $\tilde{\theta}_{s_d,J_G}$. Note that each term of the form $AV(\tilde{\theta}_{s_d,J_G})^{-1/2}(\tilde{\theta}_{s_d,J_G} - \theta_{U_d,J_G} )$ converges in distribution to a standard normal distribution by Assumption A6. Thus, $c_{1N}$ converges in distribution to a standard normal distribution. Now, note that for each $J_G \in \mathcal{G}_\mu^c$, then
	\begin{equation*}
	 AV(\tilde{\theta}_{s_d,J_G})^{-1/2}(\theta_{U_d,J_G} - \overline{y}_{U_d})= [n_N AV(\tilde{\theta}_{s_d,J_G})]^{-1/2}[n_N^{1/2}(\theta_{U_d,J_G} - \overline{y}_{U_d})] = O(n_N^{1/2}).
	\end{equation*}
	In contrast, for $J_G \in \mathcal{G}_{\mu}^{c}$, we have that $1_{J=J_G}=O_p(n_N^{-1})$ by Theorem \ref{theo:badsets} (since $J \in \mathcal{\tilde{G}}_s$). Thus, $c_{3N}=O_p(n_N^{-1/2})$. Now, note that $\theta_{U_d,J_G} - \overline{y}_{U_d}=O(N^{-1/2})$ when $J_G \in \mathcal{G}_\mu \setminus \emptyset$ by Assumption A3. Hence, for any $J_G \in \mathcal{G}_\mu \setminus \emptyset$,
	\begin{equation*}
	AV(\tilde{\theta}_{s_d,J_G})^{-1/2}(\theta_{U_d,J_G} - \overline{y}_{U_d})= [n_N AV(\tilde{\theta}_{s_d,J_G})]^{-1/2}[n_N^{1/2}(\theta_{U_d,J_G} - \overline{y}_{U_d})] = O\left(\sqrt{\frac{n_N}{N}} \right),
	\end{equation*}
	which implies that $c_{2N}=O\left(\sqrt{\frac{n_N}{N}} \right)$ (bias term). Thus, by combining these properties of $c_{1N}$, $c_{2N}$ and $c_{3N}$, we conclude that
	\begin{equation*}
	AV(\tilde{\theta}_{s_d,J})^{-1/2}(\tilde{\theta}_{s_d} - \overline{y}_{U_d}) \overset{\mathcal{L}}{\rightarrow} \mathcal{N}(B, 1),
	\end{equation*}
	where $B=O(\sqrt{\frac{n_N}{N}})$.
	
	Now, write the feasible variance estimator $\widehat{V}(\tilde{\theta}_{s_d,J})$ as 
	\begin{equation*}
	\widehat{V}(\tilde{\theta}_{s_d,J})=\widehat{V}(\tilde{y}_{s_d})1_{J=\emptyset} + \sum_{J_G \in \mathcal{G}_\mu \setminus \emptyset} \widehat{V}(\tilde{\theta}_{s_d,J_G})1_{J = J_G} +\sum_{J_G \in \mathcal{G}_\mu^c } \widehat{V}(\tilde{\theta}_{s_d,J_G})1_{J = J_G}.
	\end{equation*}
	By Assumption A6, we have that $\widehat{V}(\tilde{\theta}_{s_d,J_G}) - AV(\tilde{\theta}_{s_d,J_G})=O(n_N^{-1})$ for any $J_G$. Latter implies that $\widehat{V}(\tilde{\theta}_{s_d,J})^{1/2} - AV(\tilde{\theta}_{s_d,J})^{1/2}=O(n_N^{-1/2})$. Hence, an application of Slutsky's theorem allows to replace $AV(\tilde{\theta}_{s_d,J})^{-1/2}$ by $\widehat{V}(\tilde{\theta}_{s_d,J})^{-1/2}$. 
	
	To prove the last part of this theorem, just note that $\boldsymbol{A\mu}>\boldsymbol{0}$ implies $\mathcal{G}_\mu=\{\emptyset\}$. Thus, the term $c_{2N}$ does not exist, so the bias term vanishes.
\end{proof}

\end{document}